\documentclass[12pt]{article}
\usepackage{amssymb, amsmath, amsthm}
\usepackage{lscape} 
\usepackage{mathdots}
\usepackage{thmtools}
\usepackage{mathtools}
\declaretheorem[style=definition]{example}
\usepackage{appendix}
\usepackage{bm}
\usepackage{bbm}
\usepackage[makeroom]{cancel}
\usepackage{caption}
\usepackage{enumerate} 
\usepackage{graphicx}
\usepackage{natbib}
\usepackage{setspace}
\usepackage{hyperref} 
\usepackage{color}
\usepackage{tikz}
\usepackage{mathrsfs}
\usetikzlibrary{snakes}
\usepackage{subcaption}
\usepackage[top=1.25in, bottom=1.25in, left=1.25in, right=1.25in]{geometry}
\DeclareMathOperator*{\argmax}{arg\,max}

\DeclareMathOperator{\sgn}{sgn}
\newtheorem{assm}{Assumption}
\newtheorem{lem}{Lemma}
\newtheorem{prop}{Proposition}
\newtheorem{thm}{Theorem}

\newtheorem{defn}{Definition}

\newtheorem{obs}{Observation}

\newtheorem{assmO}{Assumption}

\newtheorem{lemO}{Lemma}

\newtheorem{propO}{Proposition}

\newtheorem{defnO}{Definition}

\newtheorem{obsO}{Observation}

\newtheorem{thmO}{Theorem}

\begin{document}
\renewcommand{\thefootnote}{\fnsymbol{footnote}}
\renewcommand\thmcontinues[1]{Continued}
\title{The Politics of Personalized News Aggregation}
\author{
Lin Hu\footnote{Research School of Finance, Actuarial Studies and Statistics, Australian National University, lin.hu@anu.edu.au. }
\and Anqi Li\footnote{Department of Economics, Virginia Polytechnic and State University, angellianqi@gmail.com.}
\and Ilya Segal\footnote{Department of Economics, Stanford University, isegal@stanford.edu. 
}}
\date{This Draft: November 2022}
\maketitle
 
 \thispagestyle{empty}

\begin{abstract}
We study how personalized news aggregation for rationally inattentive voters (NARI) affects policy polarization  and public opinion. In a two-candidate electoral competition model, an attention-maximizing infomediary aggregates source data about candidates' valence into easy-to-digest news. Voters decide whether to consume news, trading off the expected gain from improved expressive voting against the attention cost. NARI generates policy polarization even if candidates are office-motivated. Personalized news aggregation makes extreme voters the disciplining entity of policy polarization, and the skewness of their signals is crucial for sustaining a high degree of policy polarization in equilibrium. Analysis of disciplining voters yields insights into the equilibrium and welfare consequences of regulating infomediaries. 

\bigskip


\bigskip


\bigskip

\bigskip

\bigskip

\bigskip

\bigskip

\bigskip

\bigskip

\bigskip

\end{abstract} 
\renewcommand{\thefootnote}{\arabic{footnote}}
\pagebreak

\clearpage
\pagenumbering{arabic} 

\section{Introduction}\label{sec_introduction}
Recently, the idea that tech-enabled news  personalization could affect polarization has been put forward in the academia and popular press \citep{sunstein, filterbubble, gentzkow}. This paper studies how personalized news aggregation for rationally inattentive voters affects the polarization of policies and public opinion in an electoral competition model. 

Our premise is that rational demand for news aggregation in the digital era is driven by information processing costs. As the Internet and social media become important sources of information, hosting more data (2.5 quintillion bytes) than what any individual can process in a lifetime, consumers must turn to infomediaries for content aggregation, customized based on their personal data such as  demographic and psychographic attributes, digital footprints, and social network positions.\footnote{In computing, an \emph{aggregator} is a client software or a web application that aggregates syndicated web content such as online newspapers, blogs, podcasts, and vlogs in one location for easy viewing. Prominent examples of aggregators include aggregator sites, social media feeds, and mobile news apps. They have recently gained prominence as more people get news online, from social media, and through mobile devices \citep{matsa}. The top three popular news websites in 2019: Yahoo! News, Google News, and Huffington Post, are all aggregators. 

Companies that run aggregators are called  \emph{infomediaries}. They operate by sifting through a myriad of online sources and displaying \emph{snippets} (headline+excerpt) on their platforms.  Snippets contain coarse information and do not always generate click-throughs of the original content   \citep{dellarocasetal}. A major revenue source for infomediaries comes from displaying ads to users while the latter are scrolling down the snippets. See  \cite{atheyetal} for background reviews.\label{footnote_background}} In this paper, we abstract from the issue of data 
generation (e.g., original reporting), focusing instead on the role of infomediaries in aggregating source data into news that is easy to process and useful for the target audience.

We develop a model of news aggregation for rationally  inattentive consumers (NARI), in which an infomediary can flexibly aggregate source data into news using algorithm-driven systems. 
While flexibility is also assumed in the Rational Inattention (RI) model pioneered by \cite{sims} and \cite{sims1}, there decision-makers can aggregate
information optimally themselves and so have no need for 
external aggregators. To model the demand for infomediaries, we
assume that consumers can only choose whether to absorb the information offered to them but cannot digest information partially or selectively, let alone aggregate information optimally themselves.  While this assumption is certainly stylized, it is the simplest one that creates a role for infomediaries while capturing important facets of reality.\footnote{As pointed out by \cite{strombergsurvey}, the last assumption is implicitly made by the media literature, because without it the role of information provider would be much more limited. It isn't at odds with reality, since analyses of page activities (e.g., scrolling, viewport time) have established significant user attention in the reading of (snippets of) online news \citep{dellarocasetal, lagun}.}

If choosing to consume news, a consumer incurs an attention cost that is \emph{posterior separable} \citep{caplindean13} while deriving utilities from improved decision-making. Consuming news is optimal if the expected utility gain exceeds the attention cost. As for the infomediary, we assume that its goal is to maximize the total amount of attention paid by consumers, interpreted as the advertising revenue generated from consumer eyeballs. This stylized
assumption captures the key trade-off faced by the infomediary, who uses useful and easy-to-process news to attract consumers' attention while preventing them from tuning out. While we focus
on the case of a monopolistic infomediary in order to capture the market power wielded by tech giants, we also investigate an extension to perfectly competitive infomediaries, which, together with personalization, becomes equivalent to
consumers optimally aggregating information  themselves as in the standard RI model. 
 
We embed the NARI model into an electoral competition game in which two 
office-motivated candidates choose policies on a left-right spectrum.
Voters vote expressively based on policies, as well as an uncertain valence state about which candidate is more fit for office. News about candidate valence is provided by an infomediary, which moves simultaneously with the candidates. We study how NARI affects the polarization of equilibrium policies and voter opinions in this game. 

A consequence of NARI is that signal realizations prescribe \emph{recommendations} as to which candidate one should vote for. Indeed, any information beyond voting recommendations would only raise the attention cost without any corresponding benefit to voters and would thus turn away voters whose participation constraints bind at the optimum. Furthermore, voters must strictly prefer to obey the recommendations given to them, a property we refer to as \emph{strict obedience}. Indeed, if a voter has a (weakly) preferred candidate that is independent of his voting recommendations, then he could always vote for that candidate without paying attention in order to save on the attention cost. 

An important implication of strict obedience is that local deviations from a policy profile wouldn't change voters' voting decisions regardless of the recommendations they receive, suggesting that a positive degree of policy polarization could arise in equilibrium even if candidates are office-motivated. We define \emph{policy polarization} as the maximal distance between candidates' positions among all symmetric perfect Bayesian equilibria. In the baseline model featuring left-leaning, centrist,  and right-leaning voters, our main theorem shows that policy polarization is strictly positive and equals the \emph{disciplining voter's policy latitude}.

A voter's \emph{policy latitude} is an index that captures his resistance to candidates' policy deviations. It decreases with the voter's preference for the deviating candidate's policies and  increases with his pessimism about the latter's valence following unfavorable news. A voter is said to be \emph{disciplining} if his policy latitude determines policy polarization. To illustrate how personalized news aggregation affects the disciplining voter,  we compare two cases: (i) \emph{broadcast news aggregation}, in which the infomediary must offer a single signal to all voters, and (ii) \emph{personalized news aggregation}, in which the infomediary can design different signals for different voters. In the broadcast case, all voters receive the
same voting recommendation, so a candidate's deviation is profitable, i.e., strictly increases his winning probability, if and only if it attracts a majority coalition. Under the usual assumptions, this is equivalent to attracting centrist voters, who are therefore disciplining.  In the personalized case, the infomediary can provide conditionally independent signals to different voters, so each type of voter is pivotal with a positive probability when voters' population distribution is sufficiently dispersed. In that case, a policy deviation is shown to be profitable if and only if it attracts any type of voter, and voters with the smallest policy latitude are disciplining because they are the easiest to attract. 

The skewness of extreme voters' personalized signals is crucial for  sustaining a greater degree of policy polarization as news aggregation becomes personalized. To maximize the usefulness of news  consumption for an extreme voter, the recommendation to 
vote across the party line must be very strong, and, in order to prevent the voter from tuning out, must also be very rare (hereinafter, an \emph{occasional big surprise}). Most of the time, the recommendation is to vote along the party line (hereinafter, a \emph{predisposition reinforcement}), which together with the occasional big surprise has been documented in the empirical literature.\footnote{Recently, \cite{flaxman} find that the use of news aggregators not only reinforces people's predispositions but also  strengthens their opinion intensities when supporting  opposite-party candidates (i.e., occasional big surprise). Evidence for predisposition reinforcement is discussed in  \cite{fiorina} and \cite{gentzkow}. Evidence for occasional big surprise and, more generally, Bayesian voters is surveyed by \cite{dellavignagentzkow}. \label{footnote_bayesian}} When base voters are disciplining, the occasional big surprise of their signal makes them difficult to attract in the rare event where news is unfavorable to their own-party candidate. When base voters are so pessimistic about their own-party candidate's valence that even the most attractive deviation to them is still not attractive enough, opposition voters become disciplining, despite that they, too, are difficult to attract due to their preferences against the deviating candidate's policies.   If, in the end, all voters end up having bigger policy latitudes than the centrist voters in the broadcast case, then the personalization of news aggregation increases policy polarization. The last condition holds when the attention cost parameter is large and extreme voters have strong policy preferences under specific attention cost functions.

Analyses of the disciplining voter yield structural insights into the polarization effect of recent regulatory proposals to tame  tech giants. In addition to the personalization of news aggregation---the reversal of which is a plausible consequence of limiting tech companies' access to users' personal data \citep{gdpr, warren}---we study the consequences of introducing perfect competition to infomediaries. This regulatory proposal is advocated by the British government as a preferable way of regulating tech giants \citep{digitalcompetition}, and it is  mathematically equivalent to increasing voters' attention cost parameter in the monopolistic personalized case. Its policy polarization effect is negative,  because increasing the attention cost parameter tempers voters' beliefs about candidate valence and therefore reduces their policy latitudes. 

Our analysis suggests that factors carrying negative connotations in everyday discourse could have unintended consequences for policy polarization. An example is \emph{increasing mass polarization}, which we model as a mean-preserving spread of voters' policy preferences \citep{fiorina, gentzkow}. Under personalized news aggregation, increasing mass polarization can surprisingly reduce policy polarization rather than increasing it: As we keep redistributing voters' population from the center to the margin, policy polarization will eventually decrease from the centrist voters' policy latitude to the minimal policy latitude among all voters. 

In Online Appendix \ref{sec_general}, we extend the baseline model to encompass general voters and arbitrary correlation structures between their personalized signals. We develop a methodology for analyzing this general model. Among other things, we find  that correlation can only increase policy polarization, and that policy polarization is minimized when signals are conditionally independent across voters and voters' population distribution is uniform across types. Thus our baseline result prescribes the exact lower bound for the policy polarization effect of personalized news aggregation, and factors that preserve this lower bound (e.g., enrich voters' types, divide voters of the same type into multiple subgroups) wouldn't render  policy polarization trivial.

In what follows, we introduce the baseline model in Section \ref{sec_model}, conduct equilibrium analysis in Sections \ref{sec_analysis}, and report extensions of the baseline model in Section \ref{sec_extension}. We discuss the related literature throughout, but mainly in Section \ref{sec_literature}, followed by concluding remarks and discussions of future research in Section  \ref{sec_conclusion}. Additional materials and mathematical proofs can be found in the appendices. 
 
\section{Baseline model}\label{sec_model}
In this section, we first describe the model setup and then discuss the main assumptions. 

\subsection{Setup}\label{sec_setup}
Two \emph{office-motivated} candidates named $L$ and $R$ can adopt the policies on the real line. They face a unit mass of infinitesimal voters who are either \emph{left-leaning} ($k=-1$), \emph{centrist} ($k=0$), or \emph{right-leaning} ($k=1$). Each type $k \in \mathcal{K}=\left\{-1,0,1\right\}$ of voter has a  population $q\left(k\right)>0$ and values a policy $a \in \mathbb{R}$ by $u\left(a,k\right)=-|t\left(k\right)-a|$. The environment is symmetric, in that $q(1)=q(-1)$ and $t(1)>t(0)=0>t(-1)=-t(1)$.  Thus  a centrist voter is also a median voter. 

At the end of the day, the society holds an election, in which the majority winner wins the election, and ties are broken evenly between the two candidates. During the election, each voter must vote \emph{expressively} for one of the candidates. For any given profile ${\bf{a}}= (a_L, a_R) \in \mathbb{R}^2$ of policy positions, a type $k$ voter earns the following utility difference from voting for candidate $R$ rather than $L$: 
\[v\left({\bf{a}}, k\right)+\omega. \]
In the above expression, 
\[v\left({\bf{a}}, k\right)= u\left(a_{R}, k\right)-u\left(a_{L}, k\right)\] is the voter's differential valuation of candidates' policies, whereas $\omega$ is an uncertain \emph{valence state} about which candidate is more fit for office given the circumstances.\footnote{E.g., in the ongoing debate about how to battle terrorism, $\omega=-1$ if the state favors the use of soft power (e.g., diplomatic tactics), and $\omega=1$ if the state favors the use of hard power (e.g., military preemption). Candidates $L$ and $R$ are experienced with using soft and hard power,  respectively, and whoever is more experienced with handling the circumstances has an advantage over his opponent. \label{fn_terrorism}}  In the baseline model, $\omega$ takes the values in $\Omega=\left\{-1,1\right\}$ with equal probability, so its prior mean equals zero.  

When casting votes, voters observe candidates' policies but not directly the realization of the valence state. News about the latter is modeled as a \emph{finite signal structure} (or simply a \emph{signal}) $\Pi:\Omega \rightarrow \Delta\left(\mathcal{Z}\right)$, where $\mathcal{Z}$ is a finite set of signal realizations with $|\mathcal{Z}| \geq 2$, and each $\Pi\left(\cdot \mid \omega\right)$ specifies a probability distribution over $\mathcal{Z}$ conditional on the state being $\omega \in \Omega$.  News is provided by a \emph{monopolistic} infomediary who is  equipped with a \emph{segmentation technology} $\mathcal{S}$. $\mathcal{S}$ is a partition of voters' types, and each cell of it is called a \emph{market segment}. The infomediary can distinguish between voters belonging to different market segments but not those within the same market segment. Our focus is on the coarsest and finest partitions named the \emph{broadcast technology}  $b=\left\{\mathcal{K}\right\}$ and \emph{personalized technology} $p=\left\{\left\{k\right\}: k \in \mathcal{K}\right\}$, respectively: The former cannot distinguish between the various types of the voters at all, whereas the latter can do so perfectly. 

Under segmentation technology $\mathcal{S} \in \left\{b, p\right\}$, the infomediary designs $|\mathcal{S}|$ signals, one for each market segment. Within each market segment, voters decide whether to consume the signal that is offered to them.   Consuming a signal $\Pi$ means fully absorbing its information content. Doing so incurs an \emph{attention cost} $\lambda  I\left(\Pi\right)$, where $\lambda>0$ is the \emph{attention cost parameter}, and $I\left(\Pi\right)$ is the needed \emph{amount of attention} for absorbing the information content of $\Pi$. After that, voters observe signal realizations,  update their beliefs about the valence state, and cast votes. The infomediary's profit equals the total amount of attention paid by voters.

The game sequence is summarized as follows. 
\begin{enumerate}
\item The infomediary designs signal structures; voters observe the signals structures offered to them and make consumption decisions. 
\item Candidates choose policies without observing the moves in Stage 1. 
\item The state is realized.
\item Voters observe policies and signal realizations before casting votes.
\end{enumerate}
 
Our solution concept is \emph{pure strategy perfect Bayesian equilibrium} (PSPBE), or \emph{equilibrium} for short. Our goal is to characterize all \emph{symmetric} PSPBEs, where the policy profiles proposed by the candidates take the form of $(-a,a)$ with $a \geq 0$.

\subsection{Model discussion}\label{sec_modeldiscussion}
\paragraph{Attention cost} To state our assumptions about the attention cost function, recall that a signal structure $\Pi:\Omega\rightarrow \Delta\left(\mathcal{Z}\right)$ specifies how source data about the state are (randomly) aggregated into the content indexed by the signal realizations in $\mathcal{Z}$. For each $z \in \mathcal{Z}$, let 
\[\pi_z=\sum_{\omega \in \Omega} \Pi\left(z \mid \omega\right)/2\]
denote the probability that the signal realization is $z$, and assume without loss of generality (w.l.o.g.) that $\pi_z>0$. Then \[\mu_z=\sum_{\omega \in \Omega}\omega \Pi\left(z \mid \omega\right)/\left(2\pi_z\right)\]
is the posterior mean of the state conditional on the signal realization being  $z$, and it fully captures one's posterior belief after observing $z$. 

\begin{assm}\label{assm_attention}
The needed amount of attention for consuming $\Pi: \Omega  \rightarrow \Delta\left(\mathcal{Z}\right)$ is
\begin{equation}\label{eqn_attentionfunction}
I \left(\Pi\right)= \sum_{z\in \mathcal{Z}} \pi_z  h\left(\mu_z\right),
\end{equation}
where $h: \left[-1,1
\right]\rightarrow \mathbb{R}_+$ (i) is strictly convex and satisfies $h\left(0\right)=0$; (ii) is continuous on $\left[-1,1\right]$ and twice differentiable on $\left(-1,1\right)$; and (iii) is symmetric around zero. 
 \end{assm}
 
Equation (\ref{eqn_attentionfunction}) coupled with Assumption \ref{assm_attention}(i) is equivalent to \emph{weak posterior separability} (WPS)---a notion proposed by \cite{caplindean13} to generalize Shannon's entropy as a measure of attention cost. For a review of the theoretical and empirical foundations for WPS, see Section \ref{sec_literature}. In the current setting, WPS  stipulates that consuming a null signal requires no attention, and that more attention is needed for moving one's posterior belief closer to the true state and as the signal becomes more Blackwell-informative. The high-level idea is that attention is a scarce resource that reduces one's uncertainty about the underlying state.  

Parts (ii) and (iii) of Assumption \ref{assm_attention} are made for technical reasons.  While they are by no means innocuous, they are commonly seen in the (applied) RI literature.  Examples that satisfy all three assumptions include (i) mutual information ($h(\mu)=\text{binary entropy function}\left(\left(1+\mu\right)/2\right)$),  which measures the reduction in the  Shannon entropy of the state before and after news  consumption; as well as (ii) variance reduction, for which $h(\mu)=\mu^2$. 
 
\paragraph{Model assumptions}  Our model assumptions are centered around four noteworthy facts about infomediaries' business model  (see Footnote \ref{footnote_background} for full details).

\begin{enumerate}
\item The content provided by infomediaries is usually very coarse, taking the form of snippets that consist of a title and a few summary sentences. 
\item A major source of infomediary's revenue comes from displaying ads to users while the latter are browsing through snippets. In the case of Facebook News Feed, every few snippets are followed by an ad, followed by a few more snippets, etc., so that users can absorb ads and content together seamlessly. 

\item According to pundits working and consulting in the tech sector, modern news aggregators are operated by tech giants that wield significant market power. See also \cite{fanta2018publisher} for a report on how big players like Google News have been  reshaping the news landscape. 

\item The algorithms behind their operations represent trade secrets that cannot be easily reverse-engineered by third parties.  According to computer scientists working on human-centered computing,  a common way to recover these algorithms is to survey users,  who have proven effective in detecting unusual changes in their algorithms in recent cases \citep{devos2022toward}.
\end{enumerate}

In light of  Facts 1-3, we assume that a monopolistic infomediary maximizes the total amount of attention paid by voters,\footnote{One can think of a snippet as a small piece of information encountered by a decision-maker before he stops news consumption.  \cite{shannon1948mathematical},  \cite{hebertwoodford}, and \cite{morrisstrack} provide general conditions under which the expected number of consumed snippets is the posterior-separable attention cost that is needed for implementing a signal structure. } and note that our analysis remains unaffected as long as the infomediary's profit is a strictly increasing function of voters' attention.\footnote{To see why, suppose the profit generated by a voter consuming $\Pi$ equals $J\left(I\left(\Pi\right)\right)$ for some $J'>0$. For  any given set of voters whose participation constraints (i.e., to consume rather than to abstain) we wish to satisfy,  the infomediary solves $\max_{\Pi}J\left(I\left(\Pi\right)\right)\cdot$ market share,  or  equivalently  $\max_{\Pi}I\left(\Pi\right)$,  subject to voters' participation constraints. In the case where $J(\Pi)=\alpha I(\Pi)+\beta$ for some $\alpha, \beta>0$, $\beta$ captures the part of the revenue that is generated solely by the market share. } Online Appendix \ref{sec_competitive} examines the case of perfectly competitive infomediaries. 

Regarding the game sequence, we assume, based on Fact 4,  that candidates cannot condition policies on signal structures,\footnote{Allowing candidates to condition policies on the valence state (i.e., interchange stages 2 and 3 of the game) wouldn't affect the PSPBEs of the game. In one direction, one can show, through replicating our proofs step by step, that any PSPBE of the current game remains a PSPBE even if candidates can condition policies on the valence state. The the opposite direction is also clear: In any PSPBE of the augmented game,  a candidate must adopt the same policy in both states, because otherwise he will reveal the true state and lose the election for sure when the state is unfavorable to him.} whereas voters observe the signal structures they choose to consume.  Policies are made observable to  voters at the voting stage---a process that we do not explicitly model (e.g., political advertising, canvassing), but note that it typically takes place right before the election day due legal or practical barriers \citep{gerber}. Given this, as well as the long election cycles in many countries,  it is safe to conclude that in practice, the design and consumption of signals that aggregate evolving 
states of the world into everyday breaking news (e.g., whether a looming terrorism threat can be best countered by the use of soft power or hard power) are often made without observing the final policies.

\section{Analysis}\label{sec_analysis}
This section conducts equilibrium analysis. Specifically, we provide equilibrium characterizations in Sections \ref{sec_news} and \ref{sec_policy}, and investigate comparative statics in Section \ref{sec_cs}.

\subsection{Optimal signals}\label{sec_news}
In this section, we fix \emph{any} symmetric policy profile $\mathbf{a}=(-a,a)$ with $a \geq 0$ and solve for the signals that maximize the infomediary's profit (hereinafter, \emph{optimal signals}). To facilitate discussion, we say that candidate $L$ (resp. $R$) is the \emph{own-party candidate} of left-leaning (resp. right-leaning) voters. 

\paragraph{Infomediary's problem} We first formalize the infomediary's problem. Under segmentation technology $\mathcal{S} \in \left\{b,p\right\}$, any optimal signal for market segment $s \in \mathcal{S}$ solves \[\tag{$s$} \max_{\Pi} I\left(\Pi\right) \mathcal{D}\left(\Pi; {\bf{a}}, s\right) \label{eqn_problems}\]
where $\mathcal{D}\left(\Pi; {\bf{a}}, s\right)$ denotes the demand for signal $\Pi$ in market segment $s$ under policy profile $\bf{a}$. To figure out $\mathcal{D}(\cdot)$, note that since a voter could always 
vote for his own-party candidate without consuming news, news consumption is only useful if it sometimes convinces him to vote across the party line.  After consuming $\Pi$, a voter strictly prefers candidate $R$ to $L$ if $v\left({\bf{a}}, k\right)+\mu_z>0$, and he strictly prefers candidate $L$ to $R$ if $v\left({\bf{a}}, k\right)+\mu_z<0$. Ex ante, the expected utility gain from consuming $\Pi$ is 
\[V\left(\Pi; {\bf{a}}, k\right)=\begin{cases}
\sum_{z \in \mathcal{Z}} \pi_z \left[v\left({\bf{a}},k \right)+\mu_z\right]^{+} & \text{ if } k \leq 0,\\
\sum_{z \in \mathcal{Z}} -\pi_z \left[v\left({\bf{a}},k\right)+\mu_z\right]^{-}& \text{ if } k > 0,
\end{cases}\]
and consuming $\Pi$ is preferable to abstaining (hereinafter, the voter's \emph{participation constraint is satisfied}) if 
\[V\left(\Pi; {\bf{a}}, k\right) \geq \lambda  I(\Pi). \]
Therefore, 
\[\mathcal{D}\left(\Pi; {\bf{a}}, s\right)=\sum_{k \in \mathcal{K}: V\left(\Pi; {\bf{a}}, k\right) \geq \lambda   I\left(\Pi\right)} \text{ population of type $k$ voters in market segment $s$.}\]

After stating the infomediary's problem, we next impose more structures on it to make it amenable to analysis. 
\begin{assm}\label{assm_regularity}
For any policy profile above, (i) \emph{(feasibility)} there exists a signal that strictly satisfies all voters' participation constraints; (ii) the signal that fully reveals the state violates all players' participation constraints; and (iii) it is strictly optimal to include all voters in news consumption in the broadcast case.
\end{assm}

Assumption \ref{assm_regularity} has three parts. Part (i) of the assumption, also referred to as the \emph{feasibility  condition}, says that voters strictly benefit from consuming some nondegenerate signal.\footnote{The feasibility condition helps establish that the infomediary's problems satisfy strong duality and can therefore be solved by the Lagrangian method.   The Lagrangian dual approach has been increasingly applied to the study of information design problems (see, e.g., \citealt{salamanca}). Our proof exploits the unique  structure of the infomediary's problems. }  Part (ii) of the assumption rules out the uninteresting case where the infomediary simply reveals the true state to voters, whereas Part (iii) of it ensures that the optimal broadcast signal differs from any optimal personalized signal. Intuitively, these conditions should hold simultaneously when the attention cost parameter is moderate and voters' policy preferences aren't too extreme, so that it is optimal to include all voters in nontrivial news consumption, although revealing the true state to them would tune them out. 
For the case of quadratic attention cost, we can verify this intuition directly by reducing the assumption to $\lambda>1/2$ and $\lambda t(1)<3\sqrt{2}/4-1\approx .060$ (see Appendix \ref{sec_proof_news} for derivations).  For entropy attention cost, we solve the model numerically in Appendix \ref{sec_numerical} and find similar patterns.

\paragraph{Binary signal and strict obedience} We provide two characterizations of optimal signals. Our first result shows that in both the broadcast case and personalized case, the optimal signal is unique and prescribes binary voting recommendations that its consumers strictly prefers to obey.  To facilitate analysis, we say that a signal realization $z$ \emph{endorses} candidate $R$ and \emph{disapproves of} candidate $L$ if $\mu_z>0$, and that it endorses candidate $L$ and disapproves of candidate $R$ if $\mu_z<0$.  For binary signals, we write $\mathcal{Z}=\left\{L, R\right\}$. From \emph{Bayes' plausibility}, which mandates that the expected posterior mean must equal the prior mean zero:
\begin{equation}\label{eqn_br}
\tag{BP}\sum_{z \in \mathcal{Z}} \pi_z  \mu_z=0, 
\end{equation}
it follows that we can assume, w.l.o.g., that $\mu_L<0<\mu_R$. In this way, we can interpret each signal realization $z \in \left\{L, R\right\}$ as an endorsement for candidate $z$ and a disapproval of candidate $-z$.  In addition, we can define the notion of  \emph{strict obedience} as follows. 
\begin{defn}
A binary signal induces \emph{strict obedience from its consumers} if the latter strictly prefer the endorsed candidate to the disapproved one  under both signal realizations, i.e.,
\begin{equation}\label{eqn_ob}
\tag{SOB} v\left({\bf{a}},k\right)+\mu_L<0<v\left({\bf{a}},k\right)+\mu_R. 
\end{equation}
\end{defn}

The next theorem formalizes the statement made at the beginning of this section.
 
\begin{thm}\label{thm_binary}
Under Assumptions \ref{assm_attention} and  \ref{assm_regularity}(i),  the following hold for any policy profile $(-a,a)$ with $a \geq 0$.
\begin{enumerate}[(i)]
\item The optimal personalized signal for any voter is unique and binary.
\item The optimal broadcast signal is unique and binary.
\item Any optimal signal,  broadcast or personalized, induces strict obedience from its consumers.
\end{enumerate} 
 \end{thm}

The intuition behind the personalized case is easy to understand. For that case, our analysis exploits the binary nature of individual voters' decision problems, as well as the Blackwell-monotonicity of the attention cost function. Under these assumptions, any information beyond decision  recommendations would only raise the attention cost without any corresponding benefit to voters, and would thus turn away voters whose participation constraints bind at the optimum. For these voters, maximizing attention is equivalent to maximizing the usefulness of news  consumption at the maximal attention level. 

The broadcast case is more delicate, as it requires that we aggregate voters with binding participation constraints into a representative voter.  Under the assumption that voters' policy preferences exhibit increasing differences between policies and types, only extreme voters' participation constraints can bind, whereas centrist voters' participation constraint must be slack. The resulting representative voter makes at most three decisions: LL, LR, and RR (the first and second letters stand for the voting decisions of the left-leaning voter and right-leaning voter, respectively), so the optimal signal for him has at most three signal realizations. Then using the concavification method developed by \cite{aumann} and \cite{bayesianpersuasion},  we reduce the number of  signal realizations to two. The analysis exploits the assumption of binary states, as well as the posterior separability of the attention cost function.  

Strict obedience (\ref{eqn_ob}) is an essential feature of optimal binary signals. Intuitively, if a consumer of a binary signal has a (weakly) preferred candidate that is independent of his voting recommendations, then he would prefer to vote for that candidate unconditionally without consuming the signal, because doing so saves on the attention cost without affecting the expected voting utility. But this contradicts the assumption that the voter prefers to consume the signal rather than to abstain. 

\paragraph{Skewness} We next examine the skewness of optimal signals. Since the underlying state is binary, it is w.l.o.g. to identify any binary signal with the corresponding profile $\bm\mu=(\mu_L, \mu_R)$ of posterior means.\footnote{Indeed, one can back out the signal structure from $\bm\mu$ as follows: $\Pi\left(z=R \mid \omega=1\right)=\frac{-\mu_L\left(1+\mu_R\right)}{\mu_R-\mu_L}$ and $\Pi\left(z=R \mid \omega=-1\right)=\frac{-\mu_L\left(1-\mu_R\right)}{\mu_R-\mu_L}$. \label{footnote2}} For any policy profile $(-a,a)$ with $a 
\geq 0$, we shall hereinafter write $(\mu_L^b(a), \mu_R^b(a))$ for the optimal broadcast signal, and  $(\mu_L^{p}(a,k), \mu_R^{p}(a,k))$ for the optimal personalized signal for type $k$ voters. The next observation is useful for stating our result. 
\begin{obs}\label{obs_signal}
\begin{enumerate}[(i)]
\item $\bm\mu$ is more Blackwell-informative than $\bm\mu'$ if $|\mu_z| \geq |\mu_z'|$ $\forall z \in \{L, R\}$, and at least one inequality is strict. 
\item $\bm\mu$ endorses candidate $z \in \{L, R\}$ more often than candidate $-z$, i.e., $\pi_z>\pi_{-z}$, if and only if $|\mu_z|<|\mu_{-z}|$. 
\end{enumerate}
\end{obs}

  \begin{thm}\label{thm_product}
Under Assumptions \ref{assm_attention} and \ref{assm_regularity}, the following hold for any policy profile $(-a,a)$ with $a >0$.
\begin{enumerate}[(i)]
\item The optimal broadcast signal is \emph{symmetric}, in that it endorses each candidate with equal probability, and the endorsements shift voters' beliefs by the same magnitude,   i.e., $|\mu_L^{b}\left(a\right)|=\mu_R^{b}\left(a\right)$.
\item The following happen in the personalized case. 
  \begin{enumerate}[(a)]
  \item The optimal signal for centrist voters is symmetric, i.e., $|\mu_L^{p}\left(a,0\right)|=\mu_R^{p}\left(a,0\right)$. 
  \item The optimal signal for any extreme voter is \emph{skewed}, in that it endorses the voter's own-party candidate more often than his opposite-party candidate, although the endorsement for the opposite-party candidate is stronger than that of the own-party candidate, i.e., $|\mu_L^p(a,-1)|<\mu_R^p(a, -1)$ and $|\mu_L^p(a,1)|>\mu_R^p(a,1)$. Moreover, optimal signals are symmetric between left-leaning and right-leaning voters, $|\mu_L^p\left(a, -1\right)|=\mu_R^p\left(a,1\right)$ and $\mu_R^p(a,-1)=|\mu_L^p(a,1)|$. 
\end{enumerate}
\item The optimal broadcast signal is less Blackwell-informative than the optimal personalized signal for centrist voters, i.e., $|\mu_L^b(a)|<|\mu_L^p(a,0)|$ and $\mu_R^b(a)<\mu_R^p(a,0)$.
\end{enumerate}
 \end{thm}
    
Part (i) of Theorem \ref{thm_product} holds because the broadcast signal is designed for a representative voter with a symmetric policy preference and so must be symmetric. To develop intuition for Part (ii) of the theorem, recall that news consumption is useful for an extreme voter if and only if it sometimes convinces him to vote across the party line. Since the corresponding signal realization must move the posterior mean of the state far away from the prior mean, it must occur with a small probability in order to prevent the attention  cost from being excessive and the voter from tuning out. Hereinafter, we shall refer to this signal realization as an \emph{occasional big surprise}. The flip side of  occasional big surprise is a \emph{predisposition reinforcement}, meaning that most of the time, the signal endorses the voter's own-party candidate, which by Bayes' plausibility can only shift his belief moderately. Evidence for occasional big surprise and predisposition reinforcement after the use of personalized news aggregators has already been discussed in  Footnote \ref{footnote_bayesian}. 

We finally turn to Part (iii) of Theorem \ref{thm_product}. As demonstrated earlier, the optimal broadcast signal is designed for a representative voter with a symmetric policy preference, and yet the decision on whether to consume the signal is made by extreme voters who prefer skewed signals to symmetric ones. Such a mismatch of preferences limits the amount of attention that the optimal broadcast signal can attract from any voter compared to his personalized signal. Symmetry then implies that the optimal broadcast signal is less Blackwell-informative than centrist voters' personalized signal. 

\subsection{Equilibrium policies}\label{sec_policy}
This section endogenizes candidates' policy positions. Under segmentation technology $\mathcal{S} \in \left\{b,p\right\}$, a profile of signals and policies $(-a,a)$ with $a \geq 0$ can arise in a symmetric PSPBE if the following are true.
\begin{itemize}
\item The profile of signals is a $|\mathcal{S}|$-dimensional random variable. The marginal probability distribution of each dimension $s \in \mathcal{S}$ solves Problem (\ref{eqn_problems}), taking $(-a,a)$ as given;

\item The policy position $a$ maximizes candidate $R$'s winning probability, taking candidate $L$'s position $-a$, the profile of signals, voters' consumption decisions,  and their voting strategies (as functions of \emph{actual} policies and signal realizations) as given.
\end{itemize}

We characterize all symmetric PSPBEs of the game. Before proceeding, note that the analysis so far has pinned down the marginal signal distribution for each market segment but has left the joint signal distribution across market segments unspecified, despite that the latter clearly affects candidates' strategic reasoning. In what follows, we first assume that signals are \emph{conditionally independent} across market segments. Later in Section \ref{sec_extension},  we will  consider  \emph{all} joint signal distributions that are \emph{consistent} with the marginal distributions solved in Section \ref{sec_news}.

\paragraph{Key concepts} Fix any segmentation technology $\mathcal{S}\in \{b, p\}$ and voter population distribution $q$. Our first concept concerns how a candidate's unilateral deviation from a symmetric policy profile can affect voters' voting decisions. Due to symmetry, it suffices to consider candidate $R$'s deviation only.

\begin{defn}[label=defn_attract] 
A unilateral deviation of candidate $R$ from a policy profile $(-a,a)$ with $a \geq 0$ to $a'$ \emph{attracts} type $k$ voters if it wins the latter's support even when their signal realization disapproves of candidate $R$, i.e.,  \[v\left(-a,a', k\right)+\mu_L^{\mathcal{S}}\left(a,k\right)> 0.\]
It \emph{repels} type $k$ voters if it loses their support even when their signal realization endorses candidate $R$, i.e., \[v\left(-a,a', k\right)+\mu_R^{\mathcal{S}}\left(a,k\right) < 0.\]
\end{defn}

Note that if $a'$ attracts (resp. repels) a voter, then it makes the voter  vote for (resp. against) candidate $R$ unconditionally. If it neither attracts or repels a voter, then it has no effect on his voting decisions. 

We next construct an index called \emph{policy latitude} and use it to capture a voter's resistance to candidate $R$'s deviations. For starters, recall that $|\mu_z^{b}(a)|$ and $|\mu_z^p(a,k)|$ capture the magnitudes of voters' beliefs about candidate $R$'s valence given signal realization $z \in \{L, R\}$ under policy profile $(-a,a)$. For ease of notation,  write $\upsilon_z^{b}$ for $\left. |\mu_z^b(a)|\right\vert_{a=t(1)}$ and $\upsilon_z^p(k)$ for $\left.|\mu_z^p(a,k)|\right\vert_{a=|t(k)|}$. Intuitively, $\upsilon_L^b$ and $\upsilon_L^p(k)$ capture voters' pessimism about candidate $R$'s valence given unfavorable information. Based on them, we can define policy latitudes as follows. 

\begin{defn}[label=defn_latitude]
Define centrist voters' policy latitude in the broadcast case as $\xi^{b}(0)= \upsilon_L^{b}$, and type $k$ voters' policy latitude in the personalized case as $\xi^{p}(k)= -t(k)+\upsilon_L^p(k)$. 
\end{defn} 

By definition, a voter's policy latitude decreases with his preference for candidate $R$'s policies and increases with his  pessimism about candidate $R$'s valence given unfavorable information. Increasing a voter's policy latitude makes him more resistant to candidate $R$'s deviations. 

We finally describe equilibrium outcomes. Let $\mathcal{E}^{\mathcal{S},q}$ denote the set of the  nonnegative policy $a$'s such that the symmetric policy profile $(-a,a)$ can arise in an equilibrium. We are interested in \emph{policy polarization} $a^{\mathcal{S},q}=\max \mathcal{E}^{\mathcal{S},q}$, defined as the maximal symmetric equilibrium policy, and whether all policies between zero and policy polarization can arise in equilibrium. 

\begin{defn}[label=defn_discipline]
Type $k$ voters are \emph{disciplining} if their policy latitude determines policy polarization, i.e., $a^{\mathcal{S},q}=\xi^{\mathcal{S}}\left(k\right)$.
\end{defn}

\paragraph{Equilibrium characterization} 
The next theorem gives a full characterization of the equilibrium policy set. 

\begin{thm}[label=thm_main]
For any segmentation technology $\mathcal{S} \in \left\{b,p\right\}$ and population distribution $q$, policy polarization is strictly positive, and all policies between zero and policy polarization can arise in equilibrium, i.e., $a^{\mathcal{S}, q}>0$ and  $\mathcal{E}^{\mathcal{S},q}=\left[0, a^{\mathcal{S},q}\right]$. Disciplining voters always exist, and their identities are as follows. 
\begin{enumerate}[(i)]
 \item In the broadcast case, centrist voters are always disciplining, i.e., $a^{b,q}=\xi^b\left(0\right)$ $\forall q$.
 \item In the personalized case, centrist voters are disciplining if they constitute a majority of the population. Otherwise voters with the smallest policy latitude are disciplining, i.e., 
 \[a^{p,q}=\begin{cases}
 \xi^p\left(0\right) & \text{ if } q\left(0\right)>1/2,\\
 \displaystyle\min_{k \in \mathcal{K}}\xi^{p}\left(k\right) & \text{ if } q\left(0\right)\leq 1/2.
 \end{cases}\]
\end{enumerate} 
\end{thm}

The intuition behind Theorem \ref{thm_main} is as follows: When voters' population distribution is sufficiently dispersed, personalized news aggregation allows candidates to benefit from attracting extreme voters in addition to attracting centrist voters. Since voters with the smallest policy latitude are most susceptible to policy deviations, they constitute the easiest target of a deviating candidate. Their policy latitude---which captures their resistance to policy deviations---determines equilibrium policy polarization. Regardless of whether news aggregation is personalized or not, policy polarization is strictly positive despite that candidates are office-motivated:  Due to strict obedience, local deviations from a policy profile wouldn't change voters' voting decisions, which suggests that a positive degree of policy polarization could arise in equilibrium. 

\paragraph{Proof sketch} We proceed in three steps. 

\bigskip

\noindent \emph{Broadcast case.} In the broadcast case, all voters consume the same signal and so form the same belief about candidates' valence. Thus the median-voter-theorem logic holds, namely a deviation of candidate $R$ is profitable, i.e., strictly increases his winning probability, if and only if it attracts centrist voters. Formally (and no more proof is required), 

\begin{lem}\label{lem_broadcast}
In the broadcast case, a policy profile $(-a,a)$ with $a \geq 0$ can arise in equilibrium if and only if no deviation of candidate $R$ to any $a' \in \mathbb{R}$ attracts centrist voters. 
\end{lem}

\noindent \emph{Personalized case. } In the personalized case, Lemma \ref{lem_broadcast} remains valid if centrist voters constitute a majority coalition. Otherwise no type of voter alone  forms a majority coalition, and a deviation is profitable if it attracts \emph{any} type of voter, \emph{holding other things constant}. The reason is pivotality: Since the infomediary can now offer conditionally independent signals to different types of voters, the above deviation strictly increases candidate $R$'s winning probability when the remaining voters disagree about which candidate to vote for. 

The above argument leaves open the question of whether attracting some voters would cause the repulsion of others. Fortunately, this concern is ruled out by the next lemma. 

\begin{lem}[label=lem_maxmin]
In the personalized case with $q\left(0\right)\leq 1/2$, a policy profile $(-a,a)$ with $a \geq 0$ can arise in equilibrium if and only if no deviation of candidate $R$ to any $a' \in [-a,a)$ attracts any voter whose bliss point lies inside $\left[-a,a\right]$. 
\end{lem}

Lemma \ref{lem_maxmin} exploits basic properties of voters' utility functions such as weak concavity and inverted V-shape, as well as the strict obedience induced by optimal signals. To get a sense of how its proof works, consider  two kinds of global deviations from a symmetric policy profile $(-a,a)$ with $0\leq a<t(1)$: (1) $a' \notin \left[-a,a\right]$ and (2) $a' \in [-a,a)$. By committing a deviation of the first kind to $a'>a$ (as depicted in Figure \ref{figure1}), candidate $R$ may indeed attract right-leaning voters. But such a success must cause the repulsion of left-leaning voters, due to the symmetry and  weak concavity of voters' utility functions (see Appendix \ref{sec_proof_policy} for technical details). In addition, the deviation moves candidate $R$ away from centrist voters and hence runs the risk of repelling them, so it cannot benefit the candidate overall. The argument for why any deviation to $a'<-a$ is unprofitable is analogous and hence is omitted.

\begin{figure}[!h]
\centering
     \includegraphics[scale=.3]{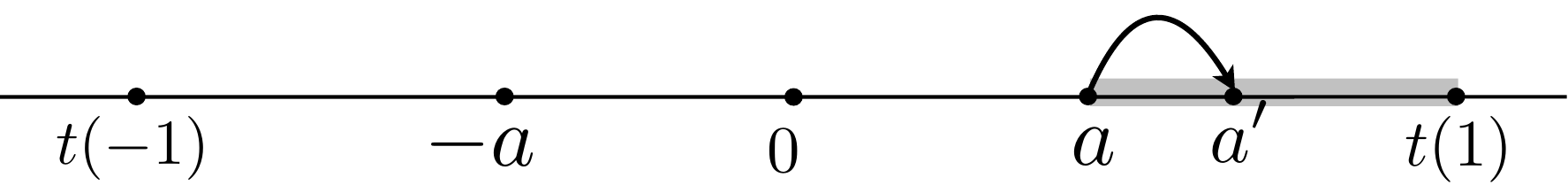}
    \caption{Consequences of candidate $R$ deviating from policy profile $(-a,a)$ to $a'>a$ when $0\leq a<t(1)$.}\label{figure1}
\end{figure} 

Consider next a deviation $a' \in [-a,a)$ of the second kind (as depicted in Figure \ref{figure2}), which moves candidate $R$ closer to centrist and left-leaning voters. While the deviation might indeed attract centrist voters, it doesn't attract  left-leaning voters, as the latter lie closer to candidate $L$'s position $-a$ than $a'$. Using a similar line of reasoning, we can demonstrate that $a'$ neither attracts nor repels right-leaning voters. Thus in order to sustain the original policy profile in an equilibrium, all we need to rule out is the possibility that the deviation attracts centrist voters.

\begin{figure}[!h]
\centering
     \includegraphics[scale=.3]{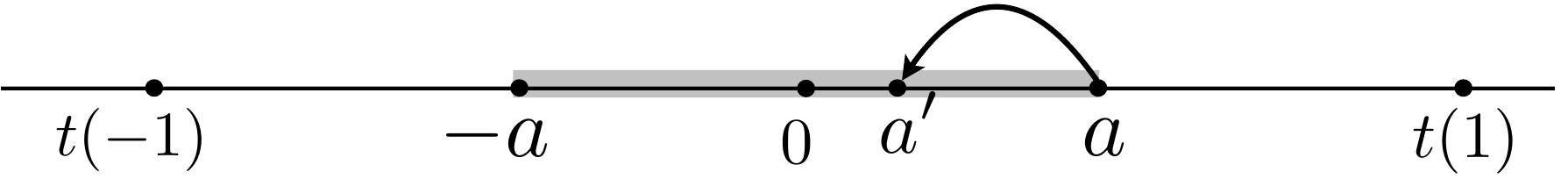}
    \caption{Consequences of candidate $R$ deviating from policy profile $(-a,a)$ to $a' \in [-a,a)$ when $0\leq a<t(1)$.  }\label{figure2}
\end{figure}


\noindent \emph{Equilibrium policy set.} The above argument establishes that a policy profile can arise in equilibrium if it deters deviations that aim at attracting certain types of voters. It is thus useful to  express equilibrium policies in terms of \emph{attraction-proof sets}.

\begin{defn}\label{defn_kproofsimplified}
Under any segmentation technology $\mathcal{S}\in \{b,p\}$,  the \emph{attraction-proof set} for type $k$  voters, denoted by $\Xi^{\mathcal{S}}(k)$, is the set of the nonnegative policy $a$'s such that no unilateral deviation of candidate $R$ from $(-a,a)$ can attract them. Since $a'=t(k)$ is the most attractive deviation to type $k$ voters, we thus have $\Xi^{\mathcal{S}}(k)=\{a \geq 0: v(-a,t(k), k)+\mu_L^{\mathcal{S}}(a,k)\leq 0\}$.
\end{defn}

Rephrasing Lemmas \ref{lem_broadcast} and \ref{lem_maxmin} using attraction-proof sets yields $\mathcal{E}^{b, q}=\Xi^b(0)$ and 
\[
\mathcal{E}^{p,q}=\begin{cases}
\Xi^p(0) & \text{ if } q(0)>1/2,\\
\left([0, t\left(1\right)) \cap \Xi^{p}\left(0\right)\right) \cup \left( [t\left(1\right), +\infty) \cap \cap_{k \in \mathcal{K}}\Xi^{p}\left(k\right)\right) & \text{ if } q(0) \leq 1/2. 
\end{cases}\]
To further simplify these results, we exploit the following properties of attraction-proof sets.
\begin{lem}\label{lem_xi}
(i)  $\xi^b(0)>0$ and $\Xi^b(0)=[0, \xi^b(0)]$; (ii) $\forall k \in \mathcal{K}$, $\xi^p(k)>|t(k)|$ and $\Xi^p(k) \cap [|t(k)|, +\infty)= [|t(k)|, \xi^p(k)]$. 
\end{lem}

Lemma \ref{lem_xi} clarifies the roles of strict obedience and policy latitude in determining equilibrium policy polarization and disciplining voter. For each $k \geq 0$, the result that $t(k)$ (and, by continuity, a neighborhood surrounding it) belongs to $\Xi^{s}(k)$ follows immediately from type $k$ voters' strict obedience at policy profile $(-t(k), t(k))$: $v(-t(k), t(k), k)+\mu_L^{\mathcal{S}}(t(k), k)<0$. Simple as it is, this result implies that equilibrium policy polarization is strictly positive. 

Meanwhile,  the fact that $\max \Xi^{\mathcal{S}}(k)=\xi^{\mathcal{S}}(k)$ isn't a coincidence: In the general model presented in Online Appendix \ref{sec_general}, we define a voter's policy latitude directly as the maximum attraction-proof position for him. Under distance utility function, policy latitudes have closed-form solutions with clear economic interpretations.  When voters' population distribution is dispersed, we must deter candidates from attracting any voter. The maximum attraction-proof position for every voter is $\min_{k\in \mathcal{K}}\xi^p(k)$,  which should intuitively pin down  the maximum policy that could arise in equilibrium.  The proof presented in the appendix formalizes this intuition by ``gluing'' the various forming pieces of $\mathcal{E}^{p,q}$ together into the desired result.  More work is needed to show that the entire interval $[|t(k)|, \xi^{\mathcal{S}}(k)]$ is contained in the attraction-proof set,  for reasons independent of the choice of the attention cost function. 

\subsection{Comparative statics}\label{sec_cs}
This section examines the comparative statics of equilibrium policy sets. Since all policies between zero and policy polarization can arise in equilibrium, it is w.l.o.g. to focus on the comparative statics of policy polarization: As policy polarization increases, the equilibrium policy set increases in the strong set order. 

\paragraph{Segmentation technology}The next proposition concerns the policy polarization effect of (disabling) personalized news aggregation \citep{gdpr, warren}. 

\begin{prop}\label{prop_personalization}
Policy polarization is strictly higher in the personalized case than in the broadcast case if and only if one of the following situations happens in the personalized case.
\begin{enumerate}[(i)]
\item Centrist voters are disciplining.
\item Right-leaning voters are disciplining, and the belief induced by the occasional big surprise of their personalized signal is sufficiently strong:  $\upsilon_L^p\left(1\right)>\upsilon_L^b+t\left(1\right)$.
\item Left-leaning voters are disciplining and have a sufficiently strong policy preference: $|t\left(-1\right)|>\upsilon_L^b-\upsilon_L^p\left(-1\right)$.
\end{enumerate} 
\end{prop}

Proposition \ref{prop_personalization} follows immediately from Theorem \ref{thm_main}. Part (i) of the proposition exploits the fact that centrist  voters' personalized signal is more Blackwell-informative than the broadcast signal, hence their policy latitude increases as news aggregation becomes personalized.

Parts (ii) and (iii) of Proposition \ref{prop_personalization} show that if extreme voters are disciplining in the personalized case, then the skewness of their signals is crucial for sustaining a greater degree of policy polarization than in the broadcast case. The role of skewness differs according to which type of extreme voter is disciplining. In the case where base voters (i.e., right-leaning voters) are disciplining, the only explanation for why they could have a big policy latitude must be the occasional big surprise of their personalized signal. Indeed, we require that base voters be significantly more pessimistic about candidate $R$'s valence following unfavorable information than the centrist voters in the broadcast case, i.e., $\upsilon_L^p\left(1\right)>\upsilon_L^b+t\left(1\right)$. 

In the case where opposition voters (i.e., left-leaning voters) are disciplining, a presumption is that they have a smaller policy latitude than base voters, i.e., $-t\left(-1\right)+\upsilon_L^p\left(-1\right)<-t\left(1\right)+\upsilon_L^p\left(1\right)$. 
The last condition, while delicate at first sight, stems naturally from the trade-off between voters' policy preferences and their beliefs about candidate valence: Since base voters most prefer candidate $R$'s policies, they seek the biggest occasional surprise and so are most pessimistic about candidate $R$'s valence following unfavorable information. In contrast, opposition voters least prefer candidate $R$'s policies but are nonetheless most optimistic about his valence following unfavorable information. For the above condition to hold, the difference in base and opposition voters' beliefs must exceed the difference in their policy preferences, i.e., $\upsilon_L^p\left(1\right)-\upsilon_L^p\left(-1\right)>t(1)-t(-1)$. Simplifying the last condition using symmetry yields
\[\tag{$\ast$} \upsilon_L^p\left(1\right)-\upsilon_R^p\left(1\right)>2t\left(1\right),\label{eqn_ast} \]
which stipulates that extreme voters' personalized signals be  sufficiently skewed that the beliefs induced by the occasional big surprise and own-party bias differ by a significant amount. In that case, candidate $R$  wouldn't target his base when contemplating a deviation. Instead, he appeals to his opposition, which itself could be challenging due to the latter's preference against his policies. When such an anti-preference is sufficiently strong, i.e., $|t\left(-1\right)|>\upsilon_L^b-\upsilon_L^p\left(-1\right)$,  policy polarization increases as a result of personalized news aggregation. Notice the role of skewness in the above argument, which is crucial yet indirect. 

The next example reduces Proposition \ref{prop_personalization} to model primitives under quadratic attention cost. See also Appendix \ref{sec_numerical} for the numerical solutions for the case of entropy attention cost.

\begin{example}\label{exm_quadratic}
When $h(\mu)=\mu^2$, solving voters' equilibrium belief magnitudes and policy latitudes  yields 
\[(\upsilon_L^b, \upsilon_R^b)=(\frac{1+\sqrt{1-16\lambda t(1)}}{4\lambda},\frac{1+\sqrt{1-16\lambda a, t(1)}}{4\lambda}),\]
\[(\upsilon_L^p(k), \upsilon_R^p(k))=\begin{cases}
(\frac{1}{2\lambda}-4t(1), \frac{1}{2\lambda}) &\text{ if } k=-1,\\
(\frac{1}{2\lambda}, \frac{1}{2\lambda}) &\text{ if } k=0,\\
(\frac{1}{2\lambda}, \frac{1}{2\lambda}-4t(1)) &\text{ if } k=1,
\end{cases}\]
and 
\[\xi^b(0)=\frac{1+\sqrt{1-16\lambda t(1)}}{4\lambda} \text{ and } 
\xi^p(k)=\begin{cases}
\frac{1}{2\lambda}-3t(1) &\text{ if } k=-1,\\
\frac{1}{2\lambda} & \text{ if } k=0,\\
\frac{1}{2\lambda}-t(1) & \text{ if } k=1.
\end{cases}\]
Thus under personalized news aggregation,  Condition (\ref{eqn_ast}) always holds: $\upsilon_L^p(1)-\upsilon_R^p(1)=4t(1)>2t(1)$, and opposition voters are always disciplining: $\min_{k \in \mathcal{K}}\xi^p(k)=\xi^p(-1)$.  Condition (\ref{eqn_doubleast}) thus becomes $|t(-1)|>\upsilon_L^p(-1)-\upsilon_L^b$, and solving it explicitly yields $\lambda t(1) \geq 1/18\approx .55$.  

To develop intuition for the last result, note that as extreme voters' policy preferences become stronger, they find news consumption less useful. Under personalized news aggregation, the infomediary is reluctant to cut back $\upsilon_R^p(1)=1/2\lambda$, the occasional big surprise that makes news consumption valuable to left-leaning voters, and so must cut back $\upsilon_L^p(1)=1/2\lambda-3t(1)$ significantly in order to prevent these voters from tuning out.  Under broadcast news aggregation,  $\upsilon_L^b$ and $\upsilon_R^b$ must decrease by the same amount to prevent extreme voters from tuning out.  When $t(1)$ is sufficiently large, the right-hand side of Condition (\ref{eqn_doubleast}) is small, whereas the left-hand side of it is big, hence the condition is satisfied.  

Meanwhile as the attention cost parameter increases, signals must become less Blackwell-informative in order to prevent voters from tuning out,  hence $\upsilon_L^p(1)$ and $\upsilon_L^b$ both decrease. When $\lambda$ is large, the right-hand side of Condition (\ref{eqn_doubleast}) is small whereas the left-hand side of it is independent of $\lambda$, hence the condition holds.  $\hfill \diamondsuit$
\end{example}

\paragraph{Attention cost parameter} 
The next proposition shows that policy polarization decreases with  voters' attention cost parameter. 

\begin{prop}\label{prop_lambda}
Let $\lambda'>\lambda>0$ be two attention cost parameters such that the corresponding environments satisfy Assumptions \ref{assm_attention} and \ref{assm_regularity}, and the beliefs induced by optimal signals take interior values between $-1$ and $1$. As we increase the attention cost parameter from $\lambda$ to $\lambda'$, policy polarization strictly decreases under both broadcast news aggregation and personalized news aggregation. 
\end{prop}

The proof of Proposition \ref{prop_lambda} exploits an important  fact: As the attention cost parameter increases,  optimal signals become less Blackwell-informative, which attenuate voters' beliefs about candidate valence. As voters become more susceptible to policy deviations, their policy latitudes fall.  

Proposition \ref{prop_lambda} sheds light on the policy polarization effect of introducing perfect competition between infomediaries, which is advocated by the British government as a preferable way of regulating tech giants \citep{digitalcompetition}. In Online Appendix \ref{sec_competitive}, we investigate an extension of the baseline model, whereby each type $k \in \mathcal{K}$ of voter is served by multiple infomediaries competing \`{a} la Bertrand and solving 
\[\max_{\Pi} V\left(\Pi; {\bf{a}},k\right)-\lambda I\left(\Pi\right).\]
The solution to this problem, which we name as the \emph{competitive signal} for type $k$ voters, coincides with their monopolistic personalized signal for some  attention cost parameter $\lambda'>\lambda$.  Intuitively, monopolistic personalized signals overfeed voters with information about candidate valence through reducing the attention cost parameters they effectively face.  Introducing competition between infomediaries corrects this overfeeding problem; its policy polarization effect is negative by Proposition \ref{prop_lambda}.  

\paragraph{Population distribution}
Recently, a growing body of the literature has been devoted to the understanding of voter polarization, also termed \emph{mass polarization}. Notably, \cite{fiorina} define mass polarization as a bimodal distribution of voters' policy preferences on a liberal-conservative scale, and \cite{gentzkow} develops a related concept that measures the average ideological distance between Democrats and Republicans. Inspired by these authors, we define \emph{increasing mass polarization} as a mean-preserving spread of voters' policy preferences. The next proposition shows that under personalized news aggregation, increasing mass polarization may surprisingly reduce policy polarization rather than increasing it. 

\begin{prop}\label{prop_population}
Let $q$ and $q'$ be two population distributions such that \emph{the mass is more polarized under $q'$ than under $q$}, i.e., $q(0)>q'(0)$. As we change the population distribution from $q$ to $q'$ in the personalized case, policy polarization weakly decreases, and it strictly decreases if $q\left(0\right)>1/2 \geq q'\left(0\right)$ and $\min\left\{\xi^p\left(-1\right), \xi^p\left(1\right)\right\}<\xi^p\left(0\right)$. 
\end{prop}

Proposition \ref{prop_population} follows immediately from Theorem \ref{thm_main}: As we keep redistributing voters' population from the center to the margin, candidates would eventually benefit from attracting extreme voters in addition to attracting centrist voters. If any extreme voter has a smaller policy latitude than that of the centrist voters as in the case of quadratic attention cost, then a reduction in the policy polarization will ensue. 

\section{Extensions}\label{sec_extension}
In this section, we report main extensions of the baseline model and their takeaways without touching any technical detail. See the online appendices for formal analysis. 

\paragraph{General voters and joint signal distribution} In Online Appendix \ref{sec_general}, we extend the baseline model to arbitrary finite types of voters holding general policy preferences. We also relax the assumption that signals are conditionally independent across market segments, and instead consider joint signal distributions that are consistent with the marginal distributions as solved in Section \ref{sec_news}.  We will impose regularity conditions on either voters' utility functions or their attention cost function,  but rest be assured that the many-voter version of the baseline model is nested as a special case. 

Our analysis leverages a new concept called \emph{influential coalition}. Loosely  speaking, a coalition of voters is influential if attracting all its members, holding other things constant, strictly increases the deviating candidate's winning probability. In the broadcast case, signals are perfectly correlated among voters, so a coalition of voters is influential if and only if it is a majority coalition. In the personalized case, non-majority coalitions can be influential, due to the imperfect correlation between different voters' signals.  Table \ref{table1} compiles the influential coalitions in the baseline model.

\begin{table}[ht]
\centering
\begin{tabular}{ r|c|c| }
\multicolumn{1}{r}{}
 &  \multicolumn{1}{c}{$\mathcal{S}=b$}
 & \multicolumn{1}{c}{$\mathcal{S}=p$} \\
\cline{2-3}
$q(0)>1/2$ & majority coalitions & majority coalitions \\
\cline{2-3}
$q(0)< 1/2$ & majority coalitions & nonempty coalitions\\
\cline{2-3}
\end{tabular}
\caption{Influential coalitions in the baseline model.}
\label{table1}
\end{table}

Personalized news aggregation affects policy polarization through changing the marginal signal distributions, as well as the influential coalitions. So far we've focused on the first effect, under the restriction that personalized signals are conditionally independent across voters. As demonstrated in Online Appendix \ref{sec_general}, lifting the last restriction while holding marginal signal distributions fixed can only increase policy polarization. Among all joint signal distributions and voter  population distributions, the exact lower bound for policy polarization: $\min_{k \in \mathcal{K}} \xi^p(k)$,  is attained when signals are conditionally independent across voters and voters' population distribution is uniform across types. 
Both findings follow from a characterization of policy polarization as the minimum policy latitude among all influential coalitions, as well as the comparative statics of influential coalitions as we vary the joint signal distribution, holding marginal signal distributions fixed. 

Two takeaways are immediate. First, results so far prescribe the exact lower bound for the policy polarization effect of personalized news aggregation. Second, as long as the lower bound $\min_{k \in \mathcal{K}}\xi^p(k)$ stays positive, changes in the environment (e.g., enrich voters' types, divide voters of the same type into multiple subgroups) wouldn't render policy polarization trivial.  

\paragraph{A continuum of states} In Online Appendix \ref{sec_state}, we extend the analysis to a continuum of states while assuming mutual information as the attention cost. All previous findings regarding the  personalized case remains qualitatively valid. As for the broadcast case, we show, as in the baseline model, that any optimal signal has at most three signal realizations: LL, LR, and RR. Interestingly, this result holds for arbitrary finite types of voters, because among all voters, only those of the most extreme types can have binding participation constraints, and the signal acquired by the representative voter acting on their behalves prescribes at most three voting recommendation profiles as above. 

The case of two signal realizations can be solved analogously as before. In the new case of three signal realizations, we argue, using the convexity of mutual information in the signal structure \citep{cover}, that the optimal signal must be symmetric, hence the posterior mean of the state given signal realization LR must equal zero. Given this, we then argue that equilibrium policy polarization must equal zero, hence the personalization of news aggregation always strictly increases policy polarization.

\section{Related literature}\label{sec_literature}
The current paper contributes to three strands of the economic literature: Rational Inattention (RI),  media bias, and electoral competition. 

\paragraph{Rational inattention} The literature on RI pioneered by \cite{sims} and \cite{sims1} assumes that decision-makers can optimally aggregate source data into signals themselves.  To create a role for infomediaries, we assume that the aggregator is designed and operated by an infomediary, whereas voters must fully absorb the information given to them. Apart from this departure from the RI paradigm, we otherwise follow the standard model of posterior-separable attention cost that nests Shannon entropy as a special case. Posterior separability \citep{caplindean13} has recently received attention from economists because of its axiomatic and revealed-preference foundations  \citep{caplindean, tsakas, denti2022posterior, zhong2022optimal},  connections to sequential sampling \citep{hebertwoodford, morrisstrack}, and validations by lab experiments \citep{ambuehl2017offer, dean}.

The flexibility of information aggregation is essential to our predictions, as well as that of many other RI models (see \citealt{risurvey} for a survey).  It is absent from most existing political models with costly information acquisition (often referred to as  \emph{rational ignorance} models),  whereby voters can only acquire signals that follow stylized, exogenous probability distributions (see, e.g., \citealp{persico, martinelli}).\footnote{An exception is our companion paper \cite{li2020electoral}, which  examines the impact of personalized news aggregation for electoral accountability and selection, assuming that voters can aggregate information optimally themselves as in the standard RI paradigm. Here, the focuses are on the attention-maximizing signals provided by a monopolistic infomediary (though we do study perfectly competitive signals in an extension), as well as  their impacts on electoral competition.}  Evidence for attentional flexibility has been documented by the aforementioned lab experiments, as well as \cite{novak}.

\paragraph{Media bias} The current paper adds to the literature on demand-driven media bias. A high-level idea it seeks to formalize---namely even rational consumers can exhibit a preference for biased information when constrained by information processing capacities---dates back to \cite{calvert} and is later expanded on by \cite{suen}, \cite{burke}, \cite{oliveros}, and \cite{che} among others. While some of these models also predict a predisposition reinforcement and, implicitly, an occasional big surprise, they work with ad-hoc information aggregation technologies and do not examine the consequences of biased information aggregation for electoral competition. Even if they did, as in \cite{chansuen}, their predictions could still depart significantly from ours due to the subtle differences in the information aggregation technology  (more on this later).   

A recent, noteworthy contribution to the literature is made by \cite{perego2022media}, who study the impact of media entry on news personalization and opinion disagreement in a variant of Salop's circle model. The news consumers in their model have heterogeneous preferences for information, whereas media outlets with limited information processing capacities compete by adjusting the locations of their products on the circle. While the model of \cite{perego2022media} also predicts personalized information aggregation, their media outlets face  different objectives, choices, and constraints from ours. Electoral competition isn't a concern in their model, but it lies at the heart of the current analysis. 

There is also a vast literature on supply-driven media bias, studying how self-interested media could persuade voters to favor one candidate over another through biased information disclosure.  Notable, recent contributions to this literature include \cite{dugganmartinelli}, \cite{gehlbachsonin}, and \cite{pratmediapower}; see also \cite{mediahandbook} for a survey of the earlier literature.  On the methodological side, there is a growing theoretical literature studying the optimal (private) persuasion of voters whose actions may have payoff externalities on each other (see, among others, \citealp{schnakenberg, alonsocamara, salcedo}).  Since NARI can be rephrased as a game of persuading a representative voter, it is closer to the single sender, single receiver problem studied by  \cite{bayesianpersuasion} than those studied by the aforementioned studies.



\paragraph{Electoral competition} In most  existing probabilistic voting models, voters' signals are assumed to be continuously distributed, so even small changes in candidates' positions could affect their voting decisions (see \citealp{duggan} for a thorough literature survey).  Under this assumption,  \cite{calvertb} establishes a policy convergence between office-seeking candidates, and  pioneers the use of policy preferences  for generating  policy polarization between candidates (hereinafter, the \emph{Calvert-Wittman logic}). Strict obedience stands in sharp contrast to this assumption, although it is a natural consequence of NARI. 

There is a small but growing literature on electoral competition with personalized information aggregation.\footnote{Certainly,  personalized information aggregation is only one of the many mechanisms that generate information personalization.  Other notable mechanisms include shrouded attributes and targeted campaign.  \cite{glaeser} explore the first mechanism by assuming that voters can only observe the policy deviations committed by their own-party candidates.  In equilibrium, policy polarization arises due to the lack of monitoring by voters from the opposite side.  \cite{levine} formalize the second mechanism in a Calvert-Wittman model, showing that an increase in the uncertainty of voters' preferences raises  both campaign spending and policy polarization.  Ensuing studies to these  papers are numerous; they are not reviewed here due to space constraints.   } The current work differs from the existing studies in two main aspects. First, the signal  structures generated by NARI are new to the literature. Second, in order to single out the policy polarization effect of NARI, we embed the analysis in a plain probabilistic voting model where candidates are office-motivated, and the only source of uncertainty is their valence shock (see, e.g., \citealp{calvertb, duggan}). Together, these modeling choices generate new insights that even the closest works to ours tend to ignore.

\cite{chansuen} study an electoral competition model where voters care about whether the realization of a random state variable is above or below their personal thresholds. Information is provided by personal media, which form bi-partitions of the state space using threshold rules. A consequence of working with this information aggregation technology, rather than NARI, is that signal realizations are monotone in voters' thresholds, i.e., if a left-leaning voter is recommended to vote for candidate $R$, then a right-leaning voter must receive the same recommendation. As a result, centrist voters are always disciplining despite a pluralism of media.\footnote{Starting from there, the analysis of \cite{chansuen} differs completely from ours. In particular, \cite{chansuen} exploit the Calvert-Wittman logic, assuming that voters observe a preference shock in addition to the state variable disclosed by media outlets, and that candidates have policy preferences. We do not use the Calvert-Wittman logic to generate polarization.  } We instead predict that the disciplining voter can vary with model primitives and will discuss the empirical implication of this prediction in the conclusion section. 

Two recent papers: \cite{matejka} and \cite{yuksel}, study electoral competition models with personalized  information acquisition.  In \cite{matejka}, voters face normal uncertainties about candidates' policies that do \emph{not} directly enter their utility functions. Information acquisition takes the form of variance reduction, generating signals that violate strict obedience and sustain policy polarization only if the cost of information acquisition differs across candidates.  The current work differs from \cite{matejka} in the source of uncertainty, the attention technology, and the driving force behind  policy polarization. 

\cite{yuksel} studies a variant of the Calvert-Wittman model, where voter learning takes the form of partitioning a multi-dimensional issue space. Aside from these modeling differences that set our reasoning apart,\footnote{In particular, Yuksel's (2022) reasoning exploits the multi-dimensionality of the issue space and the Calvert-Wittman logic. Our results hold regardless of the dimensionality of the state space (see Footnote \ref{fn_singledimension}), and they do not exploit the Calvert-Wittman logic.  \label{footnote_singledimensional}} none of our main predictions---including the rise of policy polarization between office-motivated candidates and the comparative statics of policy polarization---have analogous counterparts in \cite{yuksel}.

\section{Concluding remarks}\label{sec_conclusion}
Tech-enabled personalization is now ubiquitous and seems to maximize social surplus by best serving individuals' needs. To us, this argument ignores the vital role of modern infomediaries in shaping consumers' beliefs and, in turn, the location choices of politicians, companies, etc. After formalizing this role, the welfare consequences of many regulatory proposals to tame tech giants become less clear-cut. For example, while enabling personalization clearly makes the monopolistic infomediary better off and news consumers worse off,  holding candidates' positions fixed, it could affect the social welfare in either way once policies become the subject of candidates' strategic reasoning.   For this reason, we caution that prudence be exercised and our equilibrium characterization be considered when evaluating the overall impacts of these proposals.\footnote{A common, alternative measure of social welfare in election models is the probability of correct selection (i.e., that of selecting the candidate with the greatest valence). In a  companion paper \cite{li2020electoral}, we examine properties of this measure in a model of electoral accountability, whereby voters can aggregate information optimally themselves as in the standard RI paradigm.  Interested readers can adapt the results therein to the current context and examine the selection effect of NARI. A thorough investigation of this subject matter is beyond the scope of the current paper.} The usefulness of our theory in this regard is illustrated by the next example.

\begin{example}In our model, a voter's equilibrium expected utility equals 
\[\underbrace{V(\Pi; \mathbf{a},k)-\lambda I(\Pi)}_{(1)}+\underbrace{u(a_L, k)}_{(2)},\] 
where the first term in the above expression equals zero for voters with binding participation constraints in news consumption, 
and the second term depends on the exact location choice of candidate $R$ relative to the voter's bliss point. Absent the second effect, enabling  personalization makes voters (weakly) worse off by allowing the monopolistic infomediary to perfectly discriminate against them. To evaluate the second effect, we use a symmetric social welfare function that assigns equal weights to left-leaning voters and right-leaning voters. The weighted sum of the second effects across voters is then decreasing in $|a_L|$. Combining the two effects shows that personalized news aggregation reduces voter welfare if it increases policy polarization. The last situation happens if and only if $\lambda t(1)>1/18$ in the context laid out in Example \ref{exm_quadratic}.  $\hfill \diamondsuit$
\end{example}

In an earlier version of this paper, we applied our theory to  the study of product differentiation between firms with personalized product information aggregation for consumers.  Interested readers can consult \cite{hu2019politics} for further details.

An important takeaway from our analysis is the indeterminacy of the disciplining voter under personalized news aggregation. This prediction, while delicate at first sight, suggests that a first step towards testing our theory is to survey political consultants and volunteers about the disciplining voter---an approach that \cite{hersh} advocates in the context of personalized campaign. It also indicates the usefulness of studying shocks to infomediaries, as they may generate the needed variations for empirical research (e.g., introducing perfect competition to infomediaries is mathematically equivalent to increasing the attention cost parameter). We hope someone,  maybe us, will pursue these agendas in the future.

\begin{spacing}{1}
\bibliographystyle{econometrica} 
\bibliography{ref.bib}

\ifx\undefined\BySame
\newcommand{\BySame}{\leavevmode\rule[.5ex]{3em}{.5pt}\ }
\fi
\ifx\undefined\textsc
\newcommand{\textsc}[1]{{\sc #1}}
\newcommand{\emph}[1]{{\em #1\/}}
\let\tmpsmall\small
\renewcommand{\small}{\tmpsmall\sc}
\fi
\begin{thebibliography}{}

\harvarditem[Alonso and C{\^a}mara]{Alonso and C{\^a}mara}{2016}{alonsocamara}
\textsc{Alonso, R.,  {\small and} O.~C{\^a}mara}  (2016): ``Persuading
  voters,'' \emph{American Economic Review}, 106(11), 3590--3605.

\harvarditem[Ambuehl]{Ambuehl}{2017}{ambuehl2017offer}
\textsc{Ambuehl, S.}  (2017): ``An offer you can't refuse? Incentives change
  what we believe,'' \emph{CESifo Working Paper}.

\harvarditem[Anderson, Waldfogel, and Stromberg]{Anderson, Waldfogel, and
  Stromberg}{2016}{mediahandbook}
\textsc{Anderson, S.~P., J.~Waldfogel,  {\small and} D.~Stromberg}  (2016):
  \emph{Handbook of Media Economics}. Elsevier.

\harvarditem[Athey, Mobius, and Pal]{Athey, Mobius, and Pal}{2021}{atheyetal}
\textsc{Athey, S., M.~Mobius,  {\small and} J.~Pal}  (2021): ``The impact of
  aggregators on Internet news consumption,'' \emph{Working Paper}.

\harvarditem[Aumann, Maschler, and Stearns]{Aumann, Maschler, and
  Stearns}{1995}{aumann}
\textsc{Aumann, R.~J., M.~Maschler,  {\small and} R.~E. Stearns}  (1995):
  \emph{Repeated Games with Incomplete Information}. MIT press, Cambridge, MA.

\harvarditem[Burke]{Burke}{2008}{burke}
\textsc{Burke, J.}  (2008): ``Primetime spin: Media bias and belief confirming
  information,'' \emph{Journal of Economics \& Management Strategy}, 17(3),
  633--665.

\harvarditem[Calvert]{Calvert}{1985a}{calvertb}
\textsc{Calvert, R.~L.}  (1985a): ``Robustness of the multidimensional voting
  model: Candidate motivations, uncertainty, and convergence,'' \emph{American
  Journal of Political Science}, 29(1), 69--95.

\harvarditem[Calvert]{Calvert}{1985b}{calvert}
\textsc{\BySame{}}  (1985b): ``The value of biased information: A rational
  choice model of political advice,'' \emph{Journal of Politics}, 47(2),
  530--555.

\harvarditem[Caplin and Dean]{Caplin and Dean}{2013}{caplindean13}
\textsc{Caplin, A.,  {\small and} M.~Dean}  (2013): ``Behavioral implications
  of rational inattention with Shannon entropy,'' \emph{Working Paper}.

\harvarditem[Caplin and Dean]{Caplin and Dean}{2015}{caplindean}
\textsc{\BySame{}}  (2015): ``Revealed preference, rational inattention, and
  costly information acquisition,'' \emph{American Economic Review}, 105(7),
  2183--2203.

\harvarditem[Chan and Suen]{Chan and Suen}{2008}{chansuen}
\textsc{Chan, J.,  {\small and} W.~Suen}  (2008): ``A spatial theory of news
  consumption and electoral competition,'' \emph{Review of Economic Studies},
  75(3), 699--728.

\harvarditem[Che and Mierendorff]{Che and Mierendorff}{2019}{che}
\textsc{Che, Y.-K.,  {\small and} K.~Mierendorff}  (2019): ``Optimal dynamic
  allocation of attention,'' \emph{American Economic Review}, 109(8),
  2993--3029.

\harvarditem[Cover and Thomas]{Cover and Thomas}{2006}{cover}
\textsc{Cover, T.~M.,  {\small and} J.~A. Thomas}  (2006): \emph{Elements of
  Information Theory}. John Wiley \& Sons, Inc., Hoboken, NJ, 2nd edn.

\harvarditem[Dean and Neligh]{Dean and Neligh}{2019}{dean}
\textsc{Dean, M.,  {\small and} N.~L. Neligh}  (2019): ``Experimental tests of
  rational inattention,'' \emph{Working Paper}.

\harvarditem[Dellarocas, Sutanto, Calin, and Palme]{Dellarocas, Sutanto, Calin,
  and Palme}{2016}{dellarocasetal}
\textsc{Dellarocas, C., J.~Sutanto, M.~Calin,  {\small and} E.~Palme}  (2016):
  ``Attention allocation in information-rich environments: The case of news
  aggregators,'' \emph{Management Science}, 62(9), 2543--2562.

\harvarditem[DellaVigna and Gentzkow]{DellaVigna and
  Gentzkow}{2010}{dellavignagentzkow}
\textsc{DellaVigna, S.,  {\small and} M.~Gentzkow}  (2010): ``Persuasion:
  Empirical evidence,'' \emph{Annual Review of Economics}, 2(1), 643--669.

\harvarditem[Denti]{Denti}{2022}{denti2022posterior}
\textsc{Denti, T.}  (2022): ``Posterior separable cost of information,''
  \emph{American Economic Review}, 112(10), 3215--3259.

\harvarditem[DeVos, Dhabalia, Shen, Holstein, and Eslami]{DeVos, Dhabalia,
  Shen, Holstein, and Eslami}{2022}{devos2022toward}
\textsc{DeVos, A., A.~Dhabalia, H.~Shen, K.~Holstein,  {\small and} M.~Eslami}
  (2022): ``Toward user-driven algorithm auditing: Investigating users’
  strategies for uncovering harmful algorithmic behavior,'' in \emph{CHI
  Conference on Human Factors in Computing Systems}, pp. 1--19.

\harvarditem[Duggan]{Duggan}{2017}{duggan}
\textsc{Duggan, J.}  (2017): ``A survey of equilibrium analysis in spatial
  model of elections,'' \emph{Working Paper}.

\harvarditem[Duggan and Martinelli]{Duggan and
  Martinelli}{2011}{dugganmartinelli}
\textsc{Duggan, J.,  {\small and} C.~Martinelli}  (2011): ``A spatial theory of
  media slant and voter choice,'' \emph{Review of Economic Studies}, 78(2),
  640--666.

\harvarditem[Fanta]{Fanta}{2018}{fanta2018publisher}
\textsc{Fanta, A.}  (2018): ``The publisher’s patron: How Google’s News
  Initiative is re-defining journalism,'' \emph{European Journalism
  Observatory}, 26.

\harvarditem[Fiorina and Abrams]{Fiorina and Abrams}{2008}{fiorina}
\textsc{Fiorina, M.~P.,  {\small and} S.~J. Abrams}  (2008): ``Political
  polarization in the American public,'' \emph{Annual Review of Political
  Science}, 11(1), 563--588.

\harvarditem[Flaxman, Goel, and Rao]{Flaxman, Goel, and Rao}{2016}{flaxman}
\textsc{Flaxman, S., S.~Goel,  {\small and} J.~M. Rao}  (2016): ``Filter
  bubbles, echo chambers, and online news consumption,'' \emph{Public Opinion
  Quarterly}, 80(S1), 298--320.

\harvarditem[Gehlbach and Sonin]{Gehlbach and Sonin}{2014}{gehlbachsonin}
\textsc{Gehlbach, S.,  {\small and} K.~Sonin}  (2014): ``Government control of
  the media,'' \emph{Journal of Public Economics}, 118, 163--171.

\harvarditem[{General Data Protection Regulation}]{{General Data Protection
  Regulation}}{2016}{gdpr}
\textsc{{General Data Protection Regulation}}  (2016): \emph{Regulation (EU)
  2016/679 of the European Parliament and of the Council}. Official Journal of
  the European Union, {April 27}.

\harvarditem[Gentzkow]{Gentzkow}{2016}{gentzkow}
\textsc{Gentzkow, M.}  (2016): ``Polarization in 2016,'' \emph{Toulouse Network
  for Information Technology Whitepaper}, pp. 1--23.

\harvarditem[Gerber, Gimpel, Green, and Shaw]{Gerber, Gimpel, Green, and
  Shaw}{2011}{gerber}
\textsc{Gerber, A.~S., J.~G. Gimpel, D.~P. Green,  {\small and} D.~R. Shaw}
  (2011): ``How large and long-lasting are the persuasive effects of televised
  campaign ads? Results from a randomized field experiment,'' \emph{American
  Political Science Review}, 105(1), 135--150.

\harvarditem[Glaeser, Ponzetto, and Shapiro]{Glaeser, Ponzetto, and
  Shapiro}{2005}{glaeser}
\textsc{Glaeser, E.~L., G.~A. Ponzetto,  {\small and} J.~M. Shapiro}  (2005):
  ``Strategic extremism: Why Republicans and Democrats divide on religious
  values,'' \emph{Quarterly Journal of Economics}, 120(4), 1283--1330.

\harvarditem[H{\'e}bert and Woodford]{H{\'e}bert and
  Woodford}{2018}{hebertwoodford}
\textsc{H{\'e}bert, B.,  {\small and} M.~Woodford}  (2018): ``Rational
  inattention in continuous time,'' \emph{Working Paper}.

\harvarditem[Herrera, Levine, and Martinelli]{Herrera, Levine, and
  Martinelli}{2008}{levine}
\textsc{Herrera, H., D.~K. Levine,  {\small and} C.~Martinelli}  (2008):
  ``Policy platforms, campaign spending and voter participation,''
  \emph{Journal of Public Economics}, 92(3-4), 501--513.

\harvarditem[Hersh]{Hersh}{2015}{hersh}
\textsc{Hersh, E.~D.}  (2015): \emph{Hacking the Electorate: How Campaigns
  Perceive Voters}. Cambridge University Press, Cambridge, U.K.

\harvarditem[Hu, Li, and Segal]{Hu, Li, and Segal}{2019}{hu2019politics}
\textsc{Hu, L., A.~Li,  {\small and} I.~Segal}  (2019): ``The politics of
  personalized news aggregation,'' \emph{arXiv preprint arXiv:1910.11405}.

\harvarditem[Kamenica and Gentzkow]{Kamenica and
  Gentzkow}{2011}{bayesianpersuasion}
\textsc{Kamenica, E.,  {\small and} M.~Gentzkow}  (2011): ``Bayesian
  persuasion,'' \emph{American Economic Review}, 101(6), 2590--2615.

\harvarditem[Lagun and Lalmas]{Lagun and Lalmas}{2016}{lagun}
\textsc{Lagun, D.,  {\small and} M.~Lalmas}  (2016): ``Understanding user
  attention and engagement in online news reading,'' in \emph{Proceedings of
  the Ninth ACM International Conference on Web Search and Data Mining}, pp.
  113--122.

\harvarditem[Li and Hu]{Li and Hu}{2020}{li2020electoral}
\textsc{Li, A.,  {\small and} L.~Hu}  (2020): ``Electoral accountability and
  selection with personalized information aggregation,'' \emph{arXiv preprint
  arXiv:2009.03761}.

\harvarditem[Ma{\'c}kowiak, Mat{\v{e}}jka, and Wiederholt]{Ma{\'c}kowiak,
  Mat{\v{e}}jka, and Wiederholt}{2021}{risurvey}
\textsc{Ma{\'c}kowiak, B., F.~Mat{\v{e}}jka,  {\small and} M.~Wiederholt}
  (2021): ``Rational inattention: A review,'' \emph{Journal of Economic
  Literature}, forthcoming.

\harvarditem[Martinelli]{Martinelli}{2006}{martinelli}
\textsc{Martinelli, C.}  (2006): ``Would rational voters acquire costly
  information?,'' \emph{Journal of Economic Theory}, 129(1), 225--251.

\harvarditem[Mat{\v{e}}jka and McKay]{Mat{\v{e}}jka and
  McKay}{2015}{matejka2015}
\textsc{Mat{\v{e}}jka, F.,  {\small and} A.~McKay}  (2015): ``Rational
  inattention to discrete choices: A new foundation for the multinomial logit
  model,'' \emph{American Economic Review}, 105(1), 272--98.

\harvarditem[Mat{\v{e}}jka and Tabellini]{Mat{\v{e}}jka and
  Tabellini}{2021}{matejka}
\textsc{Mat{\v{e}}jka, F.,  {\small and} G.~Tabellini}  (2021): ``Electoral
  competition with rationally inattentive voters,'' \emph{Journal of European
  Economic Association}, 19(3), 1899--1935.

\harvarditem[Matsa and Lu]{Matsa and Lu}{2016}{matsa}
\textsc{Matsa, K.~E.,  {\small and} K.~Lu}  (2016): ``10 facts about the
  changing digital news landscape,'' \emph{Pew Research Center}, {September
  14}.

\harvarditem[Morris and Strack]{Morris and Strack}{2019}{morrisstrack}
\textsc{Morris, S.,  {\small and} P.~Strack}  (2019): ``The Wald problem and
  the relation of sequential sampling and ex-ante information costs,''
  \emph{Available at SSRN 2991567}.

\harvarditem[Nov{\'a}k, Matveenko, and Ravaioli]{Nov{\'a}k, Matveenko, and
  Ravaioli}{2021}{novak}
\textsc{Nov{\'a}k, V., A.~Matveenko,  {\small and} S.~Ravaioli}  (2021): ``The
  status quo and belief polarization of inattentive agents: Theory and
  experiment,'' \emph{IGIER Working Paper}.

\harvarditem[Oliveros and V{\'a}rdy]{Oliveros and V{\'a}rdy}{2015}{oliveros}
\textsc{Oliveros, S.,  {\small and} F.~V{\'a}rdy}  (2015): ``Demand for slant:
  How abstention shapes voters’ choice of news media,'' \emph{Economic
  Journal}, 125(587), 1327--1368.

\harvarditem[Pariser]{Pariser}{2011}{filterbubble}
\textsc{Pariser, E.}  (2011): \emph{The Filter Bubble: How the New Personalized
  Web Is Changing What We Read and How We Think}. Penguin, New York, NY.

\harvarditem[Perego and Yuksel]{Perego and Yuksel}{2022}{perego2022media}
\textsc{Perego, J.,  {\small and} S.~Yuksel}  (2022): ``Media competition and
  social disagreement,'' \emph{Econometrica}, 90(1), 223--265.

\harvarditem[Persico]{Persico}{2004}{persico}
\textsc{Persico, N.}  (2004): ``Committee design with endogenous information,''
  \emph{Review of Economic Studies}, 71(1), 165--191.

\harvarditem[Prat]{Prat}{2018}{pratmediapower}
\textsc{Prat, A.}  (2018): ``Media power,'' \emph{Journal of Political
  Economy}, 126(4), 1747--1783.

\harvarditem[Salamanca]{Salamanca}{2021}{salamanca}
\textsc{Salamanca, A.}  (2021): ``The value of mediated communication,''
  \emph{Journal of Economic Theory}, 192, 105191.

\harvarditem[Salcedo]{Salcedo}{2019}{salcedo}
\textsc{Salcedo, B.}  (2019): ``Persuading part of an audience,'' \emph{arXiv
  preprint arXiv:1903.00129}.

\harvarditem[Schnakenberg]{Schnakenberg}{2015}{schnakenberg}
\textsc{Schnakenberg, K.~E.}  (2015): ``Expert advice to a voting body,''
  \emph{Journal of Economic Theory}, 160, 102--113.

\harvarditem[Shannon]{Shannon}{1948}{shannon1948mathematical}
\textsc{Shannon, C.~E.}  (1948): ``A mathematical theory of communication,''
  \emph{The Bell System Technical Journal}, 27(3), 379--423.

\harvarditem[Sims]{Sims}{1998}{sims}
\textsc{Sims, C.~A.}  (1998): ``Stickiness,'' in \emph{Carnegie-Rochester
  Conference Series on Public Policy}, vol.~49, pp. 317--356. Elsevier.

\harvarditem[Sims]{Sims}{2003}{sims1}
\textsc{\BySame{}}  (2003): ``Implications of rational inattention,''
  \emph{Journal of Monetary Economics}, 50(3), 665--690.

\harvarditem[Str{\"o}mberg]{Str{\"o}mberg}{2015}{strombergsurvey}
\textsc{Str{\"o}mberg, D.}  (2015): ``Media and politics,'' \emph{Annual Review
  of Economics}, 7(1), 173--205.

\harvarditem[Suen]{Suen}{2004}{suen}
\textsc{Suen, W.}  (2004): ``The self-perpetuation of biased beliefs,''
  \emph{Economic Journal}, 114(495), 377--396.

\harvarditem[Sunstein]{Sunstein}{2009}{sunstein}
\textsc{Sunstein, C.~R.}  (2009): \emph{Republic.com 2.0}. Princeton University
  Press, Princeton, NJ.

\harvarditem[{The Digital Competition Expert Panel}]{{The Digital Competition
  Expert Panel}}{2019}{digitalcompetition}
\textsc{{The Digital Competition Expert Panel}}  (2019): \emph{Unlocking
  Digital Competition}. U.K.

\harvarditem[Tsakas]{Tsakas}{2020}{tsakas}
\textsc{Tsakas, E.}  (2020): ``Robust scoring rules,'' \emph{Theoretical
  Economics}, 15(3), 955--987.

\harvarditem[Warren]{Warren}{2019}{warren}
\textsc{Warren, E.}  (2019): ``Here’s how we can break up Big Tech,''
  \emph{Medium}, {March 8}.

\harvarditem[Yuksel]{Yuksel}{2022}{yuksel}
\textsc{Yuksel, S.}  (2022): ``Specialized learning and political
  polarization,'' \emph{International Economic Review}, 63(1), 457--474.

\harvarditem[Zhong]{Zhong}{2022}{zhong2022optimal}
\textsc{Zhong, W.}  (2022): ``Optimal dynamic information acquisition,''
  \emph{Econometrica}, 90(4), 1537--1582.

\end{thebibliography}


\begin{thebibliography}{}
 \bibitem[\protect\citeauthoryear{Cover and Thomas}{2006}]{infotheory}
 \textsc{Cover, T. M., and J. A. Thomas} (2006):
 \textit{Elements of Information Theory,}
 Hoboken, NJ:  John Wiley \& Sons, 2nd ed.

\bibitem[\protect\citeauthoryear{Mat\v{e}jka and McKay}{2015}]{mckay}
 \textsc{Mat\v{e}jka, F., and A. McKay} (2015):
 ``Rational inattention to discrete choices:
 A new foundation for the multinomial logit model,''
 \textit{American Economic Review},
 105(1), 272-298.

 \end{thebibliography}
\end{spacing}

\appendix

\newpage

\section{Proofs}\label{sec_proof}
The proofs presented in this appendix exploit the following properties of the distance utility function. 

\begin{obs}\label{obs_utility}
$u(a,k)=-|t(k)-a|$ satisfies the following properties, provided that the bliss point function 
$t: \mathcal{K} \rightarrow \mathbb{R}$ is strictly increasing and is symmetric around zero. 
\begin{description}
\item[Continuity and weak concavity]$u
\left(\cdot, k\right)$ is continuous and weakly concave for any $k \in \mathcal{K}$.

\item[Symmetry] $u\left(a,k\right)=u\left(-a,-k\right)$ for any $a\in \mathbb{R}$ and $k \in \mathcal{K}$.

\item[Inverted V-shape] $u\left(\cdot, k\right)$ is strictly increasing on $(-\infty, t\left(k\right)]$ and is strictly decreasing on $[t\left(k\right), +\infty)$ for any $k \in \mathcal{K}$. 

\item[Increasing differences] $v(-a,a',k) \coloneqq u\left(a,k\right)-u\left(a',k\right) $ is increasing in $k$ for any $a>a'$. For any $a>0$, $v(-a,a,k) \coloneqq u\left(a,k\right)-u\left(-a,k\right)$ is strictly positive if $k=1$, equals zero if $k=0$, and is strictly negative if $k=-1$. 
 \end{description}
\end{obs}

\subsection{Proofs for Section \ref{sec_news}}\label{sec_proof_news}
The proofs presented in this appendix take an arbitrary policy profile $\mathbf{a}=(-a,a)$ with $a > 0$ as given. Since the underlying state is binary, we can represent any signal structure by the tuple $(\pi_z, \mu_z)_{z \in \mathcal{Z}}$, where $\pi_z$ denotes the probability that the signal realization is $z \in \mathcal{Z}$, and $\mu_z$ denotes the posterior mean of the state conditional on the signal realization being $z$. 
Any binary signal structure must satisfy 
\[
\pi_L=\frac{\mu_R}{\mu_R-\mu_L} \text{ and } \pi_R=\frac{-\mu_L}{\mu_R-\mu_L},
\]
and so can be represented by the profile $\bm\mu= (\mu_L, \mu_R)$ of posterior means. Type $k$ voters' utility gain from consuming $\bm\mu$ is simply  
\[
V\left(\bm\mu; {\bf{a}}, k\right)=\begin{cases}
\pi_R\left[v\left({\bf{a}}, k\right)+ \mu_R \right]^+ &\text{ if } k \leq 0,\\
-\pi_L \left[v\left({\bf{a}}, k\right)+ \mu_L \right]^- &\text{ if } k >0,
\end{cases}
\]
where  $v(\mathbf{a},1)>0=v(\mathbf{a},0)>v(\mathbf{a}, -1)=-v(\mathbf{a},1)$ according to Observation \ref{obs_utility} \textbf{symmetry} and \textbf{increasing differences}. For ease of notation, we shall write $v$ for $v(\mathbf{a},1)$ and $-v$ for  $v(\mathbf{a},-1)$, and drop the notation of $\mathbf{a}$ from $V(\Pi; \mathbf{a}, k)$. We will also use $\overline\Pi$ to denote the signal that fully reveals the true state.

We presented two results: Theorems \ref{thm_binary} and \ref{thm_product},  in Section \ref{sec_news}. Given how interrelated these results are, we feel it is best to prove them together.  In what follows,  we will maintain Assumption \ref{assm_attention} and Assumption \ref{assm_regularity}(i) (i.e.,  the feasibility  condition)  throughout. Additional assumptions, such as Assumption \ref{assm_regularity}(ii) and (iii), will only be invoked for certain parts of the proof.  

\paragraph{Personalized case} Consider w.l.o.g. the market segment that consists of left-leaning voters. Any optimal personalized signal for these voters must solve 
\begin{align}\label{eqn_personalized}
\max_{\Pi: \Omega \rightarrow \Delta\left(\mathcal{Z}\right)}   I\left(\Pi\right) \text{ subject to } V\left(\Pi;-1\right) \geq  \lambda  I\left(\Pi\right).
\end{align} 
Let $\gamma \geq 0$ denote the Lagrange multiplier associated with voters'  participation constraint, and define the Lagrangian function as 
\[\mathcal{L}(\Pi, \gamma)=I(\Pi)+\gamma[V(\Pi; -1)-\lambda  I(\Pi)].\]
Then the primal problem: (\ref{eqn_personalized}),  can be rewritten as
$
\sup_{\Pi} \inf_{\gamma \geq 0} \mathcal{L}(\Pi, \gamma), 
$
whereas the dual problem is 
$
\inf_{\gamma \geq 0} \sup_{\Pi} \mathcal{L}(\Pi, \gamma). 
$
Let $p^*$ and $d^*$ denote the values of the primal problem and dual problem, respectively, and note that $d^* \geq p^*$. Also note that $p^*>0$, as  Assumption \ref{assm_regularity}(i) (i.e., the feasibility  condition) implies that any solution to the primal problem must be a nondegenerate signal. 

\subparagraph{Step 1.} Characterize the solution to $\sup_{\Pi}\mathcal{L}(\pi, \gamma)$ for each $\gamma \geq 0$. Show that the solution is always unique and has two signal realizations. 

 When $\gamma=0$, we have $\mathcal{L}(\Pi, \gamma)=I(\Pi)$, so the solution to $\sup_{\Pi}\mathcal{L}(\Pi, \gamma)$ is $\overline\Pi$.   When $\gamma >0$, we can rewrite $\sup_{\Pi}\mathcal{L}(\pi, \gamma)$ as $\sup_{\Pi}\gamma[V(\Pi; -1)-(\lambda-1/\gamma)I(\Pi)]$, or equivalently
\begin{equation}\label{eqn_personalized_simplified}
\sup_{\Pi}V(\Pi; -1)-(\lambda-1/\gamma)I(\Pi). 
\end{equation}
When $\lambda(\gamma) \coloneqq \lambda-1/\gamma \leq 0$, the solution to (\ref{eqn_personalized_simplified}) is again $\overline{\Pi}$. When $\lambda(\gamma)>0$, the maximand of (\ref{eqn_personalized_simplified}) becomes
\[\sum_{z \in \mathcal{Z}} \pi_z \underbrace{\left[[-v+\mu_z]^+ - \lambda\left(\gamma\right) h\left(\mu_z\right) \right]}_{f(\mu_z)},\]
where $f$ is the maximum of two strictly concave functions of $\mu$: $-\lambda\left(\gamma\right) h\left(\mu\right)$ and $-v+\mu_z -\lambda\left(\gamma\right) h\left(\mu\right)$. As depicted in Figure \ref{fig_p}, the two functions single-cross at $\mu=v$, and their maximum is M-shaped.  Thus solving (\ref{eqn_personalized_simplified}) using the concavification method developed by \cite{bayesianpersuasion} yields a unique solution with at most two signal realizations. Hereinafter we shall denote this solution by $\Pi(\gamma)$. 
\begin{figure}[!h]
  \centering
  \includegraphics[width=8cm]{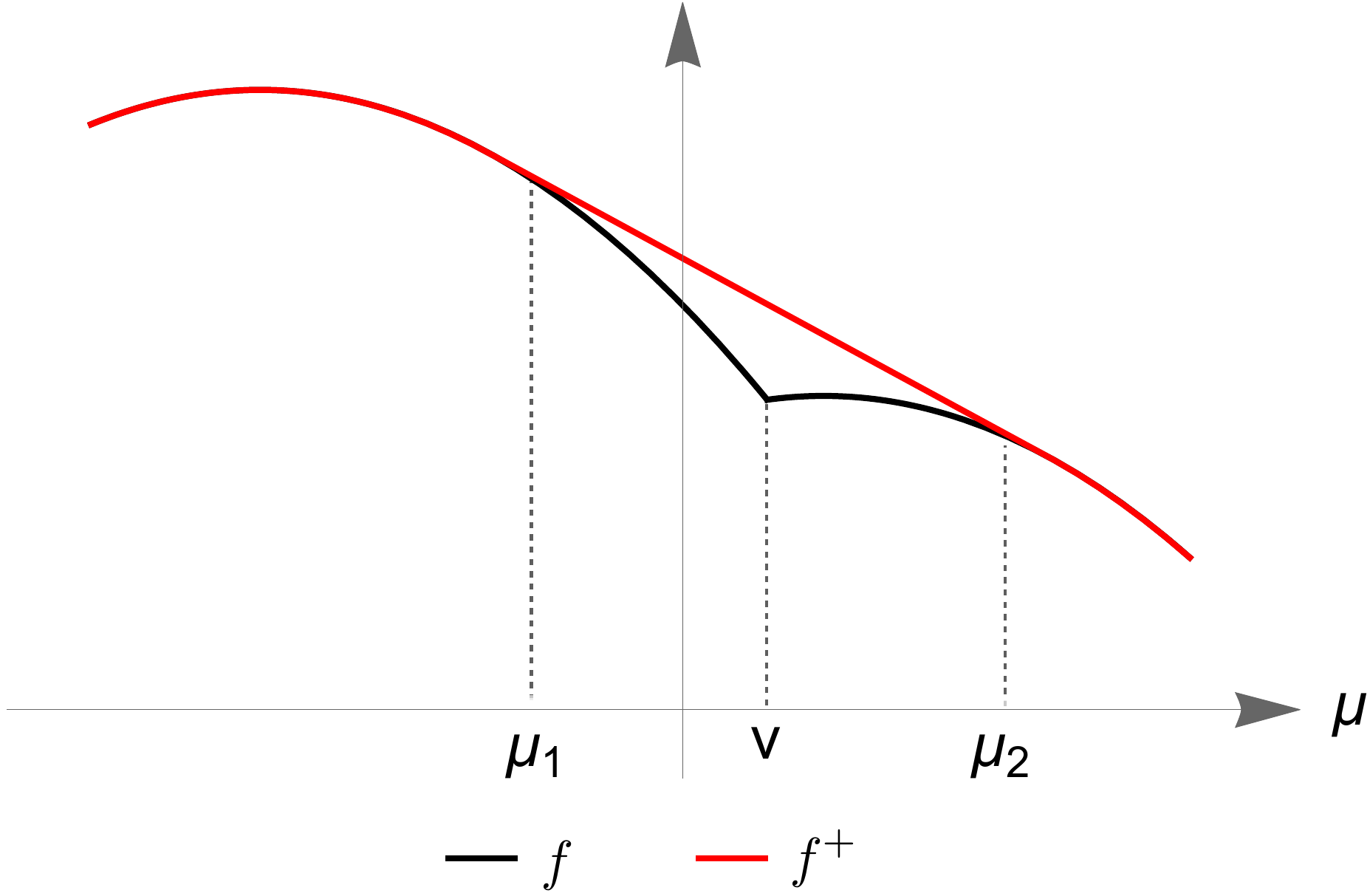}
  \caption{$f$ and its concave closure $f^+$ in the personalized case.}\label{fig_p}
\end{figure}

To rule out the case of one signal realization (i.e., a degenerate signal), let $W(\gamma)$ denote the value of (\ref{eqn_personalized_simplified}) at $\Pi=\Pi(\gamma)$. At $\gamma=+\infty$, (\ref{eqn_personalized_simplified}) becomes the optimal information acquisition problem faced by left-leaning voters if they could freely choose the signal structure themselves, i.e., $\max_{\Pi}V(\Pi; -1)-\lambda I(\Pi)$. Assumption \ref{assm_regularity}(i) implies that $\Pi(+\infty)$ is nondegenerate, and that $W(+\infty)>0$.  Meanwhile,   $W'(\gamma)=-I(\Pi(\gamma))/\gamma^2 \leq 0$ holds almost surely by the envelope theorem. Combining these observations shows that $W(\gamma)\geq W(+\infty)>0$ and, hence, $
\Pi(\gamma)$ is nondegenerate for all $\gamma \in (1/\lambda, +\infty)$, 

\subparagraph{Step 2.} Show that strong duality holds, i.e., $p^*=d^*$. Recall from Step 1 that  $I(\overline\Pi)=\mathcal{L}(\Pi(0), 0)\geq  d^* \geq p^* > 0$, i.e., $d^*$ is a positive, finite number. Also note that $\mathcal{L}(\Pi(\gamma),\gamma)=\gamma W(\gamma)$, and that $W(\gamma)$ is continuous in $\gamma$ by Berge's maximum theorem. From $W(+\infty)>0$, it follows that $\lim_{\gamma \rightarrow +\infty} \mathcal{L}(\Pi(\gamma), \gamma)=+\infty \cdot W(+\infty)=+\infty > d^*$. Thus, $d^*$ must be attained at some finite $\gamma^* \in [0, +\infty)$. 

Define $\Gamma=\{\gamma \geq 0: \Pi(\gamma)=\overline{\Pi}\}$. Note that $\Gamma \supseteq [0, 1/\lambda]$, and that the set is closed because $\Pi(\gamma)$ is the unique solution to  (\ref{eqn_personalized_simplified}) and so is continuous in $\gamma$ by Berge's maximum theorem.   Consider two situations.

$\bullet$ If $V(\overline{\Pi})-\lambda I(\overline{\Pi}) \geq 0$, then $\overline{\Pi}$ is primal feasible, and so $p^* \geq I(\overline{\Pi})=\mathcal{L}(\Pi(0), 0) \geq d^* \geq p^*$ as desired. 

$\bullet$ If $V(\overline{\Pi})-\lambda I(\overline{\Pi})<0$, then $\sup \Gamma <+\infty$. Thus $\Gamma$ is a compact set, hence every open cover $\cup_{\gamma \in \Gamma}\mathcal{B}_{\epsilon}(\gamma)$ of it has a finite subcover,  which we denote by $\Gamma_{\epsilon}$.\footnote{Throughout, $\mathcal{B}_{\epsilon}(x)$ denotes the $\epsilon$-ball surrounding $x$.} By the envelop theorem, $\mathcal{L}(\Pi(\gamma), \gamma)$ is differentiable in $\gamma$ almost surely, and the derivative, whenever it exists, is given by 
\[\frac{d}{d\gamma} \mathcal{L}(\Pi(\gamma), \gamma)=V(\Pi(\gamma))-\lambda I(\Pi(\gamma)). \]
Now, since the right-hand side of the above expression is continuous in $\gamma$, $\mathcal{L}(\Pi(\gamma), \gamma)$ must be differentiable in $\gamma$ everywhere rather than being just absolutely continuous in $\gamma$. Moreover, $\frac{d}{d\gamma}\mathcal{L}(\Pi(\gamma), \gamma)$ must be negative on $\Gamma_{\epsilon}$ when $\epsilon$ is sufficiently small, hence $\mathcal{B}_{\epsilon'}(\gamma) \subseteq  \mathbb{R}_+-\Gamma$ for some $\epsilon'>0$. At $\gamma=\gamma^*$ where $d^*$ is attained, 
\[\frac{d}{d\gamma}\mathcal{L}(\Pi(\gamma), \gamma)\bigg|_{\gamma=\gamma^*}=V(\Pi(\gamma^*); -1)-\lambda I(\Pi(\gamma^*))=0\]
must hold, which together with $0 \notin \mathcal{B}_{\epsilon}(\gamma^*)$ implies that $(\gamma^*, \Pi(\gamma^*))$ satisfies the complementary slackness condition and so is primal feasible. Thus $p^* \geq I(\Pi^*)$, which together with $d^*=\mathcal{L}(\Pi(\gamma^*), \gamma^*)=I(\Pi^*)$ implies that $p^*=d^*$ as desired.

\subparagraph{Step 3.} Show that the primal problem admits a unique solution. From Steps 1 and 2, we know that the primal problem either admits $\overline{\Pi}$ as the unique solution, or any solution to it must take the form of $\Pi(\gamma)$ for some $\gamma>1/\lambda$. In the first case, we are done. In the second case, suppose, to the contrary, that the primal problem admits two distinct solutions: $\Pi^1=\Pi(\gamma^1)$ and $\Pi^2=\Pi(\gamma^2)$, with $\gamma^1 \neq \gamma^2$. Assume w.l.o.g. that $\gamma^1>\gamma^2$, and hence that $\lambda(\gamma^1)=\lambda-1/\gamma^1>\lambda-1/\gamma^2=\lambda(\gamma^2) (>0)$. Since $\Pi^i$ is the unique solution to (\ref{eqn_personalized_simplified}) given $\lambda(\gamma^i)$, 
\[
\lambda(\gamma^1)[I\left(\Pi^1\right)-I(\Pi^2)]>V(\Pi^1; -1)-V(\Pi^2; -1)>\lambda(\gamma^2)[I(\Pi^1)-I(\Pi^2)]
\]
must hold. Therefore $I(\Pi^1)>I(\Pi^2)$, and so $\Pi^2$ cannot be a solution to the primal problem, a contradiction. 

\subparagraph{Step 4.} Show that if the solution to the primal problem differs from $\overline{\Pi}$ (as required by Assumption \ref{assm_regularity}(ii)), then it must be left-skewed, i.e., $\mu_L+\mu_R>0$. 
 From the previous steps, we know that 
 \[V\left((\mu_L, \mu_R); -1\right)=\frac{\mu_R}{\mu_R-\mu_L}\left[-v+
\mu_R\right]^+ \geq \lambda I\left((\mu_L, \mu_R)\right)>0\] and hence that $\mu_R>v$.  $\mu_L+\mu_R \geq 0$ must hold, because if the contrary is true, i.e., $\mu_L+\mu_R<0$, then consuming $(-\mu_R, -\mu_L)$ incurs the same attention cost as consuming $(\mu_L, \mu_R)$ by Assumption \ref{assm_attention}, but the first option generates a strictly higher consumption utility:  
\begin{multline*}
V\left((-\mu_R, -\mu_L); -1\right)
=\frac{\mu_R}{\mu_R-\mu_L}\left(-v-\mu_L\right)
>\frac{-\mu_L}{\mu_R-\mu_L}\left(-v+\mu_R\right)
=V\left((\mu_L, \mu_R); -1\right).
\end{multline*} 
It remains to show that $\mu_L+\mu_R \neq 0$. For starters, recall that if $(\mu_L, \mu_R) \neq (-1,1)$ as required by Assumption \ref{assm_regularity}(ii), then it must solve
\begin{equation}\label{eqn_personalized_relaxed+}
\max_{\substack{(\mu_L, \mu_R) \in \\ \left[-1,0\right] \times \left[v,1\right]}} \frac{-\mu_L}{\mu_R-\mu_L} (-v+\mu_R)-\lambda(\gamma^*)\left[\frac{\mu_R}{\mu_R-\mu_L}h\left(\mu_L\right)-\frac{\mu_L}{\mu_R-\mu_L} h\left(\mu_R\right)\right]
\end{equation}
for some $\gamma>1/\lambda$. There are two cases to consider. 

$\bullet$ If $\mu_R=1$, then $\mu_L \neq -1$ must hold in order to satisfy $(\mu_{L}, \mu_R)\neq (-1,1)$,  and so $\mu_R+\mu_L>1$ as desired.

$\bullet$ If $\mu_R \in (v,1)$, then $\mu_L >-1$ must hold in order to satisfy $\mu_L+\mu_R \geq 0$, and $\mu_L \neq 0$ must hold to prevent the signal from becoming degenerate. Thus $(\mu_L, \mu_R)$ belongs to the interior of $[-1,0] \times [v,1]$ and so is characterized by the following system of first-order conditions:  
\begin{align}
-v+\mu_R &=\lambda(\gamma)\left[\Delta h -h'\left(\mu_L\right)\Delta \mu\right]\label{eqn_foc2} \\
\text{ and } v-\mu_L &=\lambda(\gamma)\left[h'\left(\mu_R\right)\Delta \mu-\Delta h\right]. \label{eqn_foc1}
\end{align}
Now, if $\mu_L+\mu_R=0$, then $\Delta h=0$ and $h'\left(\mu_R\right)=-h'\left(\mu_L\right)$ by Assumption \ref{assm_attention}, so the right-hand sides of (\ref{eqn_foc2}) and (\ref{eqn_foc1}) are the same. Meanwhile, the left-hand sides differ, which leads to a contradiction.


\paragraph{Broadcast case} Due to space constraints, we focus on the case where it is strictly optimal to include all voters in news consumption (i.e., Assumption \ref{assm_regularity}(iii) holds) for now,  and will comment on the remaining cases at the end of this section. Since the following must hold for any nondegenerate signal structure $\Pi$: 
\begin{align*}
&V\left(\Pi; -1\right)=\sum_{z \in \mathcal{Z}}\pi_z [-v+\mu_z]^+<\sum_{z \in \mathcal{Z}}\pi_z [\mu_z]^+=  V\left(\Pi; 0\right)\\
\text{ and } &V\left(\Pi; 1\right) =\sum_{z \in \mathcal{Z}}-\pi_z [v+\mu_z]^-<\sum_{z \in \mathcal{Z}}-\pi_z [\mu_z]^- = V\left(\Pi; 0\right), 
\end{align*}
we can ignore centrist voters' participation constraint and formalize the primal problem as follows: 
\begin{equation}\label{eqn_broadcast}
\max_{\mathcal{Z}, \Pi: \{-1,1\}\rightarrow \Delta(\mathcal{Z})} I(\Pi) \text{ subject to } V(\Pi;k)\geq \lambda I(\Pi) \text{ } \forall k\in \{-1,1\}.
\end{equation}
For each $k \in \{-1,1\}$, let $\gamma_k \geq 0$ denote the Lagrange multiplier associated with type $k$ voters' participation constraint,  and define the Lagrangian function as 
\[\mathcal{L}(\Pi, \bm\gamma)=I(\Pi)+\sum_{k\in \{-1,1\}}\gamma_k[V(\Pi; k)-\lambda I(\Pi)],\]
where $\bm\gamma=(\gamma_{-1}, \gamma_1)$.  The primal problem: (\ref{eqn_broadcast}),  can be rewritten as $\sup_{\Pi}\inf_{\bm\gamma \geq \bm0} \mathcal{L}(\Pi, \bm\gamma)$,  whereas the dual problem is $\inf_{\bm\gamma \geq \bm0} \sup_{\Pi}\mathcal{L}(\Pi, \bm\gamma)$. As before, let $p^*$ and $d^*$ to denote the values of the primal problem and dual problem, respectively. Note that $d^* \geq p^*$, and that $p^*>0$ by Assumption \ref{assm_regularity}(i).  The remainder of the proof consists of three steps.  

\subparagraph{Step 1.} Characterize the solution(s) to  $\sup_{\Pi} \mathcal{L}(\Pi, \bm\gamma)$ for each $\bm\gamma \geq \bm 0$. When $\bm\gamma=\bm 0$, $\mathcal{L}(\Pi; \bm\gamma)=I(\Pi)$,  so the solution to $\sup_{\Pi} \mathcal{L}(\Pi, \bm\gamma)$ is $\overline{\Pi}$. For each $\bm \gamma \neq \bm 0$, define (i) $\bm{\delta}(\bm{\gamma})\coloneqq (\delta_{-1}(\bm\gamma), \delta_1(\bm\gamma))$, where $\delta_k(\bm\gamma)\coloneqq \gamma_k(\gamma_{-1}+\gamma_1)^{-1}$ $\forall k  \in \{-1,1\}$, (ii) $\lambda(\bm{\gamma})\coloneqq\lambda-(\gamma_{-1}+\gamma_{1})^{-1}$, as well as (iii) 
\[W(\Pi, \bm\gamma)\coloneqq \bm{\delta}(\bm\gamma) \cdot \bm{V}(\Pi)-\lambda(\bm{\gamma}) I(\Pi)\]
where $\bm{V}(\Pi)\coloneqq (V(\Pi; -1), V(\Pi; 1))$. 
Then
$\mathcal{L}(\Pi, \bm\gamma)=(\gamma_{-1}+\gamma_1) W(\Pi, \bm\gamma)$, and the problem $\sup_{\Pi} \mathcal{L}(\Pi, \bm\gamma)$ is equivalent to
\begin{equation}\label{eqn_broadcast_simplified}
\sup_{\Pi} W(\Pi, \bm\gamma).
\end{equation}
If $\lambda(\bm\gamma)\leq 0$, then the solution to (\ref{eqn_broadcast_simplified}) is again $\overline{\Pi}$. If $\lambda(\bm{\gamma})>0$, then the maximand of (\ref{eqn_broadcast_simplified}) is
\[
\sum_{z \in \mathcal{Z}} \pi_z  \underbrace{\left[\delta_{-1} [-v+\mu_z]^+ - \delta_{1}[v+\mu_z]^- -\lambda(\bm\gamma) h(\mu_z)\right]}_{f(\mu_z)},
\]
where $f$ is the maximum of three strictly concave functions of $\mu$: $\delta_{-1}\left(-v+\mu\right)-\lambda(\bm\gamma) h\left(\mu\right)$, $-\lambda(\bm\gamma) h\left(\mu\right)$, and $-\delta_{1} \left(v+\mu\right)-\lambda(\bm\gamma) h\left(\mu\right)$, as depicted in Figure \ref{fig_uf}. Let $f^+$ denote the concave closure of $f$, and note that $\mu_1 \coloneqq \inf\left\{\mu: f^+\left(\mu\right)>f\left(\mu\right)\right\}$ and $\mu_2 \coloneqq \sup\left\{\mu: f^+\left(\mu\right)>f\left(\mu\right)\right\}$ are finite and satisfy $\mu_1<0<\mu_2$. There are three cases to consider. 
\begin{figure}[!h]
\centering
\begin{subfigure}{.5\linewidth}
  \includegraphics[height=5cm]{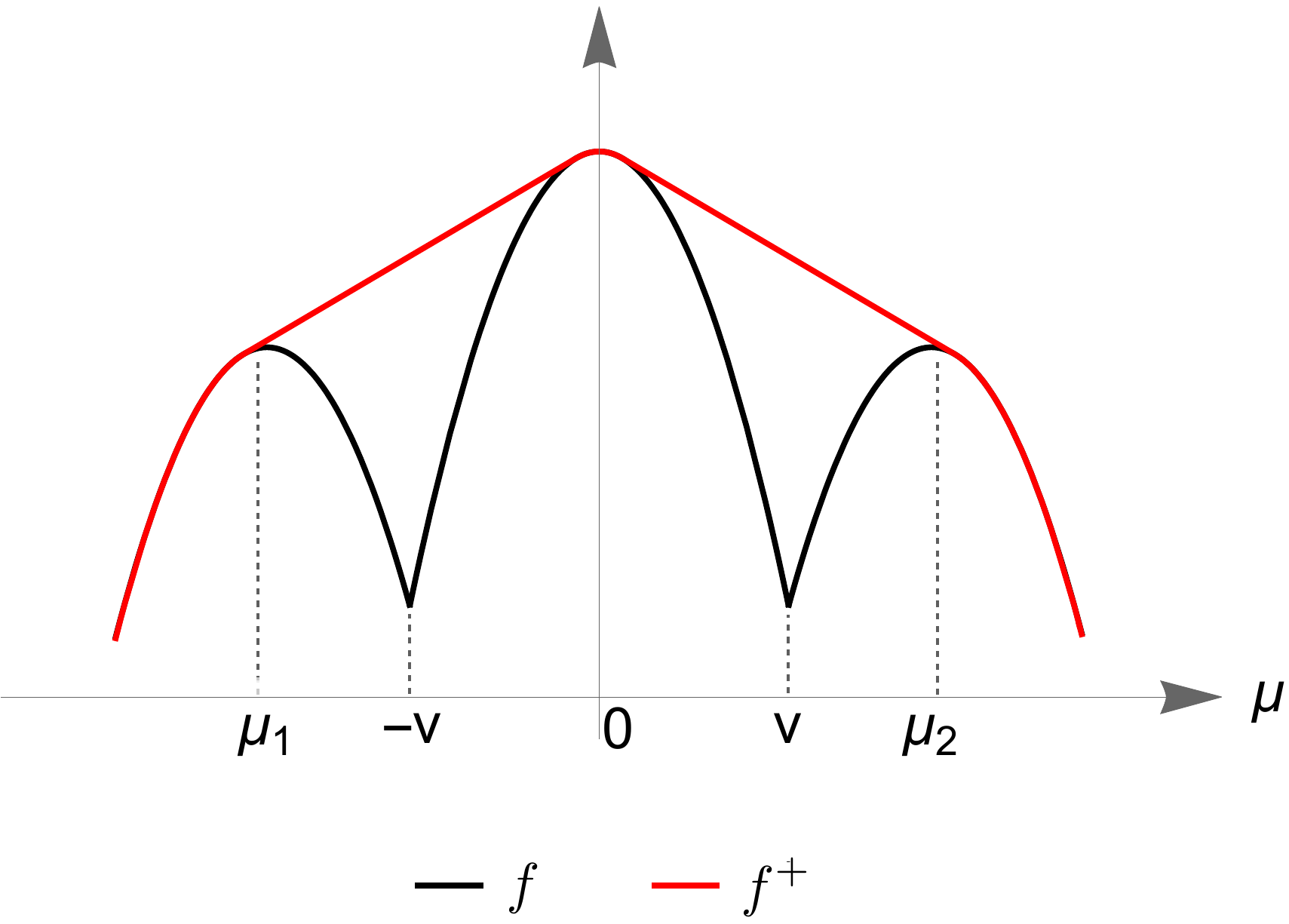}
  \caption{}
\end{subfigure}%
\begin{subfigure}{.5\textwidth}
  \includegraphics[height=5cm]{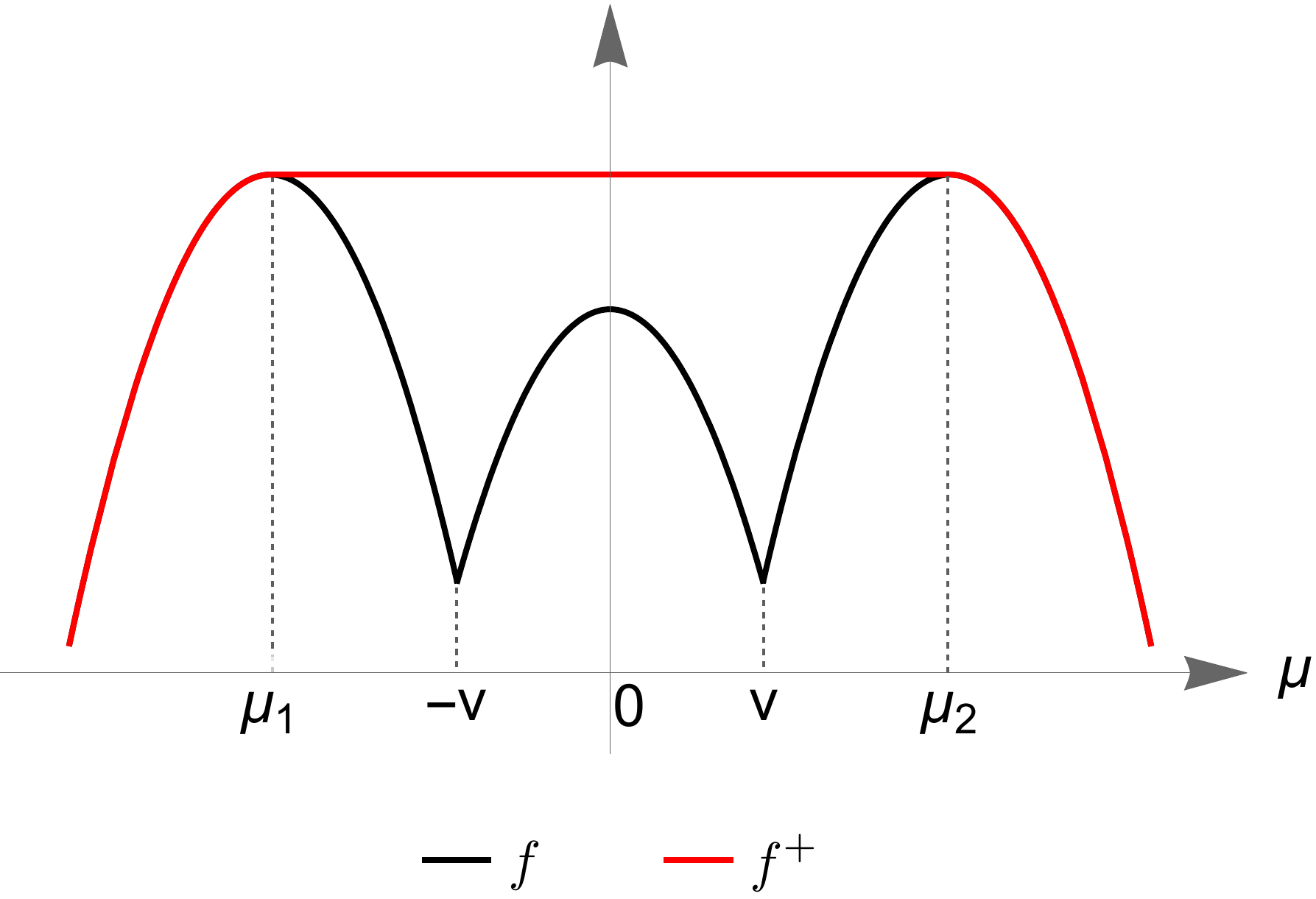}
  \caption{}
\end{subfigure}
\\
\begin{subfigure}{.5\textwidth}
  \includegraphics[height=5cm]{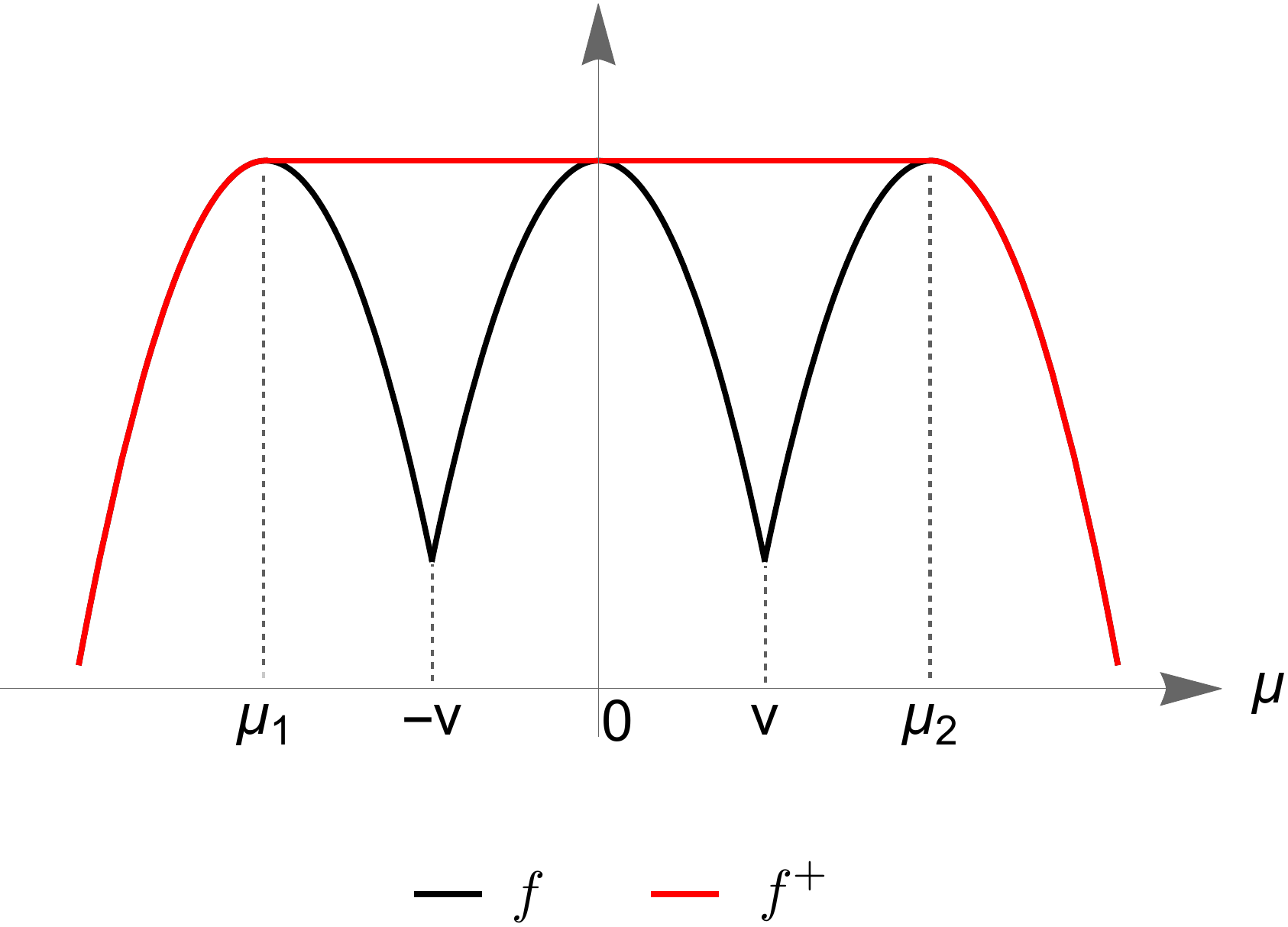}
  \caption{}
\end{subfigure}
\caption{$f$ and its concave closure $f^+$ in the broadcast case.}
\label{fig_uf}
\end{figure}
\begin{enumerate}[(a)]
\item If $f^+\left(0\right) > (1-\alpha) f^+(\mu_1)+\alpha f^+(\mu_2)$ for all $\alpha \in [0,1]$, then the unique solution to (\ref{eqn_broadcast_simplified}) is the degenerate signal $(0, 0)$ (as depicted on Panel (a) of Figure \ref{fig_uf}). 

\item If $f^+\left(0\right)=(1-\alpha) f^+(\mu_1)+\alpha f^+(\mu_2)>f(0)$ for some $\alpha \in [0,1]$, then the unique solution to (\ref{eqn_broadcast_simplified}) is the binary signal $(\mu_1, \mu_2)$ (as depicted on Panel (b) of Figure \ref{fig_uf}). 
\item $f^+\left(0\right)=(1-\alpha) f^+(\mu_1)+\alpha f^+(\mu_2)=f(0)$ for some $\alpha \in [0,1]$. In this case, (\ref{eqn_broadcast_simplified}) has multiple solutions that entail at most three signal realizations (as depicted on Panel (c) of Figure \ref{fig_uf}).   Among these solutions, the binary signal $(\mu_1, \mu_2)$ is the most Blackwell-informative, and the degenerate signal $(0,0)$ is the least Blackwell-informative. 
\end{enumerate}

\subparagraph{Step 2.} Show that strong duality holds, i.e., $p^*=d^*$, and that $p^*$ and $d^*$ are attained by a binary signal structure. 

For each $\bm\gamma \geq \bm 0$, let $\Pi(\bm\gamma)$ denote a typical solution to  (\ref{eqn_broadcast_simplified}). By Berge's maximum theorem, the correspondence $\bm\gamma \mapsto \{\Pi(\bm\gamma)\}$ is upperhemicontinuous (uhc), and $W(\bm\gamma)$ is continuous in $\bm\gamma$. Recall, from the previous step, that $I(\overline{\Pi})=\mathcal{L}(\overline{\Pi}, \bm 0) \geq d^*\geq p^*>0$, i.e., $d^*$ is a positive, finite number. To show that $d^*$ is attained by some finite $\bm\gamma^*$, suppose that the contrary is true, i.e., there exists a sequence $(\bm\gamma^n, \Pi(\bm\gamma^n))_n$ such that $\|\bm\gamma^n\| \rightarrow +\infty$ ($\|\cdot\|$ denotes the sup norm) and $\lim_{n \rightarrow +\infty}\mathcal{L}(\Pi(\bm\gamma^n), \bm\gamma^n) =d^*$. Since $\bm\delta (\bm\gamma^n) \in LS \coloneqq \{\bm\delta \in [0,1]^2: \delta_{-1}+\delta_1=1\}$ $\forall n$, there exists a subsequence $(\bm\gamma^{n_m})_{n^m}$ such that $\lim_{n^m \rightarrow +\infty}\bm\delta(\bm\gamma^{n_m})=\bm\delta'\in LS$. In the limit, (\ref{eqn_broadcast_simplified}) becomes the optimal information acquisition problem faced by a representative voter who weighs the two extreme voters by $\bm\delta'$ and faces $\lambda$ as the attention cost parameter, i.e., $\max_{\Pi} \bm\delta' \cdot \bm{V}(\Pi)-\lambda I(\Pi)$. Under Assumption \ref{assm_regularity}(i), the value of this optimization problem is strictly positive. But then 
$d^*=\lim_{n^m \rightarrow +\infty} \mathcal{L}(\Pi(\bm\gamma^{n_m}), \bm\gamma^{n_m})=\lim_{n^m \rightarrow +\infty}(\gamma_{-1}^{n_m}+\gamma_1^{n_m}) W(\Pi(\bm\gamma^{n_m}), \bm\gamma^{n_m})=+\infty > d^*$, which is impossible.

Define $\Gamma=\{\bm\gamma \geq \bm 0: \overline{\Pi} \in \{\Pi(\bm\gamma)\}\}$. Note that $\Gamma \supseteq \{\bm\gamma \geq \bm 0: \gamma_{-1}+\gamma_1 \leq 1/\lambda\}$, and that the set is closed by the upper hemicontinuity of the correspondence $\bm\gamma \mapsto \{\Pi(\bm\gamma)\}$.  By symmetry, either $V(\overline{\Pi}; k)-\lambda I(\overline{\Pi}) \geq 0$ $\forall k \in \{-1,1\}$, or $V(\overline{\Pi}; k)-\lambda I(\overline{\Pi}) <0$ $\forall k \in \{-1,1\}$.

$\bullet$ In the first case, $\overline{\Pi}$ is the unique solution to the primal problem, and $p^* \geq I(\overline{\Pi}) \geq d^*\geq p^*$ as desired. 

$\bullet$ In the second case,  $\sup\{\|\bm\gamma\|: \bm\gamma \in \Gamma\}<\infty$ must hold, because if the contrary is true, then there exists a sequence $(\bm\gamma^n)_n \in \Gamma$ with $\lim_{n\rightarrow +\infty} \|\bm\gamma^n\|=+\infty$. Let $(\bm \gamma^{n_m})_{n^m}$ be any subsequence of $(\bm\gamma^n)_n$ such that $\delta(\bm\gamma^{n_m})$ converges (to an element $\bm\delta'$ of LS). Then $\lim_{n^m \rightarrow +\infty}\bm\delta (\bm\gamma^{n_m})\cdot \bm{V}(\overline{\Pi})-\lambda (\bm\gamma^{n_m}) I(\overline{\Pi})=\bm\delta'\cdot \bm{V}(\overline{\Pi})-\lambda I(\overline{\Pi}) \geq 0$, which contradicts the fact that $V(\overline{\Pi}, k)-\lambda I(\overline{\Pi})<0$ $\forall k\in \{-1,1\}$. 

The remainder of the proof focuses on the second case, in which $\Gamma$ is compact, and  $\bm\gamma^* \in \mathbb{R}_+^2 -\Gamma$ and so must differ from $\bm 0$. Indeed, we claim that $\bm\gamma^* >\bm 0$, and so $\mathcal{B}_{\epsilon}(\bm\gamma^*) \subseteq \mathbb{R}_+^2-\Gamma$ for some $\epsilon>0$. To prove this claim, all we need to do is to rule out the case where $\gamma_k^*=0$ and $\gamma_{-k}^*>0$ for some $k \in \{-1,1\}$. In that case, $d^*$ is attained by the optimal personalized signal for type $-k$ voters. Since the infomediary could always implement that  signal structure but strictly prefers to include all voters in news consumption, $p^*>d^*$ must hold, which is impossible. 

At $\bm\gamma=\bm\gamma^*$, $W(\bm\gamma^*)=0$ if we end up in the situation depicted on Panel (a) or Panel (c) of Figure \ref{fig_uf}, in which the degenerate signal is a solution to (\ref{eqn_broadcast_simplified}). But then $d^*=(\gamma_1^*+\gamma_{-1}^*)W(\bm\gamma^*)=0<p^*$, which is impossible. The only remaining possibility is depicted on Panel (b) of Figure \ref{fig_uf}, whereby $\{\Pi(\bm\gamma^*)\}$ is a singleton, and $\Pi(\bm\gamma^*)$ is a binary signal.  
Since the mapping $\bm\gamma \mapsto \{\Pi(\bm\gamma)\}$ is uhc, there exists a small neighborhood $\mathcal{B}_{\epsilon}(\bm\gamma^*)$, over which  
$\{\Pi(\bm\gamma)\}$ is a singleton, $\Pi(\bm\gamma)$ is binary, and the mapping $\bm\gamma \mapsto \Pi(\bm\gamma)$ is continuous.  By the envelope theorem, $\mathcal{L}(\Pi(\bm\gamma), \bm\gamma)$ is differentiable in $\bm\gamma$ almost surely, and the derivative, whenever it exists, is given by 
\[\frac{\partial}{\partial \gamma_k} \mathcal{L}(\Pi(\bm\gamma), \bm\gamma)=V(\Pi(\bm\gamma), k)-\lambda I(\Pi(\bm\gamma)) \text{ } \forall k \in \{-1,1\}.\]
Now, since the right-hand side of the above expression is continuous $\bm\gamma$ over $\mathcal{B}_{\epsilon}(\bm\gamma^*)$, $\mathcal{L}(\Pi(\bm\gamma), \bm\gamma)$ must be differentiable in $\bm\gamma$, rather than being just absolutely continuous in $\bm\gamma$, over $\mathcal{B}_{\epsilon}(\bm\gamma^*)$.  Thus at $\bm\gamma=\bm\gamma^*$, 
\[\frac{\partial}{\partial \gamma_k} \mathcal{L}(\Pi(\bm\gamma), \bm\gamma)\bigg|_{\bm\gamma=\bm\gamma^*}=V(\Pi(\bm\gamma^*), k)-\lambda I(\Pi(\bm\gamma^*))=0 \text{ } \forall k \in \{-1,1\}\]
must hold, which together with $\bm\gamma^*>\bm 0$ implies that $(\bm\gamma^*, \Pi(\bm\gamma^*))$ satisfies the complementary slackness condition and is so primal feasible, i.e.,  $p^* \geq I(\Pi (\bm\gamma^*))$. Combining this finding with $d^*=\mathcal{L}(\Pi(\bm\gamma^*), \bm\gamma^*)=I(\Pi(\bm\gamma^*))$ yields $d^*=p^*$ as desired.

\subparagraph{Step 3.} Show that the optimal broadcast signal is unique and, indeed, symmetric. As demonstrated in the previous steps, the solution to the primal problem is either $\overline{\Pi}$, or it takes the form of $\Pi(\bm\gamma^*)$ for some $\bm\gamma^* >\bm 0$.  In the second case, both types of extreme voters have binding participation constraints, and so 
\begin{align*}
\lambda I\left((\mu_L, \mu_R)\right)=V\left((\mu_L, \mu_R); -1 \right)&=-\frac{\mu_L}{\mu_R-\mu_L}[-v+\mu_R]^+\\
&=V\left((\mu_L, \mu_R); 1\right)=-\frac{\mu_R}{\mu_R-\mu_L}[v+\mu_L]^-.
\end{align*}
Simplifying the above expression yields $\mu_L=-\mu_R$, so $(\mu_L, \mu_R)$  is symmetric, and $\mu_L$ solves
\begin{equation}\label{eqn_musymmetry}
\max_{\mu\in [-1,0]} h\left(\mu\right) \text{ s.t. } \frac{1}{2}\left[-v-\mu\right]^+ \geq \lambda h\left(\mu\right).
\end{equation}
Since $h$ is strictly convex and strictly decreasing on $\left[-1, 0\right]$, (\ref{eqn_musymmetry}) admits a unique solution (draw a picture yourself). 

\paragraph{Remarks} First, it is easy to see that any optimal signal must induce strict obedience from its consumers. No more proof is required than the verbal argument made in the main text. 

Second, in the case where it is optimal to induce some, but not all voters to consume news, the only possibilities are: (i) include only the centrist voters in news consumption; (ii) include centrist voters and one type of extreme voters in news consumption.  Regardless of which situation we end up with,  only one type of voters has a binding participation constraint,  and the optimal broadcast signal constitutes the optimal personalized signal for these voters.

Finally,  while solving the model in closed form is in general challenging (due to the endogeneity of the Lagrange multipliers associated with voters' participation constraints), progresses can be made by restricting attention to the case of quadratic attention cost.  In that case,  the beliefs induced by the solutions to  (\ref{eqn_personalized_simplified}) and (\ref{eqn_broadcast_simplified}) are 
\[\begin{cases}
(4\min\{a, t(1)\}-\frac{1}{2\lambda}, \frac{1}{2\lambda}) &\text{ if } k=-1,\\
(-\frac{1}{2\lambda}, \frac{1}{2\lambda}) & \text{ if } k=0,\\
(-\frac{1}{2\lambda}, \frac{1}{2\lambda}-4\min\{a, t(1)\}) & \text{ if } k=1,
\end{cases}\]
and 
\[(-\frac{1+\sqrt{1-16\lambda \min\{a, t(1)\}}}{4\lambda},\frac{1+\sqrt{1-16\lambda \min\{a, t(1)\}}}{4\lambda}), \]
respectively, when the underlying policy profile is $(-a,a)$.  These beliefs take values in $(-1,0) \times (0,1)$ if and only if $\lambda> 1/2$ and $\lambda t(1) <  1/16=.0625$. Further algebra shows that it is strictly optimal to include all voters in news consumption under broadcast news aggregation if and only if $\lambda t(1)< 3\sqrt{2}/4-1\approx .060$.  These results are used in Example \ref{exm_quadratic} to examine the comparative statics of equilibrium policies.  \qed

\subsection{Proofs for Sections \ref{sec_policy} and \ref{sec_cs}}\label{sec_proof_policy}
\paragraph{Proof of  Lemma \ref{lem_maxmin}} 
In the case of personalized news aggregation with $q(0) \leq 1/2$,  consider any unilateral deviation of candidate $R$ from a policy profile $(-a,a)$ with $a \geq 0$  to $a'$. Below we demonstrate that the deviation is unprofitable if and only if (i) $a' \notin \left[-a,a\right]$, or (ii) $a' \in [-a,a)$, and it doesn't attract any type $k$ voter with $t(k) \in [-a,a]$.

\paragraph{Step 1.} Show that no $a' >a$ strictly increases candidate $R$'s winning probability. Fix any $a'>a$. From Observation \ref{obs_utility} \textbf{inverted V-shape} and (\ref{eqn_ob}),  it follows that neither left-leaning voters nor centrist voters would find $a'$ attractive, i.e., $\forall k \leq 0$, 
\[
v\left(-a,a', k\right)+\mu_L^{p}\left(a,k\right)<v\left(-a,a,k\right)+\mu_L^{p}\left(a,k\right)<0.
\]
Given this, as well as the symmetry of the joint signal distribution, it suffices to show that if $a'$ attracts right-leaning voters, then it must repel left-leaning voters, i.e., \[v\left(-a,a',1\right)+\mu_L^{p}\left(a,1\right)>0\Longrightarrow v\left(-a,a',-1\right)+\mu_R^{p}\left(a,-1\right)<0.\]
Our argument exploits the symmetry of marginal signal  distributions, i.e., $\mu_R^{p}\left(a,-1\right)=-\mu_L^{p}\left(a,1\right)$, which together with Observation \ref{obs_utility} \textbf{symmetry} implies that 
\begin{align*}
v\left(-a,a', -1\right)+\mu_R^{p}\left(a,-1\right)&\coloneqq u\left(a', -1 \right)-u\left(-a, -1\right)+\mu_R^{p}\left(a, -1\right)\\
&=u\left(-a',1 \right)-u\left(a,1\right)-\mu_L^{p}\left(a, 1\right). 
\end{align*}
Thus if $v\left(-a,a',1\right)+\mu_L^{p}\left(a,1\right)\coloneqq u\left(a',1\right)-u\left(-a,1\right)+\mu_L^{p}\left(a,1\right)>0,$
then
\begin{multline*}
v\left(-a,a', -1\right)+\mu_R^{p}\left(a,-1\right)
=u\left(-a',1\right)-u\left(a,1\right)-\mu_L^{p}\left(a, 1\right)\\
<u\left(a',1\right)+u\left(-a',1\right)-\left[u\left(a,1\right)+u\left(-a,1\right)\right]
\leq 0,
\end{multline*}
where the last inequality follows from Observation \ref{obs_utility} \textbf{concavity}. 

\paragraph{Step 2.} Show that no $a'<-a$ strictly increases candidate $R$'s winning probability. The proof closely parallels that in Step 1. First, note that no $a'<-a$ attracts centrist or right-leaning voters by Observation  \ref{obs_utility} \textbf{inverted V-shape}, i.e., $\forall k \geq 0$,  
\[v\left(-a,a',k\right)+\mu_L^p\left(a,k\right) < v\left(-a,-a,k\right)+\mu_L^p\left(a,k\right)=0+\mu_L^p\left(a,k\right)<0.\] Second, if any $a'$ as above  attracts left-leaning voters, then it must repel right-leaning voters for the reason given in Step 1. Combining these observations gives the desired result. 

\paragraph{Step 3.} Show that no $a' \in [-a,a)$ repels any voter. Fix any $a'$ as such. From Observation \ref{obs_utility} \textbf{inverted V-shape} and (\ref{eqn_ob}), it follows that if $t\left(k\right) \leq a' (<a)$, then
\[v\left(-a,a', k\right)+\mu_R^{p}\left(a,k\right)> v\left(-a,a,k\right)+\mu_R^{p}\left(a,k\right)>0,\]
and if $t\left(k\right)>a'( \geq -a)$, then 
\[v\left(-a,a',k\right)+\mu_R^{p}\left(a,k\right)\geq v\left(-a,-a,k\right)+\mu_R^{p}\left(a,k\right)=0+\mu_R^{p}\left(a,k\right)>0.\]
Combining these observations yields $v\left(-a,a',k\right)+\mu_R^{p}\left(a,k\right)>0$ for any $k$.

\paragraph{Step 4.} Show that no $a' \in [-a,a)$ attracts any voter with $t(k) \notin \left[-a,a\right]$. Fix any $a'$ as such. From Observation  \ref{obs_utility} \textbf{inverted V-shape} and  (\ref{eqn_ob}), it follows that if $t\left(k\right)< -a (\leq a')$, then 
\[v\left(-a,a',k\right)+\mu_L^{p}\left(a,k\right)\leq v\left(-a,-a, k\right)+\mu_L^{p}\left(a,k\right)= 0+\mu_L^{p}\left(a,k\right)<0,\]
and if $t\left(k\right)>a(>a')$, then 
\[v\left(-a,a',k\right)+\mu_L^{p}\left(a,k\right)< v\left(-a,a,k\right)+\mu_L^{p}\left(a,k\right)<0.\]
Combining these observations yields $v\left(-a,a',k\right)+\mu_L^{p}\left(a,k\right)<0$ for any $k$.

\bigskip

Taken together, we conclude that a unilateral deviation from $(-a,a)$ strictly increases candidate $R$'s winning probability if and only if it belongs to $[-a,a)$ and attracts any type $k$ voter with $t(k) \in \left[-a,a\right]$. Ruling out such deviations leads us to sustain the original policy profile in an equilibrium. \qed 

\paragraph{Proof of Lemma \ref{lem_xi}} For every segmentation technology $\mathcal{S}\in \{b,p\}$, define type $k$ voters' \emph{susceptibility} to candidate $R$'s deviation from $(-a,a)$ to $a'$ as 
\[\phi^{\mathcal{S}}(-a,a',k)\coloneqq v(-a,a', k)+\mu_L^{\mathcal{S}}(a,k),\]
and note that  $a'$ attracts these voters if and only if $\phi^{\mathcal{S}}(-a,a',k)>0$. 
The attraction-proof set for type $k$ voters, or simply the \emph{$k$-proof set}, is thus $\Xi^{\mathcal{S}}(k)=\{a \geq 0: \phi^{\mathcal{S}}(-a, t(k), k) \leq 0\}$.
To show that the $k$-proof set satisfies the properties described in Lemma \ref{lem_xi}, we exploit the following property of the distance utility function $u\left(a,k\right)=-|t\left(k\right)-a|$: 
\[v(-a,a,k)=\begin{cases}
2\sgn(k)a &\text{ if } 0 \leq a <|t(k)|,\\
2t(k) &\text{ if } a \geq |t(k)|.
\end{cases}
\] 
Substituting this observation into the proofs of Theorems  \ref{thm_binary} and \ref{thm_product} yields 
$\mu_L^b\left(a\right) \equiv \mu_L^b\left(t\left(1\right)\right) \coloneqq  -\upsilon_L^b $ $\forall a \geq t(1)$;  as well as $\mu_L^p\left(a,k\right) \equiv \mu_L^p\left(|t\left(k\right)|,k\right) \coloneqq -\upsilon_L^p\left(k\right)$ $\forall a \geq |t\left(k\right)|$.

\bigskip

\noindent Part (i): Recall that $\mu_L^b\left(a\right)$ is the unique solution to (\ref{eqn_musymmetry}), i.e., 
\[
\max_{\mu\in [-1,0]} h\left(\mu\right) \text{ s.t. } \frac{1}{2}\left[v\left(-a,a,-1\right)-\mu\right]^+ \geq \lambda h\left(\mu\right).\] Solving this problem using (i) 
$h$ is strictly convex and strictly decreasing on $\left[-1,0\right]$ and (ii) $v(-a,a,-1)$ is nonincreasing in $a$, shows that $\mu_L^b(a)$ is nondecreasing in $a$ (draw a picture yourself). As a result, $\phi^b\left(-a,0,0\right)= a+\mu_L^b\left(a\right)$ is strictly increasing in $a$, which, together with $\left.\phi^b\left(-a,0,0\right)\right\vert_{a=0}=\mu_L^b\left(0\right)<0$, implies that the unique root of $\phi^b(-a,0,0)$ is strictly positive, and that $\Xi^b(0)\coloneqq \{a \geq 0: \phi^b(-a,0,0) \leq 0 \} = [0, \text{ unique root of } \phi^b(-a,0,0)].$ In case the  root exceeds $t(1)$, solving it using the fact that $\phi^b\left(-a,0,0\right)=a-\upsilon_L^b$ $\forall a \geq t(1)$ yields $\upsilon_L^b (\coloneqq \xi^b(0))$.

\bigskip

\noindent Part (ii): For $k=1$, notice that $\phi^p\left(-a, t\left(1\right), 1\right)= a+t\left(1\right)+\mu_L^p\left(a, 1\right)=a+t\left(1\right)-\upsilon_L^p\left(1\right)$ $\forall a \geq t\left(1\right)$, and that $\left.\phi^p\left(-a, t\left(1\right),1\right)\right\vert_{a=t\left(1\right)}=v\left(-t\left(1\right), t\left(1\right), 1\right)-\upsilon_L^p(1)<0$ by (\ref{eqn_ob}). Therefore, the maximum root of $\phi^p(-a,t(1), 1)$ equals $-t(1)+\upsilon_L^p(1) (\coloneqq \xi^p(1))$ and exceeds $t(1)$, and $\Xi^p(1) \coloneqq \{a \geq 0: \phi^p(-a, t(1), 1) \leq 0\}$ satisfies $\Xi^p(1) \cap [t(1), +\infty)=[t(1), \xi^p(1)]. $

The proof for $k=0$ is analogous. For $k=-1$, note that $\phi^p(-a, t(-1), -1)=a-t(1)-\upsilon_L^p(-1)$ $\forall a \geq t(1)$, and that $\left.\phi^p(-a,t(-1), -1)\right\vert_{a=t(1)}=v(-t(1), -t(1), -1)-\upsilon_L^p(-1)=-\upsilon_L^p(-1)<0$. Everything else is the same as above.\qed

\paragraph{Proof of Theorem \ref{thm_main}} 
In the broadcast case,  Lemmas \ref{lem_broadcast} and \ref{lem_xi} together imply  that $\mathcal{E}^{b,q}=\Xi^b(0)=[0, \xi^b(0)]$. The proof for the personalized case with $q(0)>1/2$ is analogous. If $q(0) \leq 1/2$, then 
\[\mathcal{E}^{p, q}=\underbrace{[0, t(1)) \cap \Xi^p(0)}_{A} \cup \underbrace{[t(1), +\infty) \cap \cap_{k \in \mathcal{K}}\Xi^p(k)}_{B}\]
by Lemma \ref{lem_maxmin}. Consider two cases. First, if $\xi^p(0)<t(1)(<\xi^p(\pm 1))$, then A $=[0, \xi^p(0)]$ by Lemma \ref{lem_xi} and B $=\emptyset$. Second, if  $\xi^p(0) \geq t(1)$, then A $=[0, t(1))$ and B  $= [t(1), \min_{k \in \mathcal{K}} \xi^p(k)]$. In both cases, $\mathcal{E}^{p,q}=A\cup B=[0, \min_{k \in \mathcal{K}}\xi^p(k)]$ as desired.  \qed

\paragraph{Proof of Proposition \ref{prop_lambda}} 

It suffices to show that optimal signals become less Blackwell-informative as we raise the attention cost parameter from $\lambda'$ to $\lambda''$. For the broadcast case, the result follows from the fact that the optimal signal is symmetric and, hence, becomes less Blackwell-informative as we raise $\lambda$. For the personalized case, we only prove the result for left-leaning voters. 

Let $\gamma(\lambda)$ denote the Lagrange multipliers associated with voters' participation constraint when the attention cost parameter is given by $\lambda$, and write $\beta(\lambda)$ for $\lambda-1/\gamma(\lambda)$. When proving Theorem \ref{thm_product}, we demonstrated that the optimal personalized signal is pinned down by the first-order conditions (\ref{eqn_foc2}) and (\ref{eqn_foc1}). Summing up  (\ref{eqn_foc2}) and (\ref{eqn_foc1}) yields
\begin{equation}\label{eqn3}
h'\left(\mu_R\right)-h'\left(\mu_L\right)=1/\beta(\lambda),
\end{equation}
and using this result to simplify the total derivative of  (\ref{eqn_foc1}) w.r.t. $\beta(\lambda)$  yields 
\begin{align*}
\frac{d\mu_L}{d\beta(\lambda)}=&\Delta h-h'\left(\mu_R\right)\Delta \mu+\beta(\lambda)\begin{bmatrix}
\cancel{h'\left(\mu_R\right)\frac{d\mu_R}{d\beta(\lambda)}}-h'\left(\mu_L\right)\frac{d\mu_L}{d\beta(\lambda)}-h''\left(\mu_R\right)\frac{d\mu_R}{d\beta(\lambda)}\Delta\mu \\
-\cancel{h'\left(\mu_R\right)\frac{d\mu_R}{d\beta(\lambda)}} +h'\left(\mu_R\right)\frac{d\mu_L}{d\beta(\lambda)}
\end{bmatrix}\\
=&\Delta h-h'\left(\mu_R\right)\Delta \mu-\beta(\lambda) h''\left(\mu_R\right)\frac{d\mu_R}{d\beta(\lambda)}\Delta\mu+\frac{d\mu_L}{d\beta(\lambda)},
\end{align*} 
where $\Delta \mu \coloneqq \mu_R-\mu_L$ and $\Delta h \coloneqq h(\mu_R)-h(\mu_L)$. Therefore, 
\begin{align}\label{eqn4}
\frac{d\mu_R}{d\beta(\lambda)}=\frac{\Delta h-h'\left(\mu_R\right)\Delta \mu}{\beta(\lambda) h''\left(\mu_R\right)\Delta \mu}=\frac{h(\mu_R)-h(|\mu_L|)-h'\left(\mu_R\right)\Delta \mu}{\beta(\lambda) h''\left(\mu_R\right)\Delta \mu}<0,
\end{align}
where the last inequality holds because $h$ is symmetric around zero, $h'>0$ on $[0,1]$, $h''>0$, and $\Delta \mu>0$, hence the numerator of (\ref{eqn4}) is bounded above by $h'\left(\mu_R\right)\left(\mu_R-|\mu_L|-\Delta \mu\right)=-2h'\left(\mu_R\right)|\mu_L|<0$ if $\mu_R>|\mu_L|$ and by $-h'\left(\mu_R\right)\Delta \mu<0$ if $\mu_R\leq |\mu_L|$.  Meanwhile, differentiating (\ref{eqn3}) with respect to $\beta(\lambda)$ yields
\[h''\left(\mu_L\right)\frac{d\mu_L}{d\beta(\lambda)}=h''\left(\mu_R\right)\frac{d\mu_R}{d\beta(\lambda)}+\frac{1}{\beta(\lambda)^2},\]
and simplifying the above expression using (\ref{eqn4}) yields 
\begin{equation}\label{eqn5}
\frac{d\mu_L}{d\beta(\lambda)}=\frac{\Delta h-h'\left(\mu_L\right)\Delta \mu}{\beta(\lambda) h''\left(\mu_L\right) \Delta \mu}>0.
\end{equation}
Together, (\ref{eqn4}) and (\ref{eqn5}) imply that the optimal personalized signal becomes less Blackwell-informative as $\beta(\lambda)$ increases. 

We next show that $\beta(\lambda)$ increases as we raise $\lambda$ from $\lambda'$ to $\lambda''$. Suppose, to the contrary, that $\beta\left(\lambda''\right)\leq \beta\left(\lambda'\right)$, and let $\Pi'$ and $\Pi''$ denote the optimal personalized signals when the attention cost parameter is given by $\lambda'$ and $\lambda''$, respectively. From the previous step, we know that if $\beta(\lambda'')<\beta(\lambda')$, then $\Pi''$ is more Blackwell-informative than $\Pi'$ and so must satisfy  $I\left(\Pi''\right)>I\left(\Pi'\right)>0$. Then from 
$V\left(\Pi'; -1 \right)=\lambda'  I\left(\Pi'\right)$ and $V\left(\Pi'';-1 \right)=\lambda''  I\left(\Pi''\right)$, 
it follows that 
$V\left(\Pi''; -1 \right)-\lambda'  I\left(\Pi''\right)>0=V\left(\Pi'; -1 \right)-\lambda'  I\left(\Pi'\right),$
which together with $I\left(\Pi''\right)>I\left(\Pi'\right)$ implies that $\Pi'$ is not optimal when $\lambda=\lambda'$, a contradiction. Meanwhile, if $\beta\left(\lambda'\right)=\beta\left(\lambda''\right)$, then $
\Pi'=\Pi''$ as demonstrated in the previous step. But then $V\left(\Pi'; -1\right)=\lambda' I\left(\Pi'\right)<\lambda'' I\left(\Pi''\right)=V\left(\Pi''; -1\right)$, which contradicts $\Pi'=\Pi''$. \qed

\section{Numerical solutions}\label{sec_numerical}
In this appendix, we solve the baseline model entropy attention cost.  We first reduce Assumption \ref{assm_regularity} to model primitives. The results depicted in Figure \ref{figure_condition1} are consistent with the intuition discussed in Section \ref{sec_news}, namely when the attention cost parameter is intermediate and extreme voters' policy preferences are moderate, it is optimal to include all voters in nontrivial news consumption, although revealing the true state to them would tune them out. 

\begin{figure}[!h]
\centering
     \includegraphics[width=6cm]{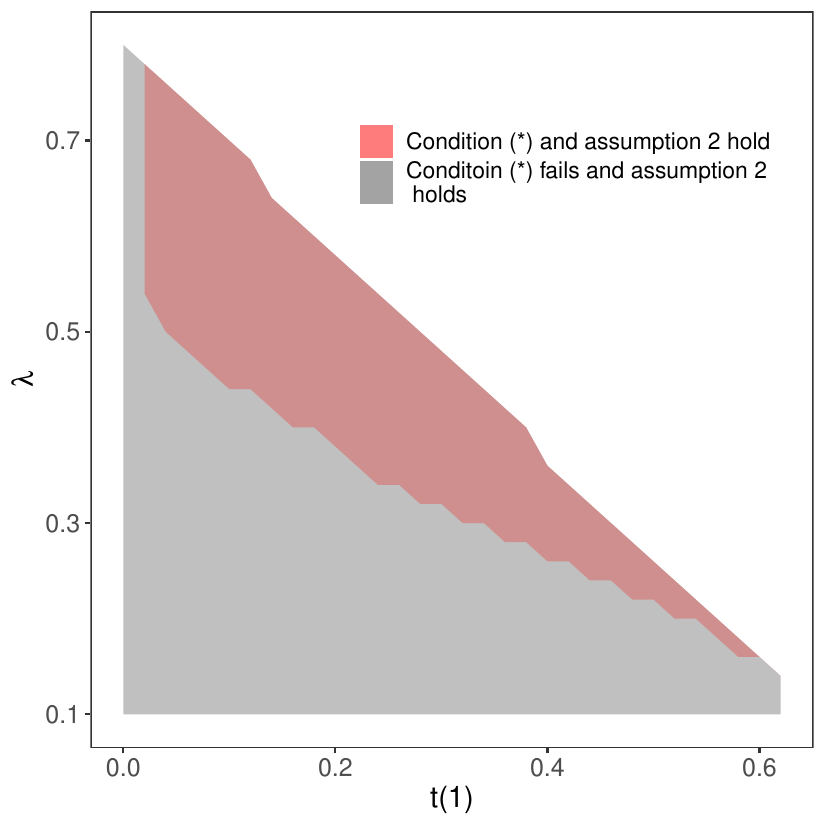}
\caption{Assumption \ref{assm_regularity} and Condition \eqref{eqn_ast}: entropy attention cost, uniform population distribution.}\label{figure_condition1}
\end{figure} 

We next solve for the primitives that render the base voters of candidate $R$ (i.e., right-leaning voters)  disciplining under personalized news aggregation. As demonstrated in Section \ref{sec_cs}, the last situation happens if and only if extreme voters' personalized signals are sufficiently skewed that the beliefs induced by the occasional big surprise and predisposition reinforcement differ by a significant amount:  
\[\tag{$\ast$}\upsilon_L^p\left(1\right)-\upsilon_R^p\left(1\right)>2t\left(1\right).\]
As depicted in Figure \ref{figure_condition1}, Condition (\ref{eqn_ast}) is most likely to hold when the attention cost parameter $\lambda$ is high and extreme voters' policy preference parameter $t(1)$ is large. The finding concerning the policy  preference parameter $t(1)$ is quite intuitive. As for the attention cost parameter $\lambda$, note that as paying attention becomes more costly, the infomediary makes right-leaning  voters' signal less Blackwell-informative in order to prevent them from tuning out. During that process, she is reluctant to cut back $\upsilon_L^p\left(1\right)$,  the occasional big surprise that makes news consumption valuable to these voters. Instead, she reduces $\upsilon_R^p\left(1\right)$ significantly, which causes the left-hand side of Condition (\ref{eqn_ast}) to increase. 

We finally solve for the primitives that increase policy polarization as news aggregation becomes personalized. As demonstrated in Section \ref{sec_cs}, the last situation happens if and only if 
\[\tag{$\ast \ast$}\xi^b\left(0\right)<\min_{k \in \mathcal{K}} \xi^p(k).\label{eqn_doubleast}\]
For all parameter values we've tried, only extreme voters can be disciplining, leaving us with two sub-cases to consider. 

\bigskip

\noindent \emph{Base voters are disciplining.} In this case, Condition \eqref{eqn_ast} fails, and Condition \eqref{eqn_doubleast} becomes $ \upsilon_L^p\left(1\right)-\upsilon_L^b >t\left(1\right)$. As depicted in Figure \ref{figure_condition2a}, the last condition is most likely to hold when the attention cost parameter $\lambda$ and extreme voters' policy preference parameter $t(1)$ are both large.  As $\lambda$ increases (hence paying attention becomes more costly), the infomediary makes signals less Blackwell-informative in order to prevent voters from tuning out. In the personalized case, she can reduce $\upsilon_R^p\left(1\right)$ significantly while keeping $\upsilon_L^p\left(1\right)$ almost unchanged in order to make news consumption still useful for right-leaning voters. Such flexibility is absent in the broadcast case, where the two posterior beliefs $\upsilon_L^b$ and $\upsilon_R^b$ must be reduced by the same magnitude. As a result, $\upsilon_L^p\left(1\right)-\upsilon_L^b$ increases, which makes Condition \eqref{eqn_doubleast} easier to satisfy. 

\begin{figure}[!h]
\centering
\includegraphics[width=6cm]{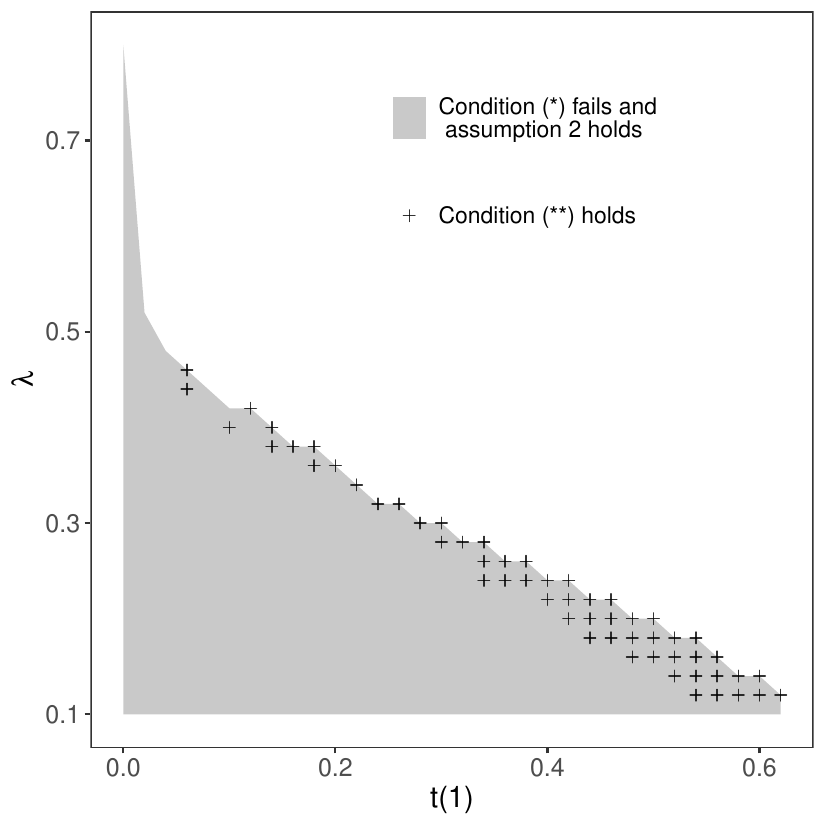}
    \caption{Condition \eqref{eqn_doubleast}: entropy attention cost, uniform population distribution, Condition \eqref{eqn_ast} fails.}\label{figure_condition2a}
\end{figure}

As for the effects of strengthening voters' policy preferences, note that as $t\left(1\right)$ increases, extreme voters find news consumption less useful, so the broadcast signal must become less Blackwell-informative in order to prevent them from tuning out, i.e., $\upsilon_L^b$ must decrease. In the meantime, $\upsilon_L^p\left(1\right)$ should increase, because to convince right-leaning voters to vote for candidate $L$ requires a bigger occasional surprise than before. Thus $\upsilon_L^p\left(1\right)-\upsilon_L^b$ increases, which relaxes Condition \eqref{eqn_ast} when $t\left(1\right)$ is sufficiently large.

\bigskip

\noindent \emph{Opposition voters are disciplining.} In this case,  Condition \eqref{eqn_ast} holds, and so Condition \eqref{eqn_doubleast} becomes $|t\left(-1\right)|>\upsilon_L^b-\upsilon_L^p\left(-1\right)$. As depicted in Figure \ref{figure_condition2b}, the last condition is most likely to hold  when $t(1)$ is large, and so extreme voters have strong policy preferences. As $t\left(1\right)$ increases (hence $t(-1)$ becomes more negative), left-leaning voters seek a bigger occasional surprise from news  consumption than before, so $\upsilon_R^p\left(-1 \right)$ should increase. Also since they find news consumption less useful, $\upsilon_L^p\left(-1\right)$ must decrease significantly to prevent them from tuning out. Meanwhile in the broadcast case, $\upsilon_L^b$ ($=\upsilon_R^b$) must decrease to prevent extreme voters from tuning out. When $t\left(1\right)$ is sufficiently large, the right-hand side of Condition \eqref{eqn_doubleast} is close to zero whereas the left-hand side of it is big, which explains the pattern depicted in Figure \ref{figure_condition2b}. 

\bigskip

 \begin{figure}[!h]
\centering
     \includegraphics[width=6cm]{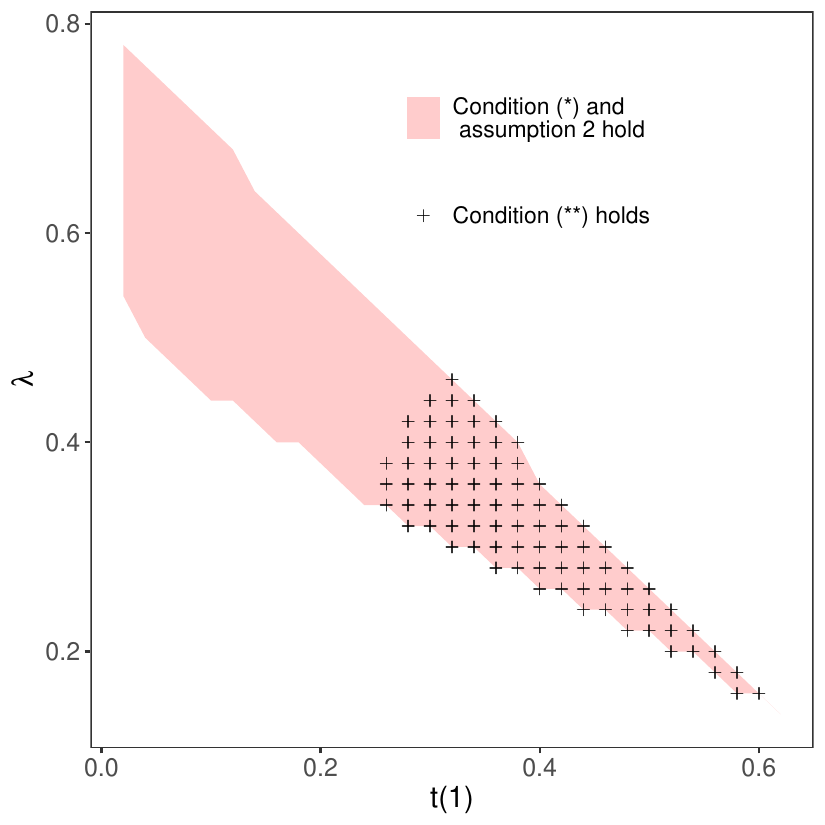}
\caption{Condition \eqref{eqn_doubleast}: entropy attention cost, uniform population distribution, Condition \eqref{eqn_ast} holds.}\label{figure_condition2b}
\end{figure}

\cleardoublepage

    \vspace*{16em}
    \begin{center}
        \Huge \textbf{Online Appendices \\(For Online Publication Only)}
    \bigbreak
    \end{center}

\thispagestyle{empty}
\cleardoublepage

\appendix
\setcounter{section}{0}
\setcounter{page}{1}
\gdef\thesection{O.\arabic{section}}

\section{General model}\label{sec_general}
This appendix has two purposes. The first purpose is to extend the baseline model to general voters. Throughout, suppose that candidates can adopt the policies in a compact interval $\mathcal{A}=[-\overline{a}, \overline{a}]$, where $\overline{a}$ is finite but large. Voters' type space is a finite set  $\mathcal{K}=\left\{-K,\cdots, 0, \cdots, K\right\}$ with $K \geq 1$. Their population function $q: \mathcal{K} \rightarrow \mathbb{R}_{++}$ has support $\mathcal{K}$ and is symmetric around zero. Their utility function $u: \mathcal{A} \times \mathcal{K} \rightarrow \mathbb{R}$ satisfies the properties listed in Observation \ref{obs_utility}, i.e., 

\begin{assmO}\label{assm_utility}
\begin{description}
\item[Continuity and weak concavity]$u
\left(\cdot, k\right)$ is continuous and weakly concave for any $k \in \mathcal{K}$.

\item[Symmetry] $u\left(a,k\right)=u\left(-a,-k\right)$ for any $a\in \mathbb{R}$ and $k \in \mathcal{K}$.

\item[Inverted V-shape] $u\left(\cdot, k\right)$ is strictly increasing on $[-\overline{a}, t\left(k\right)]$ and is strictly decreasing on $[t\left(k\right), \overline{a}]$ for any $k \in \mathcal{K}$, where $t: \mathcal{K} \rightarrow \mathcal{A}$ is strictly increasing and symmetric around zero. 

\item[Increasing differences] $v(-a,a',k) \coloneqq u\left(a,k\right)-u\left(a',k\right) $ is increasing in $k$ for any $a>a'$. For any $a>0$, $v(-a,a,k) \coloneqq u\left(a,k\right)-u\left(-a,k\right)$ is strictly positive if $k>0$, equals zero if $k=0$, and is strictly negative if $k<0$. 
 \end{description}
\end{assmO}

\begin{obsO}
Theorems \ref{thm_binary} and \ref{thm_product} of the baseline model remain valid under Assumptions \ref{assm_attention}, \ref{assm_regularity}, and \ref{assm_utility}. 
\end{obsO}

\begin{proof}
The proof for the personalized case is the exact same as before. As for the broadcast case, notice that in the current setting, as well as in the baseline model, only voters of the most extreme types can have binding participation constraints, whereas those of interim types must have slack participation constraints. Replacing $k=\pm 1$ with $k=\pm K$ in the proofs of Theorems \ref{thm_binary} and \ref{thm_product} gives the desired result. 
\end{proof}

The second purpose of this appendix is to relax the assumption that signals are conditionally independent across market segments. In what follows, we'll develop new concepts in Appendix \ref{sec_general_concept} and  conduct equilibrium analyses in Appendix \ref{sec_general_result}. 

\subsection{Key concepts}\label{sec_general_concept}
\paragraph{Joint signal distribution} A \emph{joint signal distribution} is  a tuple $(\bm{\chi}, \bf{b}^+, \bf{b}^-)$ of a  \emph{configuration matrix} $\bm{\chi}$ and \emph{probability vectors} ${\bf{b}}^+$ and ${\bf{b}}^-$. The configuration matrix $\bm{\chi}$ has $|\mathcal{K}|$ rows. Each column of it constitutes a profile of the voting recommendations to type $-K,\cdots, K$ voters that occurs with a strictly positive probability. Each entry of $\bm{\chi}$ is either $0$ or $1$, where $0$ means that candidate $R$ is disapproved of, and $1$ means that he is endorsed. For example, the configuration matrix is 
\[\bm{\chi}^{\ast}= \begin{bmatrix}
0 & 1\\
0 & 1\\
\vdots & \vdots\\
0 & 1
\end{bmatrix}\]
if $\mathcal{S}=b$, and it is 
\[\bm{\chi}^{\ast \ast} = \underbrace{\begin{bmatrix}
0 & 1 & 0 & \cdots & 0 & 1 & \cdots & 0 & \cdots & 1\\
0 & 0 & 1 & \cdots & 0 & 1 & \cdots & 0 & \cdots & 1 \\
\vdots & \vdots & \vdots & \cdots & \vdots & \vdots & \cdots & \vdots &\cdots & 1  \\
0 & 0 & 0 & \cdots & 0 & 0 & \cdots & 1 & \cdots & 1 \\
0 & 0 & 0 & \cdots & 1 & 0 & \cdots & 1 & \cdots & 1 
\end{bmatrix}}_\text{$2^{|\mathcal{K}|}$ columns}\]
if $\mathcal{S}=p$ and signals are conditionally independent across voters. The vectors ${\bf{b}}^+$ and ${\bf{b}}^-$ compile the probabilities that each column of $\bm{\chi}$ occurs in states $\omega=1$ and $\omega=-1$, respectively. By definition, all elements of ${\bf{b}}^+$ or ${\bf{b}}^-$ are strictly positive and add up to one.

We consider \emph{symmetric} joint signal distributions that are \emph{consistent} with the marginal signal distributions solved in Section \ref{sec_news}. To formally define symmetry, let ${\bf{x}}$ be a generic voting recommendation profile to type $-K, \cdots, K$ voters, $\bf{1}$ be the $|\mathcal{K}|$-vector of  ones, and \[{\bf{P}}=\left[\begin{matrix}
& & 1\\
& \iddots & \\
1 & &
\end{matrix}\right]\]
be a $|\mathcal{K}| \times |\mathcal{K}|$ permutation matrix. Define the \emph{symmetry operator} $\Sigma$ as 
\[\Sigma  ({\bf{x})} = {\bf{P}} \left({\bf{1}}-{\bf{x}}\right),\] 
so that ${\bf{x}}$ recommends candidate $z \in \left\{L, R\right\}$ to type $k$ voters if and only if $\Sigma ({\bf{x}})$ recommends candidate $-z$ to type $-k$ voters. A joint signal distribution is \emph{symmetric} if the probability that a voting recommendation profile ${\bf{x}}$ occurs in state $\omega=1$ equals the probability that $\Sigma( {\bf{x}})$ occurs in state $\omega=-1$. Formally, 

\begin{defnO}
A configuration matrix $\bm{\chi}$ is \emph{symmetric} if for any \linebreak $m \in \left\{1, \cdots, \#\mathrm{columns}\left(\mathcal{\bm{\chi}}\right)\right\}$, there exists $n\in \left\{1,\cdots, \#\mathrm{columns}\left(\mathcal{\bm{\chi}}\right)\right\}$ such that $\Sigma (\left[\bm{\chi}\right]_m)=\left[\bm{\chi}\right]_n$. A joint signal distribution $(\bm{\chi}, {\bf{b}}^+, {\bf{b}}^-)$ is \emph{symmetric} if $\bm{\chi}$ is symmetric and $\left[{\bf{b}}^+\right]_m=\left[{\bf{b}}^-\right]_n$ for any $m,n$ as above.\footnote{With a slight abuse of notation, we use $\left[\cdot\right]_m$ to denote both the $m^{th}$ entry of a column vector and the $m^{th}$ column of a matrix. $\#\mathrm{columns}\left(\mathcal{\bm{\chi}}\right)$ denotes the number of the columns of $\bm{\chi}$.}
\end{defnO}

Next is our notion of \emph{consistency}. In Footnote \ref{footnote2}, we solved for the marginal probabilities that the signal consumed by type $k$ voters endorses candidate $R$  in states $\omega =1$ and $\omega=-1$, respectively, holding any segmentation technology $\mathcal{S}$ and symmetric policy profile $(-a,a)$ fixed. Compiling these probabilities across type $-K,\cdots, K$ voters  yield two $|\mathcal{K}|$-vectors $\bm{\pi}^{\mathcal{S}, +}\left(a\right)$ and $\bm{\pi}^{\mathcal{S}, -}\left(a\right)$ of marginal probabilities. 

\begin{defnO}[label=defn_consistency]
A joint signal distribution $(\bm{\chi}, {\bf{b}}^+, {\bf{b}}^-)$ is \emph{$(\mathcal{S},a)$-consistent} for some $\mathcal{S} \in \left\{b,p\right\}$ and $a \in \left[0, \overline{a}\right]$ if 
\[
\bm{\chi} {\bf{b}^+} =\bm{\pi}^{\mathcal{S},+}\left(a\right) \text{ and } \bm{\chi} {\bf{b}}^-={\bm \pi}^{\mathcal{S},-}\left(a\right).\]
A configuration matrix $\bm{\chi}$ is \emph{$(\mathcal{S},a)$-consistent}  if there exist probability vectors ${\bf{b}}^+$ and ${\bf{b}}^-$ such that the joint signal distribution $(\bm{\chi}, {\bf{b}}^+, {\bf{b}}^-)$ is \emph{$(\mathcal{S}, a)$-consistent}. $\bm{\chi}$ is \emph{$\mathcal{S}$-consistent} if it is $(\mathcal{S}, a)$-consistent for all $a \in \left[0, \overline{a}\right]$. 
\end{defnO}
 
 By definition, $\bm{\chi}^{\ast}$ is $b$-consistent and, indeed, the only $(b,a)$-consistent configuration for any given $a \in \left[0, \overline{a}\right]$.  $\bm{\chi}^{\ast \ast}$ is $p$-consistent, but it is not the only $p$-consistent configuration in general  (examples are available upon request). 

\paragraph{Attraction-proof set} In the baseline model, we defined several related concepts, including a voter's susceptibility to policy deviations, his attraction-proof set, and his policy latitude.  
We now generalize these concepts to sets of voters. 
 
\begin{defnO}
Under segmentation technology $\mathcal{S} \in \{b,p\}$, a deviation of candidate $R$ from a symmetric policy profile $(-a,a)$ with $a \in \left[0, \overline{a}\right]$  to $a'$ \emph{attracts} a set $\mathcal{D} \subseteq \mathcal{K}$ of voters if it attracts all members of $\mathcal{D}$, i.e., $\phi^{\mathcal{S}}\left(-a,a',k\right)>0$ $\forall k \in \mathcal{D}$. This is equivalent to  \[\phi^{\mathcal{S}}\left(-a,a',\mathcal{D}\right)\coloneqq \displaystyle \min_{k \in \mathcal{D}}\phi^{\mathcal{S}}\left(-a,a',k\right)>0,\] where $\phi^{\mathcal{S}}\left(-a,a',\mathcal{D}\right)$ is the \emph{$\mathcal{D}$-susceptibility} to $a'$ following unfavorable news to candidate $R$'s valence. The \emph{$\mathcal{D}$-proof set} $\Xi^{\mathcal{S}}\left(\mathcal{D}\right)$ gathers all  nonnegative policy $a$'s such that no deviation of candidate $R$ from $(-a,a)$ attracts $\mathcal{D}$, i.e., 
\[\Xi^{\mathcal{S}}\left(\mathcal{D}\right)\coloneqq \left\{a \in \left[0, \overline{a}\right]: \max_{a' \in \mathcal{A}}\phi^{\mathcal{S}}\left(-a,a', \mathcal{D}\right) \leq 0\right\}.\]
The maximum of the $\mathcal{D}$-proof set
\[\xi^{\mathcal{S}}\left(\mathcal{D}\right)\coloneqq \max\Xi^{\mathcal{S}}\left(\mathcal{D}\right)\] 
is $\mathcal{D}$'s \emph{policy latitude}. 
\end{defnO}

\paragraph{Influential coalition}  The next concept  is integral to the upcoming analysis.

\begin{defnO}[label=defn_general_coalition]
Fix any segmentation technology $\mathcal{S} \in \left\{b,p\right\}$,  symmetric policy profile $(-a,a)$ with $a \in \left[0, \overline{a}\right]$, and population function $q$, and let the default be the strictly obedient outcome induced by any joint signal distribution $(\bm{\chi}, {\bf{b}}^+, {\bf{b}}^-)$ that is $(\mathcal{S}, a)$-consistent. A set $\mathcal{C} \subseteq \mathcal{K}$ of voters constitutes an \emph{influential coalition} if attracting $\mathcal{C}$ while holding other things constant strictly increases candidate $R$'s winning probability compared to the default. 
\end{defnO}

By definition, majority coalitions are influential, and supersets of influential coalitions are influential. In the broadcast case, all voters consume the same signal, so a coalition of voters is influential if and only if it is a majority coalition. In the personalized case, non-majority coalitions can be influential due to the imperfect correlation between voters' signals (see Table \ref{table1} for an illustration). 

In principle, influential coalitions can depend on the entire joint signal distribution (and certainly voters' population distribution). The next lemma limits such dependence to the configuration matrix only. 

\begin{lemO}\label{lem_coalition}
Let everything be as in Definition \ref{defn_general_coalition}. Then influential coalitions depend on the joint signal distribution $(\bm{\chi}, {\bf{b}}^+, {\bf{b}}^-)$ and voter's population distribution $q$ only through the pair $(\bm{\chi}, q)$,  and they are independent of the policy profile $(-a,a)$ if $\bm{\chi}$ is $\mathcal{S}$-consistent.
\end{lemO}

To make sense of Lemma \ref{lem_coalition}, let us visualize what happens to the configuration matrix after a candidate commits a policy deviation that attracts some voters without repelling anyone. For any voter who is attracted by the deviation, we turn the corresponding row of voting decisions to all ones. The deviation strictly increases the candidate's winning probability if there exists a column, i.e., a voting recommendation profile, such that after the augmentation of rows, there are enough ones that allow the candidate to win the election in that column.  The probabilities of the columns do not matter, after the configuration matrix and voters' population distribution have been taken into account. The proof presented in Online Appendix \ref{sec_general_proof} formalizes this idea. 

\subsection{Main results}\label{sec_general_result}
The next lemma giving a full characterization of the symmetric policy profiles that can arise in equilibrium, thus extending Lemma \ref{lem_maxmin}  to general voters and joint signal distributions. 

\begin{lemO}[label=lem_general_maxmin]
Fix any pair of segmentation technology $\mathcal{S} \in \left\{b,p \right\}$ and population function $q$, and assume Assumptions \ref{assm_attention}, \ref{assm_regularity}, and \ref{assm_utility}. Then the following are equivalent. 
\begin{enumerate}[(i)]
\item A symmetric policy profile $(-a,a)$ with $a \in \left[0, \overline{a}\right]$ can arise in an equilibrium with a joint signal  distribution $(\bm{\chi}, {\bf{b}}^+, {\bf{b}}^- )$ that is $(\mathcal{S}, a)$-consistent.
\item No deviation of candidate $R$ from $(-a,a)$ to $a' \in [-a,a)$ attracts any influential coalition formed under $(\bm{\chi}, q)$ whose members have ideological bliss points in $\left[-a,a\right]$.
\end{enumerate}
\end{lemO}

\begin{proof}
In the proof of Lemma  \ref{lem_maxmin}, replacing the left-leaning voters (of type $k=1$) with any type $k<0$ voter and the right-leaning voters (of type $k=1$) with any type $k>0$ voter gives the desired result. 
\end{proof}

In what follows, we'll use $\mathcal{E}^{\mathcal{S}, \bm{\chi}, q}$ denote the set of the nonnegative policy $a$'s such that $(-a,a)$ can arise in equilibrium under segmentation technology $\mathcal{S}$, configuration matrix $\bm{\chi}$, and population function $q$. As before, we are interested in the degree of policy polarization $a^{\mathcal{S}, \bm{\chi}, q}$, defined  as the maximum of $\mathcal{E}^{\mathcal{S}, \bm{\chi}, q}$, and whether all policies between zero and $a^{\mathcal{S}, \bm{\chi}, q}$ can arise in equilibrium.  We focus on $\mathcal{S}$-consistent $\bm{\chi}$s, so that $\mathcal{E}^{\mathcal{S}, \bm{\chi}, q}$ can be computed in two steps.\footnote{For an arbitrary $\bm{\chi}$, one needs to check, after Step 2, whether the output policy $a$ is $(\mathcal{S},a)$-consistent with $\bm{\chi}$. } 
\begin{enumerate}
\item Compute the influential coalitions formed under $(\bm{\chi}, q)$.
\item For each $a \in \left[0, \overline{a}\right]$, check if any  deviation as in Lemma \ref{lem_general_maxmin}(ii) attracts any influential coalition whose members have ideological bliss points in $[-a,a)$. If the answer is negative, then add $a$ to the output set. 
\end{enumerate}
In Online Appendix \ref{sec_general_proof},  we express the equilibrium policy set in terms of attraction-proof sets using the above algorithm.  We then argue, as in the baseline model, that the equilibrium policy polarization  must be strictly positive, due to basic properties of voters' preferences and the strict obedience induced by optimal signals. 
To obtain sharper characterizations, we impose the following regularity condition on the susceptibility function.  

\begin{assmO}\label{assm_monotone_a}
$\phi^{\mathcal{S}}\left(-a,a',k\right)$ is increasing in $a$ on $\left[|t\left(k\right)|, \overline{a}\right]$ for any $\mathcal{S} \in \left\{b,p\right\}$, $k \in \mathcal{K}$, and $a' \in \mathcal{A}$. 
\end{assmO}

To grasp the meaning of Assumption \ref{assm_monotone_a}, note that $\phi^{\mathcal{S}}(-a,a',k)$ depends on $a$ only through $-u(-a,k)+\mu_L^{\mathcal{S}}(a,k)$. Since  $-u(-a,k)$ is strictly increasing in $a$ on $[|t(k)|, \overline{a}]$, Assumption \ref{assm_monotone_a} holds if $\mu_L^{\mathcal{S}}(a,k)$ doesn't vary significantly with $a$ on $[|t(k)|, \overline{a}]$. This is indeed the case under the distance utility function, as $\mu_L^{\mathcal{S}}(a,k)\equiv -\upsilon_L^{\mathcal{S}}(k)$ $\forall a \in [|t(k)|, \overline{a}]$ and so is independent of $a$ on that domain. The next lemma provides additional sufficient conditions for Assumption \ref{assm_monotone_a} to hold. 

\begin{lemO}\label{lem_monotone_a}
Under Assumptions \ref{assm_attention}, \ref{assm_regularity}, and \ref{assm_utility}, 
 Assumption \ref{assm_monotone_a} holds if $\mathcal{S}=b$ or if $\mathcal{S}=p$ and either $u\left(a,k\right)=-|t\left(k\right)-a|$ or $h(\mu)=\mu^2$. 
\end{lemO}

With the above preparations, we can now state our main result. 

\begin{thmO}[label=thm_general_main]
Fix any segmentation technology $\mathcal{S} \in \left\{b,p\right\}$, $\mathcal{S}$-consistent configuration matrix $\bm{\chi}$, and population function $q$. Let $\mathcal{C}$ denote a typical influential coalition formed under $(\bm{\chi}, q)$. Under Assumptions \ref{assm_attention}, \ref{assm_regularity}, \ref{assm_utility}, and \ref{assm_monotone_a}, $\mathcal{E}^{\mathcal{S}, \bm{\chi},q}=\left[0, a^{\mathcal{S}, \bm{\chi},q}\right]$, where $a^{\mathcal{S}, \bm{\chi},q}=\displaystyle \min_{\mathcal{C}s \text{ formed under }(\bm{\chi}, q)} \xi^{\mathcal{S}}\left(\mathcal{C}\right)>0$. 
\end{thmO}

Theorem \ref{thm_general_main} conveys two important messages. First, policy polarization is in general  disciplined by the influential coalition with the smallest policy latitude and is strictly positive. Second, marginal signal distributions affect policy polarization through policy latitudes, whereas the joint signal distribution does so through the configuration matrix, holding marginal signal distributions constant.  

The remainder of this appendix investigates the comparative statics of policy polarization regarding influential coalitions, holding marginal signal distributions (and, hence, voters' policy latitudes) fixed. Our starting observation is that enriching the configuration matrix enriches influential coalitions and, by Theorem \ref{thm_general_main}, reduces policy polarization. Formally, we say that $\bm{\chi}$ is \emph{richer than $\bm{\chi}'$} and write $\bm{\chi} \succeq \bm{\chi}'$ if $\bm{\chi}$ prescribes more voting recommendation profiles than $\bm{\chi}'$.

\begin{defnO}
$\bm{\chi} \succeq \bm{\chi}'$ if every column of $\bm{\chi}'$ is a column of $\bm{\chi}$. 
\end{defnO}

\begin{obsO}\label{obs_richness}
For any segmentation technology $\mathcal{S} \in \left\{b,p\right\}$, any $\mathcal{S}$-consistent configuration matrices $\bm{\chi}$ and $\bm{\chi}'$ such that $\bm{\chi}\succeq \bm{\chi}'$, and any population function $q$, $\left\{\mathcal{C}s \text{ formed under } (\bm{\chi}', q )\right\} \subseteq \left\{\mathcal{C}s \text{ formed under } (\bm{\chi}, q)
\right\}$.
\end{obsO}

\begin{proof}
In the proof of Lemma \ref{lem_coalition}, we noted that a policy deviation attracts a voter if it turns his row of decisions in the configuration matrix into all ones, and that it strictly increases candidate $R$'s winning probability if and only if there exists a column of the configuration matrix such that after the augmentation of rows, the candidate strictly increases his winning probability in that column. Thus the richer the configuration matrix is, the more influential coalitions there are.
\end{proof}

We examine two implications of Observation \ref{obs_richness}.  The next proposition shows that under personalized news aggregation,  policy polarization is minimized when signals are conditionally independent across voters for any given  population distribution, and the global minimum $\displaystyle \min_{k \in \mathcal{K}}\xi^p\left(k\right)$ is attained when voters' population distribution is uniform across types. 

\begin{propO}\label{prop_richness}
Let everything be as in Theorem \ref{thm_general_main}. Then 
$\displaystyle \min_{k \in \mathcal{K}}\xi^p\left(k\right)=a^{p, \bm{\chi}^{\ast \ast}, \text{uniform}}\leq a^{p, \bm{\chi}^{\ast \ast}, q} \leq a^{p, \bm{\chi},q}$
for any $p$-consistent configuration matrix $\bm{\chi}$ and any population function $q$.
\end{propO}

\begin{proof}
The second inequality holds because $\bm{\chi}^{\ast \ast} \succeq \bm{\chi}$ for any  $p$-consistent $\bm{\chi}$. To establish the first equality and inequality, note that under $\bm{\chi}^{\ast \ast}$ and uniform population distribution, each type of voter is influential, and the resulting collection $2^{\mathcal{K}}-\left\{\emptyset\right\}$ of influential coalitions is the richest across all scenarios. 
\end{proof}

The implications of Proposition \ref{prop_richness} are twofold. First, Theorem \ref{thm_main} of the main text prescribes the exact lower bound for the policy polarization under personalized news aggregation. Second, as long as this lower bound stays positive, changes in the environment (e.g., enrich voters' types, divide the same type of voters into multiple subgroups) wouldn't render policy polarization trivial. 

Consider next the transition from broadcast news aggregation to  personalized news aggregation. As demonstrated in the next proposition, such a transition enriches the configuration matrix and hence has a negative policy polarization effect, holding other things constant. In light of this result, the reader can safely attribute the increasing policy polarization shown in Proposition \ref{prop_personalization} to   changes in marginal signal distributions.

\begin{propO}\label{prop_btop}
$\left\{\mathcal{C}s \text{ formed under } (\bm{\chi}^{\ast}, q)\right\} \subseteq \left\{\mathcal{C}s \text{ formed under } (\bm{\chi}, q)
\right\}$ for any $p$-consistent configuration matrix $\bm{\chi}$ and any population function $q$. 
\end{propO}

\begin{proof}
$\forall \bm{\chi}$ and $q$ as above,
$\left\{\mathcal{C}s \text{ formed under } (\bm{\chi}, q)
\right\} \supseteq \left\{\text{majority coalitions}\right\}\\=
\left\{\mathcal{C}s \text{ formed under } (\bm{\chi}^{\ast}, q)\right\} 
$.
\end{proof}

We finally examine the policy polarization effect of increasing mass polarization under personalized news aggregation. As in the main text, we define increasing mass polarization as a mean-preserving spread of voters' policy preferences. 

\begin{defnO}\label{defn_population}
\emph{The mass is more polarized} under $q'$ than $q$  (write $q \succeq_{mass} q'$) if $\sum_{k=m}^K q\left(k\right)\leq \sum_{k=m}^K q'\left(k
\right)$ $\forall m=1,\cdots, K$.
\end{defnO}

The analysis assumes quadratic attention cost. Since we are only concerned with personalized news aggregation, all we need, in addition to Assumption \ref{assm_attention} and \ref{assm_utility}, is the following assumption, which ensures that Assumption \ref{assm_regularity} holds under personalized news aggregation and quadratic attention cost. 

\begin{assmO}\label{assm_quadratic}
$h(\mu)=\mu^2$, $\lambda >1/2$, and $4\lambda v\left(-\overline{a},\overline{a}, K\right)<1$. 
\end{assmO}

The next proposition proves a similar result to Proposition \ref{prop_population} for general voters and $p$-consistent  configurations under Assumption \ref{assm_quadratic}. 

\begin{propO}\label{prop_general_population}
Under Assumptions \ref{assm_attention}, \ref{assm_utility}, and \ref{assm_quadratic}, $a^{p, \bm{\chi},q} \geq a^{p, \bm{\chi},q'}$ 
for any $p$-consistent configuration $\bm{\chi}$ and any population functions $q$ and $q'$ such that $q \succeq_{mass} q'$.
\end{propO}

Proposition \ref{prop_general_population} exploits the fact that under quadratic attention cost, voters become more resistant to policy deviations as their types become more centrist. For a deviating candidate, this reduces the problem of attracting a voter coalition to that of attracting the most centrist voter in that coalition. As the mass becomes more polarized, the voters that the candidate needs to attract become more susceptible to policy deviations, hence equilibrium policy polarization falls. The proof presented in Online Appendix \ref{sec_general_proof} formalizes this idea. 

\section{Competitive infomediaries}\label{sec_competitive}
This appendix extends the baseline model to perfectly competitive infomediaries. In the environment laid out in Appendix  \ref{sec_general}, suppose that each type $k \in \mathcal{K}$ voter is served by $m\left(k\right) \geq 2$ infomediaries. A \emph{market segment} is a pair $\left(k,i\right)$, where $k \in \mathcal{K}$ represents the type of the voters being served, and $i \in \left\{1,\cdots, m\left(k\right)\right\}$ represents the serving infomediary.  The population of the voters in market segment $\left(k,i\right)$ is $\rho\left(k,i\right)$, where $\rho\left(k, i\right)>0$ and 
$\sum_{i=1}^{m\left(k\right)} \rho\left(k, i\right)=q\left(k\right)$. The functions $m$ and $\rho$ are symmetric (i.e., $m\left(k\right)=m\left(-k\right)$ and $\rho\left(k, i\right)=\rho\left(-k, i\right)$ for any $ k \in \mathcal{K}$ and $i=1,\cdots, m\left(k\right)$), and they are taken as given throughout this appendix. 

Infomediaries compete \`{a} la Bertrand for voters.  The resulting signal for market segment $\left(k,i\right)$ maximizes the net expected utilities of the voters therein (as in the standard RI model), taking candidates' policy profile ${\bf{a}}=(-a,a)$ as given: 
\[\max_{\Pi} V\left(\Pi; {\bf{a}},k\right)-\lambda I\left(\Pi\right).\]
Across market segments, we consider all joint signal distributions that are symmetric and consistent with the marginal signal distributions that solve the above problem  (hereinafter, \emph{$c$-consistency}). As in Appendix  \ref{sec_general}, we can represent a joint signal distribution by its matrix form and define the $c$-consistency of the configuration matrix. The exercise is omitted to for brevity's sake. 

We examine the policy polarization effect of introducing perfect  competition between infomediaries. To facilitate comparison with the monopolistic personalized case, we redefine $p$-consistency by first forming  market segments using functions $m$ and $\rho$, and then restricting voters of the same type to receiving the same voting recommendation. By Theorem \ref{thm_general_main}, equilibrium policies are fully determined by (i) $\mathcal{S} \in \left\{c,p\right\}$, which pins down marginal signal distributions, (ii) the configuration matrix $\bm{\chi}$, and (iii) population functions $m$ and $\rho$. Hereinafter we shall use $\mathcal{E}^{\mathcal{S}, \bm{\chi}, m,\rho}$ to denote the equilibrium policy set and $a^{\mathcal{S}, \bm{\chi}, m,\rho}$ to denote policy polarization.

Compared to the monopolistic personalized case, the perfect competition between infomediaries affects  policy polarization through both the marginal news distributions and the news configuration matrix.  In the baseline model, we already discussed the first channel, namely monopolistic personalized signals overfeed voters with information about candidates' valence, in the sense of Blackwell,  in order to maximize the infomediary's profit rather than voters' utilities. Competition corrects this overfeeding problem and, hence, attenuate voters' beliefs about candidates' valence. As voters become more susceptible to policy deviations, their policy latitude should fall. 

The second channel is new and more subtle, namely  competition reduces the correlation between voters' signals if competing infomediaries act independently of each other. The situation we have in mind is that voters of the same type receive perfectly correlated signals from the monopolistic infomediary, but conditionally independent signals from competitive infomediaries. While nothing prevents the monopolistic infomediary from randomizing among the same type of voters (which has no impact on its  profit, taking candidates' policies as given), it is fair to say that in the case where competing infomediaries act more independently than the monopolistic infomediary, competition enriches the news configuration matrix compared to the monopolistic personalized case. 

Combining the aforementioned forces, which move in the same direction, we obtain the following proposition.

\begin{propO}\label{prop_competitive}
Fix any functions $m$ and $\rho$ as above, and suppose that Assumptions \ref{assm_attention}, \ref{assm_regularity}, \ref{assm_utility}, and \ref{assm_monotone_a} hold for $\mathcal{S} \in \left\{c,p\right\}$. Then $\mathcal{E}^{c,\bm{\chi},m,\rho}=\left[0, a^{c, \bm{\chi}, m, \rho}\right] \subseteq  \mathcal{E}^{p,\bm{\chi}',m,\rho}=\left[0, a^{p, \bm{\chi}', m, \rho}\right]$ for any $c$-consistent configuration $\bm{\chi}$ and $p$-consistent configuration $\bm{\chi}'$ such that $\bm{\chi} \succeq \bm{\chi}'$.
\end{propO}

\section{Continuous state space}\label{sec_state}
This appendix extends the baseline model to a continuum of states. In the environment laid out in Appendix  \ref{sec_general}, suppose that the valence state is distributed on $\mathbb{R}$ according to a c.d.f. $G$ that is differentiable and symmetric around zero.\footnote{Assuming $\omega \in \mathbb{R}$ is w.l.o.g. because an RI voter who cares ultimately about the differential quality between the two candidates would only acquire information about this single-dimensional random variable.\label{fn_singledimension}} A  signal structure is a mapping $\Pi: \mathbb{R} \rightarrow \Delta\left(\mathcal{Z}\right)$, where each $\Pi\left(\cdot \mid\omega\right)$ specifies a probability distribution over a finite but potentially large set $\mathcal{Z}$ of signal realizations conditional on the state being $\omega \in \mathbb{R}$. Under signal structure $\Pi$, 
\[\pi_z=\int_{\omega \in \mathbb{R}} \Pi\left(z \mid \omega\right) dG\left(\omega\right)\]
is the probability that the signal realization is $z \in \mathcal{Z}$. Whenever $\pi_z>0$,  
\[\mu_z=\int_{\omega \in \mathbb{R}} \omega \Pi\left(z \mid \omega\right) dG\left( \omega\right)/\pi_z\]
is the posterior mean of the state conditional on the signal realization being $z \in \mathcal{Z}$. 

The analysis assumes entropy attention cost. 
\begin{assmO}\label{assm_entropy}
The needed amount of attention for consuming $\Pi: \mathbb{R} \rightarrow \Delta\left(\mathcal{Z}\right)$ is 
\[I\left(\Pi\right)= H\left(G\right)-\mathbb{E}_{\Pi}\left[H\left(G\left(\cdot \mid z\right)\right)\right]\]  where $H\left(G\right)$ is the entropy of the valence state,  and $H\left(G\left(\cdot \mid z\right)\right)$ is the conditional entropy of the valence state given signal realization $z$. 
\end{assmO}

We proceed in two steps, first characterizing the optimal signals for any given symmetric policy profile and then examining their impacts for equilibrium policy polarization. The next proposition establishes the counterparts of Theorems \ref{thm_binary} and \ref{thm_product} for a continuum of states and entropy attention cost.  

\begin{propO}\label{prop_state}
For any symmetric policy profile $(-a,a)$ with $a \in (0, \overline{a}]$, the following are true under Assumption \ref{assm_regularity}, \ref{assm_utility}, and \ref{assm_entropy}.
\begin{enumerate}[(i)]
\item The optimal personalized signal  for any voter is unique, binary, induces \emph{(\ref{eqn_ob})} from its consumers, and satisfies the skewness properties stated in  Theorem \ref{thm_product}(ii).
\item Any optimal broadcast signal, whose support is denoted by $\mathcal{Z}$ and induced beliefs $(\mu_z)_{z \in \mathcal{Z}}$, has two or three signal realizations. Specifically, 
 \begin{enumerate}[(a)]
 \item if $|\mathcal{Z}|=2$, then $\mathcal{Z}=\{LL, RR\}$, $\mu_{LL}<0<\mu_{RR}=|\mu_{LL}|$, and  $v(\mathbf{a}, k)+\mu_{LL}<0<v(\mathbf{a}, k)+\mu_{RR}$ $\forall k$; 
 \item if $|\mathcal{Z}|=3$, then $\mathcal{Z}=\{LL, LR, RR\}$, $\mu_{LL}\left(a\right)<0$, $\mu_{LR}\left(a\right)=0$, $\mu_{RR}\left(a\right)=|\mu_{LL}\left(a\right)|>0$, and $v(\mathbf{a}, k)+\mu_{LL}<0<v(\mathbf{a}, k)+\mu_{RR}$ $\forall k$.
\end{enumerate}
\end{enumerate}
\end{propO}

The major difference between Proposition  \ref{prop_state} and its baseline counterpart lies in the broadcast case. In the main body of the paper, we already explained why this case can be analyzed by aggregating voters with binding participation constraints into a representative voter. Under the assumption that voters' preferences exhibit increasing differences, only voters of the most extreme types can have binding participation constraints, and the representative voter acting on their behalves faces three terminal decisions: LL, LR, RR (the first and second letters stand for the voting decisions of the most left-leaning and most right-leaning voters, respectively). Then from the Blackwell monotonicity of the attention cost function, it follows that the optimal broadcast signal prescribes at most three voting recommendation profiles that the representative voter strictly prefers to obey. The case of a single signal realization is ruled out by the feasibility condition, leaving two or three signal realizations as the only possibilities. 

The case of two signal realizations can be analyzed analogously as before. In the new case of three signal realizations, strict obedience implies that $v\left({\bf{a}}, k\right)+\mu_{LL}<0<v\left({\bf{a}}, k\right)+\mu_{RR}$ $\forall k \in\mathcal{K}$, and that $v(\mathbf{a}, -K)+\mu_{LR}<0<v({\bf{a}}, K)+\mu_{LR}$. To pin down $\mu_{LR}$, we argue, based on the the convexity of mutual information in the signal structure \citep{infotheory}, that the optimal broadcast signal must be symmetric, and hence the posterior mean of the state given LR must equal zero. Simple as it is, this result implies that when we endogenize candidates' policies,  the only symmetric policy profile that can arise in equilibrium is $(0,0)$, hence the transition from broadcast news aggregation to personalized news aggregation always strictly increases policy polarization.\footnote{The proof of this claim combines  Lemma \ref{lem_maxmin} with the standard median voter theorem. Specifically, from any symmetric policy profile $(-a,a)$ with $a>0$,  the deviation to $a'=0$ weakly increases candidate $R$'s winning probability when the recommendation profile is either LL or RR, because under these signal realizations, moving towards the center doesn't repel any one by Lemma \ref{lem_general_maxmin}. In addition, it strictly increases candidate $R$'s winning probability when the signal realization is LR, due to the standard median voter theorem. In the opposite direction, no deviation from $(0,0)$ to $a'>0$ increases candidate $R$'s winning probability when the recommendation profile is  LL or RR, because under these signal realizations, moving away from the center doesn't attract anyone by Lemma  \ref{lem_general_maxmin}. Moreover, it strictly reduces candidate $R$'s winning probability when the signal realization is LR, due to, once again, the standard median voter theorem.  }  

\section{Proofs}\label{sec_general_proof}

\paragraph{Proof of Lemma \ref{lem_coalition}}
Fix any segmentation technology $\mathcal{S}$, symmetric policy profile $(-a, a)$ with $a \geq 0$, $(\mathcal{S}, a)$-consistent signal distribution $(\bm{\chi}, {\bf{b}}^+, {\bf{b}}^-)$,  and population function $q$. Let $\bf{q}$ denote the $|\mathcal{K}|$-column vector that compiles the populations of voters $-K,\cdots, K$. Let the default be the strictly obedient outcome induced by the joint signal distribution. 

Define two matrix operations. First, for any $\mathcal{C}\subseteq \mathcal{K}$, let $\bm{\chi}_{\mathcal{C}}$ be the resulting matrix from replacing every row $k \in \mathcal{C}$ of $\bm{\chi}$ with a row of ones. Second, for any matrix ${\bf{A}}$, let $\widehat{{\bf{A}}}$ be the resulting matrix from rounding the entries of ${\bf{A}}$, i.e., replacing those entries above $1/2$ with $1$ and those below $1/2$ with zero. By definition, the row vector $\widehat{{\bf{q}}^{\top}\bm{\chi}}$ compiles candidate $R$'s default winning probabilities across the voting recommendation profiles that occur with strictly positive probabilities, and $(\widehat{{\bf{q}}^{\top}\bm{\chi}} {\bf{b}}^{+}+\widehat{{\bf{q}}^{\top}\bm{\chi}} {\bf{b}}^-)/2$ 
is candidate $R$'s default winning probability in expectation. After candidate $R$ commits a unilateral deviation from $(-a,a)$ that attracts a set $\mathcal{C} \subseteq \mathcal{K}$ of voters without affecting anything else, his winning probability vector becomes $\widehat{{\bf{q}}^{\top}\bm{\chi}_{\mathcal{C}}}$, and his expected winning probability becomes
$(\widehat{{\bf{q}}^{\top}\bm{\chi}_{\mathcal{C}}}{\bf{b}}^{+} +  \widehat{{\bf{q}}^{\top}\bm{\chi}_{\mathcal{C}}}{\bf{b}}^-)/2$.
Since $\widehat{{\bf{q}}^{\top}\bm{\chi}_{\mathcal{C}}} \geq \widehat{{\bf{q}}^{\top}\bm{\chi}}$, the deviation strictly increases candidate $R$'s winning probability in expectation if and only if it does so under some voting recommendation profile, i.e., $(\widehat{{\bf{q}}^{\top}\bm{\chi}_{\mathcal{C}}}{\bf{b}}^{+}+ \widehat{{\bf{q}}^{\top}\bm{\chi}_{\mathcal{C}}}{\bf{b}}^-)/2>(\widehat{{\bf{q}}^{\top}\bm{\chi}} {\bf{b}}^{+}+\widehat{{\bf{q}}^{\top}\bm{\chi}} {\bf{b}}^-)/2$ if and only if $\widehat{{\bf{q}}^{\top}\bm{\chi}_{\mathcal{C}}} \neq \widehat{{\bf{q}}^{\top}\bm{\chi}}.$ The last condition is  equivalent to $\mathcal{C}$ being an influential coalition, and it depends on $\mathcal{S}$, $(-a,a)$, $(\bm{\chi}, {\bf{b}}^+, {\bf{b}}^-)$, and $q$ only through the pair $(\bm{\chi}, q)$.  \qed

\begin{lemO}\label{lem_quadratic}
Under Assumptions \ref{assm_attention}, \ref{assm_utility},  and \ref{assm_quadratic}, the following must hold for any $ a \geq 0$ and $a' \in \left[-a,a\right]$. 
\begin{enumerate}[(i)]
\item $\phi^{p}\left(-a, a', k\right)$ is decreasing in $k$ on $\left\{k \in \mathcal{K}: k \leq 0\right\}$ and is increasing in $k$ on $\left\{k \in \mathcal{K}: k \geq 0\right\}$. 
\item $\phi^{p}\left(-a,a', k\right) \leq \phi^{p}\left(-a,a',-k\right)$ for any $
k>0$.
\end{enumerate}
\end{lemO}

\begin{proof}
Fix any $a$ and $a'$ as above. Under Assumption \ref{assm_quadratic},  solving the case of personalized news aggregation yields 
\begin{equation}\label{eqn_quadraticmuL}
\mu_L^{p}\left(a,k\right)=
\begin{cases}
-2v\left(-a,a,k\right)-1/\left(2\lambda\right) &\text{if } k \leq 0,\\
-1/\left(2\lambda\right) &\text{if } k>0.
\end{cases}
\end{equation}
Also recall that $\phi^{p}\left(-a,a',k\right)\coloneqq v\left(-a,a',k\right)+\mu_L^p\left(a,k\right)$. 
\bigskip

\noindent Part (i): If $ k \leq 0$, then 
\begin{align*}
\phi^{p}\left(-a,a',k\right)&=v\left(-a,a',k\right)-2v\left(-a,a,k\right)-\frac{1}{2\lambda}\\
&=u\left(a',k\right)-u\left(-a,k\right)-2\left[u\left(a,k\right)-u\left(-a,k\right)\right]-\frac{1}{2\lambda}\\
&=u\left(a',k\right)+u\left(-a,k\right)-2u\left(a,k\right)-\frac{1}{2\lambda}\\
&=-\left[v\left(a',a,k\right)+v\left(-a,a,k\right)\right]-\frac{1}{2\lambda}, 
\end{align*}
where the last line is decreasing in $k$ by Assumption \ref{assm_utility} \textbf{increasing differences}. If $k>0$, then 
$\phi^{p}\left(-a,a',k\right)=v\left(-a,a',k\right)-1/\left(2\lambda\right)$, which is increasing in $k$ again by  Assumption \ref{assm_utility} \textbf{increasing differences}.

\bigskip

\noindent Part (ii): Under Assumption \ref{assm_utility},  the following must hold for any $k>0$: 
\begin{align*}
&\phi^{p}\left(-a,a',k\right)-\phi^{p}\left(-a,a',-k\right)\\
&=v\left(-a,a',k\right)-\frac{1}{2\lambda}-\left[v\left(-a,a',-k\right)-2v\left(-a,a,-k\right)-\frac{1}{2\lambda}\right]\\
\tag{$\because$ \textbf{symmetry}}&=v\left(-a,a',k\right)-v\left(a, -a',k\right)-2v\left(-a,a,k\right)\\
&=\left[u\left(-a,k\right)-u\left(-a',k\right)\right]-\left[u\left(a,k\right)-u\left(a',k\right)\right]\\
\tag{$\because$ \textbf{symmetry}}&=v\left(a',a,-k\right)-v\left(a',a,k\right)\\
\tag{$\because$ \textbf{increasing differences}}& \leq 0. 
\end{align*}
\end{proof}

\paragraph{Proof of Lemma \ref{lem_monotone_a}}
\begin{proof}
We wish to verify that $\phi^{\mathcal{S}}\left(-a,a',k\right)\coloneqq v\left(-a,a',k\right)+\mu_L^{\mathcal{S}}\left(a,k\right)$ is increasing in $a$ on $\left[|t\left(k\right)|, \overline{a}\right]$ for any $k \in \mathcal{K}$ and $a' \in \mathcal{A}$. Since $v\left(-a,a',k\right)$ is strictly increasing in $a$ on $\left[|t\left(k\right)|, \overline{a}\right]$ by Assumption \ref{assm_utility} \textbf{inverted V-shape}, it suffices to show that $\mu_L^{\mathcal{S}}\left(a,k\right)$ is nondecreasing in $a$ on $\left[|t\left(k\right)|, \overline{a}\right]$, in the following cases. 

\bigskip

$\bullet$ $\mathcal{S}=b$.  Recall that $\mu_L^b\left(a\right)$ is the unique solution to  (\ref{eqn_musymmetry}), i.e., 
\[
\max_{\mu\in [-1,0]} h\left(\mu\right) \text{ s.t. } \frac{1}{2}\left[v\left(-a,a,-K\right)-\mu\right]^+ \geq \lambda h\left(\mu\right), \]
where $h$ is strictly convex and strictly decreasing on $\left[-1,0\right]$. Also note that $v\left(-a,a,-K\right)$ is decreasing in $a$, because under Assumption \ref{assm_utility}, the following must hold for any $a'>a \geq 0$: 
\begin{align*}
&v\left(-a', a', -K\right)-v\left(-a,a,-K\right)\\
&=u\left(a',-K\right)-u\left(-a', -K\right)-u\left(a,-K\right)+u\left(-a,-K\right)\\
\tag{$\because$ \textbf{symmetry}}&=u\left(a', -K\right)-u\left(a, -K\right)-\left[u\left(a',K\right)-u\left(a, K\right)\right]\\
&=v\left(a,a',-K\right)-v\left(a,a',K\right)\\
\tag{$\because$ \textbf{increasing differences}}& \leq 0.
\end{align*}
Combining these observations gives the desired result (drawing a picture will make the point clear). 

\bigskip

$\bullet$ $\mathcal{S}=p$ and $u\left(a,k\right)=-|t\left(k\right)-a|$. In this case, the fact that $v\left(-a,a,k\right)$  is invariant with $a$ (indeed, $\equiv 2t(k)$) on $\left[|t\left(k\right)|, \overline{a}\right]$ implies that $\mu_L^p\left(a,k\right)\equiv \mu_L^p\left(|t\left(k\right)|,k\right)$ on $\left[|t\left(k\right)|, \overline{a}\right]$.

\bigskip

$\bullet$ $\mathcal{S}=p$ and $h\left(\mu\right)=\mu^2$.  In this case, a careful inspection of the expression for $\mu_L^p\left(a,k\right)$ in (\ref{eqn_quadraticmuL}) gives the desired result.
\end{proof}

\begin{lemO}[label=lem_general_xi]
Let everything be as in Theorem \ref{thm_general_main}. Then for any  $k \in \left\{0,\cdots, K\right\}$ and any $\mathcal{D} \subseteq \left\{-k,\cdots,k \right\}$ such that $\mathcal{D} \cap \left\{-k,k\right\} \neq \emptyset$, the following must hold: $\xi^{\mathcal{S}}\left(\mathcal{D}\right)>t\left(k\right)$ and  
$\left[t\left(k\right), \overline{a}\right] \cap \Xi^{\mathcal{S}}\left(\mathcal{D}\right)=\left[t\left(k\right), \xi^{\mathcal{S}}\left(\mathcal{D}\right)\right]$. 
\end{lemO}

\begin{proof}
Fix any $k$ and $\mathcal{D}$ as above. Recall that 
\[\Xi^{\mathcal{S}}\left(\mathcal{D}\right)\coloneqq \left\{a \geq 0: \max_{a' \in \mathcal{A}}\phi^{\mathcal{S}}\left(-a,a', \mathcal{D}\right) \leq 0\right\},\]
where 
$\phi^{\mathcal{S}}\left(-a, a', \mathcal{D}\right) \coloneqq  \displaystyle \min_{k' \in \mathcal{D}} \phi^{\mathcal{S}}\left(-a, a', k'\right)$. Let $t(\mathcal{D})$ denote the image of $\mathcal{D}$ under mapping $t$, and write $\widetilde{\mathcal{D}}$ for $\left[\min t\left(\mathcal{D}\right), \max t\left(\mathcal{D}\right)\right]$.  Under Assumption \ref{assm_utility} \textbf{inverted V-shape}, we can restrict attention to $a' \in \tilde{\mathcal{D}}$ in the definition of $\Xi^{\mathcal{S}}(\mathcal{D})$, i.e., 
\[
\Xi^{\mathcal{S}}\left(\mathcal{D}\right)=\left\{a \geq 0: \max_{a' \in \widetilde{\mathcal{D}}}\phi^{\mathcal{S}}\left(-a,a', \mathcal{D}\right) \leq 0\right\}.
\]

Fix the policy profile to be $(-t\left(k\right), t\left(k\right))$, as well as any $a' \in \widetilde{\mathcal{D}}$. From Assumption \ref{assm_utility} and (\ref{eqn_ob}), it follows that $a'$ doesn't attract type $k$ voters:
\begin{align*}
\phi^{\mathcal{S}}\left(-t\left(k\right), a' ,k\right) 
& \coloneqq  
v\left(-t\left(k\right), a',k\right)+\mu_L^{\mathcal{S}}\left(t\left(k\right),k\right)
 \\ 
\tag{$\because$ \textbf{inverted V-shape}} &\leq 
v\left(-t\left(k\right), t\left(k\right),k\right)+\mu_L^{\mathcal{S}}\left(t\left(k\right),k\right)\\
\tag{$\because$ \ref{eqn_ob}}&<0,
\end{align*}
and it doesn't attract type $-k$ voters, either: 
\begin{align*}
&\phi^{\mathcal{S}}\left(-t\left(k\right), a' ,-k\right)\\
&\coloneqq v\left(-t\left(k\right),a', -k\right)+\mu_L^{\mathcal{S}}\left(t\left(k\right),-k\right) \\
\tag{$\because$ \textbf{inverted V-shape}}&  \leq v\left(-t\left(k\right),t\left(-k\right), -k\right)+\mu_L^{\mathcal{S}}\left(t\left(k\right),-k\right)\\
\tag{$\because$ \textbf{symmetry}}& =0+\mu_L^{\mathcal{S}}\left(t\left(k\right),-k\right)\\
&<0.
\end{align*}
Thus $\phi^{\mathcal{S}}\left(-t\left(k\right), a', \mathcal{D}\right) \coloneqq  \displaystyle \min_{k' \in \mathcal{D}} \phi^{\mathcal{S}}\left(-t\left(k\right),a', k'\right)<0$, and taking maximum over $a'$ yields $\displaystyle\max_{a' \in\widetilde{\mathcal{D}}} \phi^{\mathcal{S}}\left(-t\left(k\right), a', \mathcal{D}\right)<0$. Meanwhile, Assumption \ref{assm_monotone_a} implies that $\phi^{\mathcal{S}}\left(-a, a', \mathcal{D}\right)$ is increasing in $a$ on $\left[t\left(k\right), \overline{a}\right]$ for any  $a'$, and taking maximum over $a'$ yields: $\forall a_2>a_1 \geq t\left(k\right)$, 
\begin{align*}
\displaystyle\max_{a' \in \widetilde{\mathcal{D}}} \phi^{\mathcal{S}}\left(-a_1,a', \mathcal{D}\right)
&=\phi^{\mathcal{S}}\left(-a_1, \argmax_{a' \in \widetilde{\mathcal{D}}} \phi^{\mathcal{S}}\left(-a_1,a', \mathcal{D}\right), \mathcal{D}\right)\\
&\leq \phi^{\mathcal{S}}\left(-a_2, \argmax_{a' \in \widetilde{\mathcal{D}}} \phi^{\mathcal{S}}\left(-a_1,a', \mathcal{D}\right), \mathcal{D}\right)\\
&\leq \displaystyle\max_{a' \in \widetilde{\mathcal{D}}}\phi^{\mathcal{S}}\left(-a_2,a', \mathcal{D}\right), 
\end{align*}
i.e., $\displaystyle\max_{a' \in \widetilde{\mathcal{D}}} \phi^{\mathcal{S}}\left(-a, a', \mathcal{D}\right)$ is increasing in $a$ on $\left[t\left(k\right), \overline{a}\right]$. Taken together, we conclude that $\mathcal{D}$'s policy latitude exceeds $t\left(k\right)$:
\[\xi^{\mathcal{S}} \left(\mathcal{D}\right)\coloneqq \max\Xi^{\mathcal{S}}\left(\mathcal{D}\right)=\max\left\{a \geq 0: \max_{a' \in \widetilde{\mathcal{D}}}\phi^{\mathcal{S}}\left(-a,a', \mathcal{D}\right) \leq 0\right\}>t\left(k\right),\]
and that all policies in $\left[t\left(k\right), \xi^{\mathcal{S}} \left(\mathcal{D}\right)\right]$ belong to the $\mathcal{D}$-proof set: 
\begin{align*}
\left[ t\left(k\right),\overline{a}\right] \cap \Xi^{\mathcal{S}}\left(\mathcal{D}\right)= \left\{a \geq t\left(k\right): \max_{a' \in \widetilde{\mathcal{D}}}\phi^{\mathcal{S}}\left(-a, a', \mathcal{D}\right) \leq 0\right\} = \left[t\left(k\right), \xi^{\mathcal{S}}\left(\mathcal{D}\right)\right]. 
\end{align*}
\end{proof}

\paragraph{Proof of Theorem \ref{thm_general_main}}
Fix any segmentation technology $\mathcal{\mathcal{S}} \in \left\{b,p\right\}$, $\mathcal{S}$-consistent configuration $\bm{\chi}$, and population function $q$. Let $\mathcal{C}$ denote a typical influential coalition formed under $(\bm{\chi}, q)$. For each $k=0,\cdots, K-1$, define
\[A\left(k\right)= \begin{cases}
[t\left(k\right), t\left(k+1\right))\cap \displaystyle\bigcap_{\mathcal{C} \subseteq \left\{-k,\cdots, k\right\}} \Xi^{\mathcal{S}}\left(\mathcal{C}\right) & \text{ if } \exists \mathcal{C} \subseteq \{-k,\cdots, k\},\\
[t\left(k\right), t\left(k+1\right))& \text{ else}.
\end{cases}\]
For $k=K$, define 
\[A\left(K\right)=\left[t\left(K\right), \overline{a}\right]\cap \displaystyle\bigcap_{\mathcal{C} } \Xi^{\mathcal{S}}\left(\mathcal{C}\right). \]
Lemma \ref{lem_general_maxmin} shows that 
\[\mathcal{E}^{\mathcal{S},\bm{\chi},q}=\bigcup_{k=0}^{K} A\left(k\right). \] Below we prove by induction that $\cup_{k=0}^K A\left(k\right)=\left[0, \displaystyle \min_{\mathcal{C}} \xi^{\mathcal{S}}\left(\mathcal{C}\right)\right]$.

\begin{description}
\item[Step 0.] Letting $k=0$ in Lemma \ref{lem_general_xi} shows that $\Xi^{\mathcal{S}}\left(\left\{0\right\}\right)=\left[0, \xi^{\mathcal{S}}\left(\left\{0\right\}\right)\right]$, so 
\[A\left(0\right)=
\begin{cases}
\left[0, \xi^{\mathcal{S}}\left(\left\{0\right\}\right)\right] &\text{ if } \{0\} \text{ is influential and } \xi^{\mathcal{S}}\left(\left\{0\right\}\right)<t\left(1\right),\\
[0, t\left(1\right)) &\text{ else}.
\end{cases}\]
In the first case,  $A\left(k\right)\subseteq [t\left(k\right), t\left(k+1\right)]  \cap \Xi^{\mathcal{S}}\left(\{0\}\right)=\emptyset$ for any $k \geq 1$, and $\displaystyle \min_{\mathcal{C}} \xi^{\mathcal{S}}\left(\mathcal{C}\right)=\xi^{\mathcal{S}}\left(\left\{0\right\}\right)$ because $\xi^{\mathcal{S}}\left(\mathcal{C}\right)>t\left(1\right)$ for any $\mathcal{C} \neq \{0\}$ by Lemma \ref{lem_general_xi}. Taken together, we conclude that $\cup_{k=0}^K A\left(k\right)=\left[0, \displaystyle \min_{\mathcal{C}}\xi^{\mathcal{S}}\left(\mathcal{C}\right)\right]$ and terminate the procedure. In the second case, we proceed to the next step. 

\item[Step m.] The output of Step $m-1$ is $\displaystyle \cup_{k=0}^{m-1} A\left(k\right)=[0, t\left(m\right))$. Then from Lemma \ref{lem_general_xi}, which shows that $\left[t\left(m\right), \overline{a}\right] \cap \Xi^{\mathcal{S}}\left(\mathcal{C}\right)=\left[t\left(m\right),\xi^{\mathcal{S}}\left(\mathcal{C}\right)\right]$ for any $\mathcal{C} \subseteq \left\{-m, \cdots, m\right\}$ such that $\mathcal{C} \cap \left\{-m,m\right\}\neq \emptyset$, it follows that 
\begin{align*}
\displaystyle \cup_{k=0}^m A\left(k\right)=\begin{cases}
\left[0, \displaystyle \min_{\mathcal{C} \subseteq \{-m,\cdots,m\}} \xi^{\mathcal{S}}\left(\mathcal{C}\right)\right] &\text{ if } \displaystyle \min_{\mathcal{C} \subseteq \{-m,\cdots,m\}} \xi^{\mathcal{S}}\left(\mathcal{C}\right)<t\left(m+1\right),\\
[0, t\left(m+1\right)) &\text{ else}.
\end{cases}
\end{align*}
In the first case,  $A\left(k\right) \subseteq \left[t\left(k\right), t\left(k+1\right) \right] \cap \displaystyle\cap_{\mathcal{C} \subseteq \left\{-m, \cdots, m\right\}} \xi^{\mathcal{S}}\left(\mathcal{C}\right)= \emptyset$ for any $k \geq m+1$, and $\displaystyle \min_{\mathcal{C}} \xi^{\mathcal{S}}\left(\mathcal{C}\right)=\displaystyle \min_{\mathcal{C} \subseteq \{-m,\cdots,m\}} \xi^{\mathcal{S}}\left(\mathcal{C}\right)$ because $\xi^{\mathcal{S}}\left(\mathcal{C}'\right)>t\left(m+1\right)$ for any $\mathcal{C'} \nsubseteq \left\{-m,\cdots, m\right\}$ by Lemma \ref{lem_general_xi}.  Taken together, we conclude that $\cup_{k=0}^K A\left(k\right)=\left[0, \displaystyle \min_{\mathcal{C}}\xi^{\mathcal{S}}\left(\mathcal{C}\right)\right]$ and terminate the procedure. In the second case, we proceed to the next step.
\end{description}

The above procedure terminates in at most $K+1$ steps, and the output is always $\cup_{k=0}^K A\left(k\right)=\left[0, \displaystyle \min_{\mathcal{C}}\xi^{\mathcal{S}}\left(\mathcal{C}\right)\right]$. \qed

\paragraph{Proof of Proposition \ref{prop_general_population}}
We wish to demonstrate that $\displaystyle \min_{\mathcal{C}s \text{ formed under } (\bm{\chi}, q)}\xi^{p}\left(\mathcal{C}\right) \\ \geq \displaystyle \min_{\mathcal{C}s \text{ formed under } (\bm{\chi}, q')} \xi^p\left(\mathcal{C}\right)$ for any $p$-consistent $\bm{\chi}$ and any population functions $q$ and $q'$ such that $q \succeq_{mass} q'$. The proof exploits the following consequences of Lemma \ref{lem_quadratic}: For any $a \geq 0$ and $a' \in \left[-a,a\right]$, (i) $\phi^p\left(-a,a',-K\right)=\displaystyle \max_{k \in \mathcal{K}} \phi^p\left(-a,a',k\right)$, and (ii) $\phi^p\left(-a,a',k\right)$ is decreasing in $k$ on $\left\{k: k \leq 0\right\}$ and is increasing in $k$ on $\left\{k: k \geq 0\right\}$. 

\paragraph{Step 1.} Show that $\xi^p\left(\mathcal{D}\right)>t\left(K\right)$ for any $\mathcal{D} \subseteq \mathcal{K}$. Fix any $a' \in \left[t\left(-K\right), t\left(K\right)\right]$ and any $\mathcal{D}\subseteq \mathcal{K}$.  Notice two things. First, 
\begin{align*}
&\phi^p\left(-t\left(K\right),a', \mathcal{D}\right) \\
&\coloneqq \min_{k \in \mathcal{D}} \phi^p\left(-t\left(K\right),a',k\right) \\
&\leq \max_{k \in \mathcal{D}} \phi^p\left(-t\left(K\right),a',k\right)\\
\tag{$\because$ Lemma \ref{lem_quadratic}}& \leq  \phi^p\left(-t\left(K\right), a', -K\right)\\
\tag{$\because$ Assumption \ref{assm_utility} \textbf{inverted V-shape}}& \leq \phi^p\left(-t\left(K\right), t\left(-K\right), -K\right)\\
& \coloneqq v\left(-t\left(K\right), t\left(-K\right), K\right)+\mu_L^p\left(t\left(K\right),- K \right)\\
\tag{$\because$ Assumption \ref{assm_utility} \textbf{symmetry}}& =0+\mu_L^p\left(t\left(K\right),- K \right)\\
& <0.
\end{align*}
Second, since $\phi^p\left(-a,a',k\right)$ is increasing in $a$ on $\left[t\left(K\right), \overline{a}\right]$ for any $k \in \mathcal{D}$ by Lemma \ref{lem_monotone_a}, $\phi^p\left(-a,a', \mathcal{D}\right) \coloneqq \displaystyle\min_{k \in \mathcal{D}} \phi^p\left(-a,a', k\right)$ is increasing in $a$ on $\left[t\left(K\right), \overline{a}\right]$, too. Combining these observations yields
\begin{align*}
\xi^p\left(\mathcal{D}\right)&\coloneqq \max\left\{a \geq 0: \max_{a' \in \left[t\left(-K\right), t\left(K\right)\right]} \phi^p\left(-a,a',\mathcal{D}\right)\leq 0\right\}\\
&=\max\left\{a \geq t\left(K\right): \max_{a' \in \left[t\left(-K\right), t\left(K\right)\right]} \phi^p\left(-a,a',\mathcal{D}\right)\leq 0\right\}\\
&> t(K). 
\end{align*}

\paragraph{Step 2.} There are three kinds of influential coalitions: (a) $\max \mathcal{C} \leq 0$, (b) $\min \mathcal{C} \geq 0$, and (c) $\min \mathcal{C}<0<\max \mathcal{C}$. Consider case (a), and notice two things. First,  the following are equivalent for any $a \geq t\left(K\right)$ and any $a' \in \left[-a,a\right]$ by Lemma \ref{lem_quadratic}: (i) $\phi^p\left(-a,a',\mathcal{C}\right) \leq 0$, (ii) $\phi^p\left(-a,a', \max \mathcal{C}\right) \leq 0$, and (iii) $\phi^p\left(-a,a', \{k: k \leq \max \mathcal{C}\}\right) \leq 0$. Second, since $\mathcal{C}$ is influential and $\mathcal{C} \subseteq \left\{k: k \leq \max \mathcal{C}\right\}$, $\left\{k: k \leq \max \mathcal{C}\right\}$ is influential, too. Combining these observations yields
\begin{equation}\label{eqn_quadraticcoalition}
\min_{\substack{\mathcal{C}s \text{ formed under } (\bm{\chi}, q)\\ \text{ s.t. }\max\mathcal{C} \leq 0}} \xi^p\left(\mathcal{C}\right)=\min_{\substack{\mathcal{C}s \text{ formed under } (\bm{\chi}, q) \\ \text{ s.t. }\mathcal{C}=\{k: k \leq \alpha\}, \alpha \leq 0}} \xi^p\left(\mathcal{C}\right).
\end{equation}
A close inspection of (\ref{eqn_quadraticcoalition}) reveals two things. First, 
\[\xi^p\left(\left\{k: k \leq \alpha\right\}\right) = \max \left\{a \geq t\left(K\right): \max_{a' \in \left[-t\left(K\right), t\left(K\right)\right]} \phi^p\left(-a, a', \left\{k: k \leq \alpha\right\}\right) \leq 0\right\}\]
is increasing in $\alpha$ on $\left\{\alpha: \alpha \leq 0\right\}$ by Lemma \ref{lem_quadratic}. Second, every set $\{k: k \leq \alpha\}$ with $\alpha<0$ is more likely to be influential under $q'$ than under $q$ because $q \succeq_{mass} q'$. Thus
\[\min_{\substack{\mathcal{C}s \text{ formed under } (\bm{\chi}, q) \\\text{s.t.} \mathcal{C}=\{k: k \leq \alpha\}, \alpha \leq 0}} \xi^p\left(\mathcal{C}\right)  \geq \min_{\substack{\mathcal{C}s \text{ formed under } (\bm{\chi}, q') \\\text{s.t. } \mathcal{C}=\{k: k \leq \alpha\}, \alpha \leq 0}} \xi^p\left(\mathcal{C}\right),\]
which proves the desired result for case (a). The proofs for cases (b) and (c) are similar and thus are omitted for brevity.  \qed

\paragraph{Proof of Proposition \ref{prop_competitive}} Recall that for any given policy profile $(-a,a)$ with $a \in [0, \overline{a}]$, the monopolistic personalized signal for type $k$ voters is their competitive signal when the attention cost parameter equals $\lambda-1/\gamma$ for some $\gamma>0$. From Proposition \ref{prop_lambda}, it follows that $\mu_L^p(a,k)<\mu_L^c(a,k)$, and hence  that $\forall a'$ and $\mathcal{D} \subseteq \mathcal{K}$:
\begin{align*}
\phi^c\left(-a,a' ,\mathcal{D}\right)\coloneqq \min_{k \in \mathcal{D}} \phi^c\left(-a,a',k\right)&\coloneqq \min_{k \in \mathcal{D}} v\left(-a,a',k\right)+\mu_L^c\left(a,k\right)\\
&>\min_{k \in \mathcal{D}} v\left(-a,a',k\right)+\mu_L^p\left(a,k\right) \coloneqq \phi^p\left(-a,a',\mathcal{D}\right).
\end{align*}
Substituting this result into the proof of Lemma \ref{lem_general_xi} yields $\xi^c\left(\mathcal{D}\right)\leq \xi^p\left(\mathcal{D}\right)$, where $\xi^c\left(\mathcal{D}\right)$ denotes $\mathcal{D}$'s policy latitude in the competitive case. Thus for any $c$-consistent $\bm{\chi}$ and $p$-consistent $\bm{\chi}'$ such that $\bm{\chi} \succeq \bm{\chi}'$, the following must hold: 
\begin{align*}
\tag{$\because$ Theorem \ref{thm_general_main}; $\bm{\chi}$ is $c$-consistent}\mathcal{E}^{c,\bm{\chi}, \rho}
&=\left[0, \min_{ \mathcal{C}s \text{ formed under } 
(\bm{\chi}, \rho)} \xi^c\left(\mathcal{C}\right)\right] \\
\tag{$\because$ Proposition \ref{prop_richness}; $\bm{\chi} \succeq \bm{\chi}'$}&\subseteq \left[0, \min_{ \mathcal{C}s \text{ formed under } 
(\bm{\chi}', \rho)} \xi^c\left(\mathcal{C}\right)\right]\\
\tag{$\because$ $\xi^c\left(\mathcal{C}\right)<\xi^p\left(\mathcal{C}\right)$}&\subseteq \left[0, \min_{ \mathcal{C}s \text{ formed under } 
(\bm{\chi}', \rho)} \xi^p\left(\mathcal{C}\right)\right]\\
\tag{$\because$ Theorem \ref{thm_general_main}; $\bm{\chi}'$ is $p$-consistent}&=\mathcal{E}^{p,\bm{\chi}', \rho}. 
\end{align*}
\qed

\paragraph{Proof of Proposition \ref{prop_state}}  The proof presented below fixes any $\mathbf{a}=(-a,a)$ with $a \in (0, \overline{a}]$, and strengthens Assumption \ref{assm_utility} \textbf{increasing differences} to strict increasing differences, i.e., $v\left({\bf{a}},k\right)$ is strictly increasing in $k$. Doing so is w.l.o.g. because in the case where
$v\left({\bf{a}}, k\right)=v\left({\bf{a}},k+1\right)$ for some $k$, we can treat type $k$ and $k+1$ voters as a single entity. 

 \paragraph{Personalized case} The proof follows that of Theorem \ref{thm_binary} and \ref{thm_product} closely, with important caveats. 

\subparagraph{Step 1.} Characterize the solution to the dual problem (\ref{eqn_personalized_simplified}) for any given $\gamma \geq 0$. Show that the solution is unique and has two signal realizations. By \cite{matejka2015}, the solution to (\ref{eqn_personalized_simplified}), denoted by $\Pi(\gamma)$, is $\overline{\Pi}$ if and only if $\gamma \in [0, 1/\lambda]$. For any $\gamma>1/\lambda$, $\Pi(\gamma)$ is unique and has at most two signal realizations. The case of one signal realization, i.e., a degenerate signal, is ruled out by the feasibility condition as in the baseline model, leaving two signal realizations as the only possibility. 

\subparagraph{Steps 2-3.} Show that strong duality holds, and that the optimal signal is unique. These steps are much simpler than their baseline counterparts, as we can now pin down $\Gamma\coloneqq \{\gamma \geq 0:\Pi(\gamma)=\overline{\Pi}\}$ to $[0, 1/\lambda]$ using  \cite{matejka2015}. Everything else is the same as before. 

\subparagraph{Step 4.} Show that if the optimal signal differs from $\overline{\Pi}$ (as required by Assumption \ref{assm_regularity}(ii)), then it must satisfy the skewness property stated in Theorem \ref{thm_product}(ii). This step differs from its baseline counterpart. Specifically, let $\gamma(k)$ denote the Lagrange multiplier associated with type $k$ voters' participation constraints, and note that $\gamma(k) >1/\lambda$ must hold in order to satisfy $\Pi(\gamma(k)) \neq \overline{\Pi}$. For ease of notation, write $\hat{v}$ for $v(\mathbf{a}, k)/(\lambda-1/\gamma(k))$, $\hat{\omega}$ for $\omega/(\lambda-1/\gamma(k))$, $\hat{G}$ for the c.d.f. of $\hat{\omega}$, and $\rho$ for the likelihood that type $k$ voters votes candidate $R$ rather than candidate $L$. Since the optimal signal is binary, $\rho$ must belong to $(0,1)$. By \cite{matejka2015}, the probability that a type $k$ voter votes for candidate $R$ in state $\omega$ equals 
$\frac{\mathcal{\rho} \exp\left(\hat{v}+\hat{\omega}\right)}{\rho \exp\left(\hat{v}+\hat{\omega}\right)+1}
$, and the average probability that he votes for candidate $R$ equals $\frac{\rho}{\rho+1}$. Bayes' rule mandates that 
\[
\underbrace{\int_{\hat{\omega} \in \mathbb{R}} \frac{\rho\exp\left(\hat{v}+\hat{\omega}\right)}{\rho \exp\left(\hat{v}+\hat{\omega}\right)+1} d\hat{G}(\hat{\omega})}_{\text{LHS}(\rho, \hat{v})}=\underbrace{\frac{\rho}{\rho+1}}_{\text{RHS}(\rho)}. 
\]

When $\hat{v}=0$ (equivalently, $k=0$), 
\begin{align*}
\tag{$\because \hat{G}$ is symmetric} \text{LHS}(\rho, 0)&=
\int_{0}^{\infty} \frac{\rho\exp\left(\hat{\omega}\right)}{\rho\exp\left(\hat{\omega}\right)+1} + \frac{\rho\exp\left(-\hat{\omega}\right)}{\rho\exp\left(-\hat{\omega}\right)+1} d\hat{G}\left(\hat{\omega}\right)\\
&= \int_0^{\infty} \frac{\rho(2\rho+X(\hat{\omega}))}{\rho^2+\rho X(\hat{\omega})+1} d\hat{G}(\hat{\omega}),
\end{align*}
where $X(\hat{\omega}) \coloneqq \exp(\hat{\omega})+\exp(-\hat{\omega})>2$ almost surely. 
Subtracting RHS($\rho$) from the last expression yields
\[\text{LHS}(\rho, 0)-\text{RHS}(\rho)=\int_0^{\infty}\frac{(X(\hat{\omega})-2)(1-\rho)\rho}{(\rho^2+\rho X(\hat{\omega})+1)(\rho+1)}d\hat{G}(\hat{\omega}),\]
which, upon a close inspection, reveals that RHS($\rho$) single-crosses LHS$(\rho, 0)$ from below at $\rho=1$. To complete the proof, note that LHS$(\rho, \hat{v})$ is increasing in $\hat{v}$ for any $\rho$, and that LHS$(\rho, \hat{v})$ and RHS$(\rho)$ are both increasing in $\rho$. Thus, the root of LHS$(\rho, \hat{v})=$ RHS($\rho$) must be greater than one if $\hat{v}>0$ (equivalently $k>0$), and it must be smaller than one if $\hat{v}<0$ (equivalent $k<0$) as desired. 

\paragraph{Broadcast case} Since voters' preferences satisfy strict increasing differences, only voters of the most extreme types $-K$ and $K$ can have binding participation constraints, whereas those of interim types must have slack participation constraints. 
The remaining proof proceeds in three step as in the baseline case.  

\subparagraph{Step 1.} Characterize the solution to the dual problem (\ref{eqn_broadcast_simplified}) for any $\bm\gamma \geq \bm 0$. Show that the solution is unique and has two or three signal realizations. 

We prove the same result as before but using a different method. By \cite{matejka2015}, the solution to (\ref{eqn_broadcast_simplified}) is $\overline{\Pi}$ if and only if $\gamma_{-K}+\gamma_K \leq 1/\lambda$, i.e., $\Gamma\coloneqq \{\bm\gamma \geq \bm 0: \overline{\Pi}\in\{\Pi(\bm\gamma)\}\}=\{\bm\gamma \geq \bm 0: \gamma_{-K}+\gamma_K\leq 1/\lambda\}$. For any $\bm\gamma \notin \Gamma$, recall that  (\ref{eqn_broadcast_simplified}) is the optimal information acquisition problem faced by a representative voter who weighs the two extreme voters by $\bm\delta=(\delta_{-K}, \delta_K) =(\frac{\gamma_{-K}}{\gamma_{-K}+\gamma_K}, \frac{\gamma_{K}}{\gamma_{-K}+\gamma_K})$ and faces an effective attention cost parameter $\lambda(\bm\gamma) = \lambda-(\gamma_{-K}+\gamma_K)^{-1}$, i.e., $\max_{\Pi} \bm\delta \cdot \bm V(\Pi)-\lambda(\bm\gamma)I(\Pi)$. Since the attention cost function is Blackwell monotone, any  solution to this problem has at most three signal realizations LL, LR, and RR, where the first and second letters stand for the voting recommendations to type $-K$ voters and type $K$ voters, respectively.  The possibility of a degenerate signal is ruled out by the feasibility condition as before, leaving two or three signal realizations as the only possibilities. Regardless of which situation we we end up with, the solution is always unique.  This is because in state $\omega$, the representative voter's payoff from taking action $z \in \{LL, LR, RR\}$ is given by 
\[u_z(\omega)=\begin{cases}
-\delta_K(v(\mathbf{a},K)+\omega) &\text{ if } z=LL,\\
0 & \text{ if } z=LR,\\
\delta_{-K}(v(\mathbf{a}, -K)+\omega) & \text{ if } z=RR. 
\end{cases}\]
Since the random variables $\exp(u_z(\omega)/\lambda(\bm\gamma))$, $z \in \{LL, LR, RR\}$, are linearly independent with unit scaling,\footnote{That is, for every $z$, there doesn't exist $\{\alpha_{z'}\}_{z' \neq z}$ with $\sum_{z' \neq z}\alpha_z=1$ such that $\exp(u_z(\omega)/\lambda(\bm \gamma))=\sum_{z' \neq z} \alpha_{z'}\exp(u_{z'}(\omega)/\lambda(\bm\gamma))$ almost surely.} the solution to (\ref{eqn_broadcast_simplified}) is unique by Lemma 2 of the online appendix of \cite{matejka2015}. Hereinafter we shall denote this solution by $\Pi(\bm\gamma)$, and note that it is fully determined by (i) the vector of the average decision probabilities induced by $\Pi(\bm\gamma)$, hereinafter denoted by $\bm\pi(\bm\gamma)$, as well as (ii) the multinomial logit formula for conditional probabilities. Thus whenever convenient, we shall use $\bm\pi(\bm\gamma)$ to represent $\Pi(\bm\gamma)$, and note that $\bm\pi(\bm\gamma)$ is continuous in $\bm\gamma$ by Berge's maximum theorem. 

\subparagraph{Step 2. } Show that strong duality holds. The proof for the baseline case involves several substeps.

$\bullet$ The first step was to show that $d^*$ is attained at some finite $\bm\gamma^*$ such that $\mathcal{B}_{\epsilon}(\bm\gamma^*) \subseteq \mathbb{R}_+^2-\Gamma$ for some $\epsilon>0$. That argument carries seamlessly over to the current context. 

$\bullet$ The second step was to show that $\{\Pi(\bm\gamma)\}$ is a singleton and $\Pi(\bm\gamma)$ is binary for all $\bm\gamma \in \mathcal{B}_{\epsilon}(\bm\gamma^*)$.  Here, $\{\Pi(\bm\gamma)\}$ remains a singleton, but $\Pi(\bm\gamma)$ could have two or  three signal realizations. Regardless of which situation we end up with, we can proceed to the next step. 

$\bullet$ The third step was to show that 
\[\frac{\partial}{\partial \gamma_k}\mathcal{L}(\Pi(\bm\gamma), \bm\gamma)\bigg|_{\bm\gamma=\bm\gamma^*}=V(\Pi(\bm\gamma^*); k)-\lambda I(\Pi(\bm\gamma^*))=0 \text{ }\forall k \in \{-K, K\}.\]
Here, we can prove the same result, but using a slightly different argument: By the envelope theorem, $\mathcal{L}(\Pi(\bm\gamma), \bm\gamma)$ is differentiable in $\bm\gamma$ almost surely on $\mathcal{B}_{\epsilon}(\bm\gamma^*)$. Whenever the derivative exists, it is given by 
\[\frac{\partial}{\partial \gamma_k}\mathcal{L}(\Pi(\bm\gamma), \bm\gamma)=V(\bm\pi(\bm\gamma); k)-\lambda I(\bm\pi (\bm\gamma)) \text{ }\forall k \in \{-K, K\},\]
where the term $\bm\pi(\bm\gamma)$ in  the above expression represents the vector of the average decision probabilities induced by $\Pi(\bm\gamma)$. Since  $\bm\pi(\bm\gamma)$ is continuous in $\bm\gamma$,  $\mathcal{L}(\Pi(\bm\gamma), \bm\gamma)$ must differentiable in $\bm\gamma$, rather than being just absolutely continuous in $\bm\gamma$, on $\mathcal{B}_{\epsilon}(\bm\gamma^*)$. Everything else is the same as before. 

\paragraph{Step 4.} Show that any optimal broadcast signal must be symmetric. This step differs completely from its baseline counterpart. As before, let $\mathcal{Z}$ denote the support of the signal, $\pi_z$ denote the average probability that the signal realization is $z \in \mathcal{Z}$, and $\mu_z$ denote the posterior belief induced by the signal realization $z$. There are two cases to consider.  

\bigskip
$\bullet$ $|\mathcal{Z}|=2$.  In this case, the signal must recommend LL and RR to the representative voter, i.e., $\mathcal{Z}=\{LL, RR\}$. If, instead, $\mathcal{Z}=\{LL, LR\}$ or $\{LR, RR\}$, then the very type of the extreme voters who receives the same voting recommendation all the time would abstain from news consumption in order to save on the attention cost. This violates Assumption \ref{assm_regularity}(iii),  saying that it is strictly optimal to include all voters in news consumption. Back to the case of $\mathcal{Z}=\{LL, RR\}$, we know, from the previous argument, that both voting recommendations must be strictly obeyed by the representative voter, i.e.,  $v(\mathbf{a}, K)+\mu_{LL}<0<v(\mathbf{a}, -K)+\mu_{RR}$. Since $v(\mathbf{a}, k)$ is strictly increasing in $k$, $v(\mathbf{a}, k)+\mu_{LL}<0<v(\mathbf{a}, k)+\mu_{RR}$ must hold for all $k$. 

It remains to show that $|\mu_{LL}|=\mu_{RR}$. To this end, we prove a stronger claim, namely the signal structure, denoted by $\Pi$, is symmetric,  i.e., $\Pi\left(LL\mid \omega\right)=\Pi\left(RR \mid -\omega\right)$ almost surely.  Suppose the contrary is true, and construct a new signal structure $\Pi': \Omega \rightarrow \Delta(\mathcal{Z})$ whereby $\Pi'\left(LL \mid \omega\right)=\Pi\left(RR \mid -\omega\right)$ for all $\omega$. For each $z \in \mathcal{Z}$, write $\pi'_{z}$ for $\int \Pi'\left(z \mid \omega\right) dG\left(\omega\right)$ and $\mu'_z$ for $\int \omega \Pi'\left(z \mid \omega\right) dG\left(\omega\right)/\pi_z'$. By construction, the following must hold: $\pi'_{LL}=\pi_{RR}$, $\pi'_{RR}=\pi_{LL}$, $\mu'_{LL}=-\mu_{RR}$, $\mu'_{RR}=-\mu_{LL}$, and $I\left(\Pi\right)=I\left(\Pi'\right)$. Thus
 \begin{align*}
 V\left(\Pi'; {\bf{a}}, -K\right)
& =\pi'_{RR}\left[v\left({\bf{a}}, -K\right)+\mu'_{RR}\right]   \\
\tag{$\because$ symmetry} &=\pi_{LL}\left[-v\left({\bf{a}}, K\right)-\mu_{LL}\right] \\
 &=V\left(\Pi; {\bf{a}}, K\right) \\
\tag{$\because$ $K$'s participation constraint is binding} &=\lambda I\left(\Pi\right)\\ 
\tag{$\because$ $-K$'s participation constraint is binding} &=V\left(\Pi; {\bf{a}}, -K\right), 
 \end{align*}
 and $V\left(\Pi'; {\bf{a}}, K\right)=V\left(\Pi; {\bf{a}}, K\right)$ can be shown analogously. Compared to $\Pi$ and $\Pi'$, the signal structure that randomizes between them with equal probability generates the same consumption utility to the representative voter, and yet incurs a strictly lower attention cost because mutual information is strictly convex in the signal structure (see, e.g., Theorem 2.7.4. of \citealp{infotheory}). This means that $\Pi$ isn't optimal to the representative voter, a contradiction. 

\bigskip

$\bullet$ $\mathcal{Z}=\left\{LL, LR, RR\right\}$.  In this case, both types of extreme voters must have binding participation constraints, and they must strictly obey the voting recommendations prescribed to them, i.e., $v(\mathbf{a}, K)+\mu_{RR}<0$, $v(\mathbf{a}, -K)+\mu_{LL}>0$, and $v(\mathbf{a}, -K)+\mu_{LR}<0<v(\mathbf{a}, K)+\mu_{LR}$. The first result implies that $v(\mathbf{a}, k)+\mu_{LL}<0<v(\mathbf{a}, k)+\mu_{RR}>0$ $\forall k$. If, in addition, the signal is symmetric, i.e., $\Pi\left(LL \mid \omega\right)=\Pi \left(RR\mid -\omega\right)$, $\Pi\left(LR\mid \omega\right)=\Pi\left(LR \mid -\omega\right)$, and $\Pi\left(RR \mid \omega\right)=\Pi \left(LL\mid -\omega\right)$ almost surely, then $\mu_{LR}=0$ and $|\mu_{LL}|=\mu_{RR}$ as desired. 

To establish symmetry, consider the signal structure  $\Pi': \Omega \rightarrow \Delta\left(\mathcal{Z}\right)$ whereby $\Pi'\left(LL \mid \omega\right)=\Pi\left(RR\mid -\omega\right)$, $\Pi'\left(LR \mid \omega\right)=\Pi\left(LR \mid -\omega\right)$, and $\Pi'\left(RR \mid \omega\right)=\Pi\left(LL \mid -\omega\right)$. By construction, the following must hold: $\pi'_{LL}=\pi_{RR}$, $\pi'_{LR}=\pi_{LR}$, $\pi'_{RR}=\pi_{LL}$, $\mu'_{LL}=-\mu_{RR}$, $\mu'_{LR}=-\mu_{LR}$, and $\mu'_{RR}=-\mu_{LL}$, which, together with strict obedience, implies that 
 \begin{align*}
v\left({\bf{a}}, -K\right)+\mu'_{LL}& =-v\left({\bf{a}}, K\right)-\mu_{RR}<0,\\
v\left({\bf{a}}, -K\right)+\mu'_{LR}& =-v\left({\bf{a}}, K\right)-\mu_{LR} < 0,\\
v\left({\bf{a}}, -K\right)+\mu'_{RR}& =-v\left({\bf{a}}, K\right)-\mu_{LL} > 0,\\
v\left({\bf{a}}, K\right)+\mu'_{LL}& =-v\left({\bf{a}}, -K\right)-\mu_{RR}<0,\\
v\left({\bf{a}}, K\right)+\mu'_{LR}& =-v\left({\bf{a}}, -K\right)-\mu_{LR} > 0,\\
\text{ and } v\left({\bf{a}}, K\right)+\mu'_{RR}& =-v\left({\bf{a}}, -K\right)-\mu_{LL} > 0. 
\end{align*}
As a result, \begin{align*}
V\left(\Pi'; {\bf{a}}, -K\right)
&=\pi'_{RR}\left[v\left({\bf{a}}, -K\right)+\mu'_{RR}\right]\\
&=\pi_{LL}\left[-v\left({\bf{a}}, K\right)-\mu_{LL}\right] 
=V\left(\Pi; {\bf{a}}, K\right) 
=\lambda I\left(\Pi\right)
=V\left(\Pi; {\bf{a}}, -K\right), 
\end{align*}
and $V\left(\Pi'; {\bf{a}}, K\right)=V\left(\Pi; {\bf{a}}, K\right)$ can be shown analogously. The remainder of the proof is the exact same as that of the previous case. \qed

\begin{spacing}{1}

 \end{spacing}
 
\end{document}